\documentclass[11pt]{article}

\usepackage[T1]{fontenc}

\usepackage[margin=1in]{geometry}
\usepackage{mathpazo}

\usepackage{lipsum} 
\usepackage{bm}
\usepackage[lined,boxed,commentsnumbered,ruled,vlined,linesnumbered,boxed,longend,noend]{algorithm2e}
\usepackage[hang,flushmargin]{footmisc}
\usepackage{amssymb}
\usepackage{tcolorbox}
\tcbuselibrary{skins}
\usepackage{xcolor}
\usepackage{amsthm}
\usepackage{mathtools}
\usepackage{thmtools}
\usepackage{framed}
\usepackage{mdframed}
\usepackage{amsmath}
\usepackage{tabularx}
\usepackage{float}

\usepackage[noend]{algpseudocode}

\usepackage{bbm}

\usepackage[dvipsnames]{xcolor}
\usepackage{xparse}

\usepackage{booktabs}
\usepackage{makecell}
\usepackage{array}
\usepackage{adjustbox}

\usepackage{amsfonts}

\usepackage{enumitem}
\newlist{Properties}{enumerate}{2}
\setlist[Properties]{label=\textbf{Property} \arabic*.,leftmargin=*}

\usepackage{tcolorbox}
\tcbuselibrary{skins}
\usepackage{xcolor}

\usepackage{subcaption}
\colorlet{mix}{red!50!black}

\usepackage{titlesec}
\titleclass{\subsubsubsection}{straight}[\subsection]

\usepackage{threeparttable}

\usepackage{hyperref}
\hypersetup{colorlinks={true},linkcolor={mix},citecolor=blue, breaklinks=true}

\usepackage{nameref}
\usepackage{footnotehyper}

\usepackage{cleveref}
\crefname{subsection}{subsection}{subsections}
\Crefname{subsection}{Subsection}{Subsections}

\newtheorem{theorem}{Theorem}[section]
\newtheorem{corollary}{Corollary}[theorem]
\newtheorem{lemma}[theorem]{Lemma}
\newtheorem*{remark}{Remark}

\theoremstyle{definition}
\newtheorem{definition}{Definition}[section]

\crefname{assumption}{Assumption}{Assumptions}
\Crefname{assumption}{Assumption}{Assumptions}



\newcounter{subsubsubsection}[subsubsection]
\renewcommand\thesubsubsubsection{\thesubsubsection.\arabic{subsubsubsection}}

\titleformat{\subsubsubsection}
  {\normalfont\normalsize\bfseries}{\thesubsubsubsection}{1em}{}
\titlespacing*{\subsubsubsection}
{0pt}{3.25ex plus 1ex minus .2ex}{1.5ex plus .2ex}

\makeatletter
\renewcommand\paragraph{\@startsection{paragraph}{5}{\z@}%
  {3.25ex \@plus1ex \@minus.2ex}%
  {-1em}%
  {\normalfont\normalsize\bfseries}}
\renewcommand\subparagraph{\@startsection{subparagraph}{6}{\parindent}%
  {3.25ex \@plus1ex \@minus .2ex}%
  {-1em}%
  {\normalfont\normalsize\bfseries}}
\def\toclevel@subsubsubsection{4}
\def\toclevel@paragraph{5}
\def\toclevel@subparagraph{6}
\def\l@subsubsubsection{\@dottedtocline{4}{7em}{4em}}
\def\l@paragraph{\@dottedtocline{5}{10em}{5em}}
\def\l@subparagraph{\@dottedtocline{6}{14em}{6em}}
\makeatother

\AtBeginDocument{%
  }

\SetKwComment{Comment}{/* }{ */} 

\DeclarePairedDelimiter\angleb{\langle}{\rangle}

\newcommand{\estimateZ}{\mathsf{EstimateF_0}}

\newcommand{\estimateK}{\mathsf{EstimateF_k}}
\newcommand{\estimateN}{\mathsf{EstimateF_n}}

\newcommand{\estimateSR}{\mathsf{EstimateSR}}
\newcommand{\estimateSLFA}{\mathsf{EstimateSLFA}}
\newcommand{\estimateGB}{\mathsf{EstimateGB}}

\newcommand{\evaluate}{\mathsf{Evaluate}}
\newcommand{\helpevaluate}{\mathsf{HelpEvaluate}}

\newcommand{\throwprob}{\alpha}
\newcommand{\poly}{q_{\mathcal{S}}}

\newcommand{\conf}{\delta}
\newcommand{\error}{\varepsilon}

\newcommand{\freq}[1]{\mathsf{f}_{#1}}

\newcommand{\F}[1]{\mathsf{F}_{#1}}
\newcommand{\stream}{\mathcal{S}}

\newcommand{\SR}{\mathsf{SR}}
\newcommand{\SLFA}{\mathsf{SLFA}}

\newcommand{\streamlen}{m}
\newcommand{\univ}{\Omega}
\newcommand{\f}[1]{\mathsf{f}_{#1}}

\newcommand{\maxfreq}{\tau}

\newcommand{\Ber}{\mathsf{Bernoulli}}

\newcommand{\GBS}{\mathsf{S}}

\NewDocumentCommand{\ddt}{ m o }{%
  \IfNoValueTF{#2}
    {%
      \frac{\mathrm d}{\mathrm d#1}%
    }
    {%
      \frac{\mathrm d^{#2}}{\mathrm d#1^{#2}}%
    }%
}

\newcommand{\expt}{t}
\newcommand{\espoly}{Q_{\mathcal S}}

\makeatletter
\newcommand{\algorithmfootnote}[2][\footnotesize]{%
  \let\old@algocf@finish\@algocf@finish
  \def\@algocf@finish{\old@algocf@finish
    \leavevmode\rlap{\begin{minipage}{\linewidth}
    #1#2
    \end{minipage}}%
  }%
}

\newlist{axioms}{enumerate}{1}
\setlist[axioms]{
  label=\textbf{Axiom A\arabic*}:, 
  leftmargin=*, 
  resume 
}

\newcommand{\appropto}{\mathrel{\vcenter{
  \offinterlineskip\halign{\hfil$##$\cr
    \propto\cr\noalign{\kern2pt}\sim\cr\noalign{\kern-2pt}}}}}

\def\@testdef #1#2#3{%
  \def\reserved@a{#3}\expandafter \ifx \csname #1@#2\endcsname
 \reserved@a  \else
\typeout{^^Jlabel #2 changed:^^J%
\meaning\reserved@a^^J%
\expandafter\meaning\csname #1@#2\endcsname^^J}%
\@tempswatrue \fi}

\title{
Unlocking Fractional Moments in Delphic Set Streams
}
\author{
  Aranya Kumar Bal\thanks{Indian Statistical Institute, Kolkata, India. Email: hayatea90@gmail.com} 
  \and 
  Sourav Chakraborty\thanks{Indian Statistical Institute, Kolkata, India. Email: chakraborty.sourav@gmail.com} 
  \and 
  Arijit Ghosh\thanks{Indian Statistical Institute, Kolkata, India. Email: arijitkgpster@gmail.com}
  \and
  Rudrayan Kundu\thanks{Indian Statistical Institute, Kolkata, India. Email: rudrayan574@gmail.com}
}

\date{} 

\begin{document}

\maketitle

\begin{abstract}
We consider estimation of non-integer frequency moments $\F{k}$ and related Bernstein-type statistics (functions whose derivatives are monotone) in the Delphic set stream model under a mild bounded-frequency assumption: every universe element appears at most $\maxfreq$ times. The main challenge of this model is to keep space low while also keeping update time low, which is not trivial because the sets can be exponential in size compared to their representations. Our core insight is that by sampling the stream at different rates and observing the resulting distinct-counts, we can 'probe' the frequency distribution and numerically integrate these probes to reconstruct a broad class of statistics. Building on that, we crucially observe that the distinct-count of a randomly sampled substream, viewed as a function of the sampling rate, is a single analytic object whose evaluations determine a broad class of statistics via a complementary Laplace–type integral. Algorithmically we exploit this by 
\begin{enumerate}
    \item estimating those evaluations using only standard $\F{0}$ (distinct-count) algorithms on sampled substreams and
    \item recovering target statistics by controlled numerical integration on a judiciously chosen grid.
\end{enumerate}
For $\F{k}$ with $k\in (0,1)$ (and for several other statistics such as Saturated Richness and Smoothed Log-frequency Aggregate) we obtain the \textit{first} one-pass streaming algorithms for Delphic set streams whose space and per-set update time are $\mathrm{poly}(\log|\univ|,\log\streamlen,\error^{-1},\log(1/\conf))$ in the practically relevant regime $\maxfreq=\mathrm{polylog}(|\univ|,\streamlen)$; in general the bounds are polynomial in $\maxfreq$ and $\error^{-1}$ and logarithmic in $\conf^{-1}$. While the general bounds are polynomial in $\maxfreq$, our primary contribution is establishing the feasibility of efficient estimation in this challenging model. The method is a new approach via integral transform and discretization, and unifies estimating a range of statistics under the sampling-plus-integration paradigm.

We also give a complexity-theoretic barrier explaining why lower bounds for removing the bounded-frequency assumption appear difficult: ruling out polylogarithmic algorithms for unrestricted Delphic $\F{k}$ would imply a linear-space threshold-counting separation.
\end{abstract}

    

\thispagestyle{empty}


\newpage
\setcounter{page}{1}
\section{Introduction}
Let us start with a motivating example. Consider a large geographical region monitored by a set of wireless sensors or cameras.
Each sensor $j$ covers an (axis-parallel rectangular) region $B_j \subseteq [0,1]^3$, representing its sensing range.  As sensors are deployed or activated over time, the coverage regions $B_1, B_2, \dots, B_n$ arrive as a stream.
Let $f_i$ denote the number of sensors whose coverage includes a given grid cell $i$. The following measures correspond to characteristics of the sensor network.


\begin{itemize}
    \item $F_0 = |\{ i : f_i > 0 \}|$ corresponds to the total monitored area,
    i.e., the measure of the union of all sensing regions;
    \item $F_1 = \sum_i f_i$ gives the total sensing area counting overlaps;
    \item other moments $\F{k}$ (for $k > 0$) quantify the degree of 
    redundant coverage - regions covered by multiple sensors contribute polynomially more.
\end{itemize}

Estimating $\F{k}$ in a streaming manner thus enables real-time monitoring of coverage and redundancy without storing all sensor locations, which is crucial in large-scale Internet-of-Things (IoT) or environmental surveillance systems.

The above problem is formally modeled as follows:

\begin{definition}[Streaming $\F{k}$ Estimation in the Klee's Measure Setting]
Let $\mathcal{B} = \{B_1, B_2, \dots, B_n\}$ be a collection of axis-parallel boxes in $[0,1]^d$. 
The boxes arrive sequentially as a stream
\[
B_1, B_2, \dots, B_n,
\]
where each $B_j \subseteq [0,1]^d$ is specified by its lower and upper coordinates 
$(\ell_j^{(1)}, u_j^{(1)}), \dots, (\ell_j^{(d)}, u_j^{(d)})$.

We discretize the domain $[0,1]^d$ into a fine grid $\univ = [m]^d$ of resolution $1/m$, 
so that each grid cell (or point) corresponds to a unique index $x \in \univ$.
For each grid cell $x$, define
\[
\freq{x} = \bigl|\{\, j \in [n] : x \in B_j \,\}\bigr|,
\]
that is, $\freq{x}$ is the number of boxes covering cell $x$.
For a fixed integer $k \ge 0$, define the $k$-th frequency moment of the coverage function 
$f : \mathcal{G} \to \mathbb{N}$ as
\[
 \F{k}(\stream):=\sum_{x\in \univ}\freq{x}^k
\]
The \emph{streaming $\F{k}$ estimation problem in the Klee's measure setting} is to maintain 
a compact summary (or sketch) of the stream of boxes $B_1, \dots, B_n$ such that, at the 
end of the stream, one can output an estimate $\widehat{\F{k}}$ satisfying
\[
(1 - \varepsilon) \F{k} \le \widehat{\F{k}} \le (1 + \varepsilon) \F{k},
\]
with probability at least $1-\conf$.
\end{definition}

The goal is to optimize two measures:
\begin{itemize}
    \item \textbf{Space complexity}: that is the amount of space used by the algorithm 
    \item \textbf{Update time complexity}: the worst-case time required to update the memory for any item in the stream, that is, for any set $B_i$.
\end{itemize}


In \cite{CVM}, the notion of axis-parallel boxes was generalized to the notion of Delphic sets, which allow efficient size, membership, and sampling oracles (see \Cref{def:Delphicsetdefinition} for the precise definition). Note that Delphic set streams immediately generalize the notion of axis-parallel box streams and data streams (where the latter is the case when the sets are singletons). So solving the $\F{k}$ estimation problem for Delphic sets immediately solves the $\F{k}$ estimation problem in the Klee's measure setting (and several other natural settings like arithmetic progressions and solutions to DNF formulae).

Now note that the above problem of $\F{k}$ estimation for Delphic set streams (\Cref{def:Fkdefinition}) can be reduced to a problem of data streams by listing out all the grid points in a set $S_i$ as elements in the stream. However, the number of grid points in each set $S_i$ can be exponential in the size of the representation of the set (for example, in the Klee's measure setting, boxes come in the stream as a pair of diagonal points, and the number of points inside the box is exponential in the number of bits representing the box). Which means that listing all the elements in a set $S_i$ would imply the update time becomes $\sim |\Omega|$ in the worst case. 

A big open problem in this area was to design algorithms for estimating $\F{0}$ for the Klee's measure problem using space and update time complexity polynomial in $d$. Recently, it was solved, generalized, and optimized respectively in \cite{CVM, CVMapprox, NandiVGMP024} in the Delphic set setting where (in \cite{NandiVGMP024}) they gave an algorithm for obtaining an $(\error,\conf)$-estimate (see \Cref{def:epsilondelta}) for $\F{0}$, using space and update time complexity $\tilde O(\error^{-2}\log^2|\univ|)$. This directly gives a $(\error, \conf)$-approximation algorithm for the $\F{0}$ Klee's measure problem.

For $k > 2$, the algorithm of \cite{alon1996space} can be adapted to estimate $\F{k}$ of Delphic set streams using space and update time complexity of $O(|\univ|^{1-2/k})$, and this is optimal. However, $0<k<2$, not much is known in the Klee's Measure or the Delphic set setting. The status is the same even when one assumes that the maximum frequency of any element is at most  $\mathrm{poly}(\log|\Omega|,\log m)$. In this work, we consider the case of $k\in (0,1)$ for Delphic set streams.  

\paragraph{}In this paper, we design an efficient algorithm for estimating the $\F{k}$ for a stream of sets for $k\in (0,1)$ \emph{when the frequency is bounded above by $\maxfreq$}. We call such streams $\maxfreq$-frequency bounded set streams (\Cref{def:freqbound}). More concretely, in this paper, we show the following (informally stated) core result. 


\begin{center}
    \emph{Given a stream $\stream$ of $\streamlen$ Delphic sets such that each element of the universe is present in at most $\mathrm{poly}(\log|\Omega|,\log \streamlen)$ sets in the stream, there exists a one-pass streaming algorithm that estimates $\F{k}$ of $\stream$ [within $(1\pm \varepsilon)$ factor and with probability of correctness $(1-\conf$)] using  $\mathrm{poly}(\log|\Omega|,\log \streamlen, 1/\error,\log(1/\conf))$ space and update time per set for $k\in (0,1)$.}
\end{center}

We use analytic techniques to do this, which also allows us to estimate several other statistics of such streams efficiently. In \Cref{sub:ourapproach}, we outline our approach and talk about why porting some of the known techniques for estimating $\F{k}$ when $k\in (0,1)$ (\cite{cohen2017hyperlogloghyperextendedsketches}) or $(0,2]$ ($k$-stable distribution sketches \cite{indyk2006stable}) for data streams do not work in the setting of Delphic set streams.

\paragraph{}At the end, we also complement our algorithmic results with a complexity-theoretic explanation for why proving lower bounds in the unrestricted Delphic model appears to be difficult. We show that a sufficiently strong lower bound ruling out polylogarithmic-space and polylogarithmic-update algorithms for unrestricted Delphic $\F{k}$ would imply a threshold-counting time-space separation. More concretely, if a natural linear-space PP-type class $\mathrm{LinPP}$ were contained in $\mathrm{DTISP}(\mathrm{poly},\mathrm{LINSPACE})$\footnote{For the precise definitions, see \Cref{sec:complexity-barrier}.}, then finite-precision stable sketches could be simulated for Delphic set streams using aggregate queries over the implicit sets. Thus lower bounds for unrestricted Delphic $\F{k}$ would imply $\mathrm{LinPP}\not\subseteq\mathrm{DTISP}(\mathrm{poly},\mathrm{LINSPACE})$. This is an $\F{k}$ analogue of the $\F{0}$ oracle barrier of \cite{NandiVGMP024}, but with nondeterminism replaced by threshold counting.

\subsection{Prior work}


The moment estimation problem in streaming algorithms has been studied mostly in the data stream model, where every update is an element of some universe. The optimal results for $\F{0}$ estimation were given by Flajolet and Martin \cite{flajolet1985probabilistic} and optimized by Alon, Matias, and Szegedy (AMS) in \cite{alon1996space} to give a $\Theta(\error^{-2}+\log|\univ|)$ algorithm. The problem is so crucial that several practical algorithms were developed for $\F{0}$ estimation, achieving the same complexity bound, one of the most famous ones being HyperLogLog \cite{flajolet2007hyperloglog}. Alon, Matias and Szegedy~\cite{alon1996space} showed that it is possible to compute $\F{k}$ for $k>1$ in $\widetilde{O} (\error^{-2}|\univ|^{1-1/k})$. Kane, Nelson, and Woodruff~\cite{KNW2010} then gave algorithms for $\F{k}$ with $k\in (0,2]$ which were optimal for $k\in (0,1)$ up to some factors (in the turnstile model). Cohen \cite{cohen2017hyperlogloghyperextendedsketches} gave an extension of HyperLogLog that gives the optimal space bound of $\Theta(\error^{-2}+\log\log|\Omega|)$ for $k\in (0,1)$ in the random oracle model whereas Jayaram, Rajesh and Woodruff~\cite{JayaramRajeshWoodruff23} gave a $\Theta(\error^{-2}+\log|\Omega|)$ algorithm for $k\in (0,1)$ without the random oracle assumption. Cohen \cite{cohen2017hyperlogloghyperextendedsketches} also gives algorithms with the same complexity for several other statistics.

In the world of sets, the history starts in 1977, when Klee \cite{klee1977can} asked the problem of computing the length of the union of a collection of $n$ intervals in $\mathbb R$. This was solved algorithmically by Klee himself \cite{klee1977can}, and the algorithm was optimal. In 1977, Jon Bentley asked a similar question about the union of $n$ axis-parallel boxes in $\mathbb R^2$, and who also solved it algorithmically, and it was optimal. The natural generalization of this question is to find the volume of the union of $n$ axis-parallel boxes in $\mathbb R^d$. This is the famous \textit{Klee's measure problem}. Progress on this happened slowly (with the most recent development by Chan, giving a simple $\mathcal O(n^{d/2})$ algorithm)\cite{chan2013klee}.

The question was then studied in the streaming setup (where we assume the boxes are coming to us in a stream)\cite{bar2002reductions, braverman2010recursive, indyk2005optimal, pavan2007range, sharma2015efficient, sun2009two}. Pavan and Tirthapura \cite{pavan2007range} gave an optimal algorithm achieving $\tilde O(\error^{-2}\log|\univ|)$ for the case of $d=1$ (intervals). Chakraborty, Meel, and Vinodchandran generalized this notion and introduced the Delphic set streaming model in \cite{CVM} where the stream elements are subsets of the universe. They achieved a $\mathrm{poly}(1/\error, \log(1/\conf),\log\streamlen)$ algorithm to compute the $\F{0}$ of such streams (so solving the Klee's measure problem for boxes as a result); the model was generalized in \cite{CVMapprox} and the results improved in \cite{NandiVGMP024} to $\tilde O(\error^{-2}\log^2|\univ|)$ space and update time for $\F{0}$ for Delphic sets. They left the problem of estimating $\F{k}$ for other $k$ as open problems in the Delphic set stream model. \emph{This is the main motivation of our paper}.


For Delphic sets, almost nothing else was known, even for restricted classes of streams. We elaborate on some of the difficulties in \Cref{subsec:previousnotworking}. In this paper, we look at a restricted class of Delphic set streams, namely $\maxfreq$-frequency bounded streams, where every element occurs in at most $\maxfreq$ sets. Such a restriction is \emph{natural} because often streams have frequency caps, for example, rate limits and physical constraints for sensors. In this setting, we give the first algorithms for $\F{k}$ when $k\in (0,1)$ and several other statistics which have space and update time $\mathrm{poly}(\tau, \log|\univ|,\log m, \error^{-1})$. So when $\tau=\mathrm{poly}(\log|\univ|, \log m)$, we get the first ever $\mathrm{poly}(\log|\univ|,\log m, \error^{-1})$ algorithm for $\F{k}.$ We also show (in \Cref{sec:Fkamsanalytic}) that a slight modification of the AMS algorithm for $k>1$ also works for Delphic sets in this setting. As an extension of our techniques, we show in \Cref{sec:FkN} that we can also get $\F{k}$ for when $k\in\mathbb N$ with slightly worse bounds than the AMS modification.

In \Cref{tab:prior_vs_our_work}, we compare known results for space complexity with ours, where the Delphic set stream column is for polylog update times. We hide all factors of $\log\maxfreq$, $\log(1/\error)$, $\log(1/\conf)$, $\log(R_i)$, $\log(U)$ and all $\log\log$ factors to make the bounds simpler looking.

\begin{table}[htbp]
\centering
\caption{Prior work in comparison to our work. Space complexity (in bits) for $(1\pm\error)$-approximation with constant probability. We hide all factors of $\log\maxfreq$, $\log(1/\error)$, $\log(R_i)$, $\log(U)$ and all $\log\log$ factors. $^*$ marked complexities are in the random oracle model. The rightmost column shows both space and per-set update time for $\maxfreq$-frequency-bounded Delphic set streams.}
\label{tab:prior_vs_our_work}
\setlength{\extrarowheight}{6pt} 
\renewcommand{\arraystretch}{1}
\small
\begin{adjustbox}{max width=1.5\textwidth}
\begin{tabularx}{\textwidth}{|
  >{\centering\arraybackslash}m{0.30\textwidth}
  |>{\centering\arraybackslash}m{0.25\textwidth}
  |>{\centering\arraybackslash}m{0.37\textwidth}|
}
\hline
\makecell[c]{Stream\\Statistic}
  & \makecell[c]{Data\\streams}
  & \makecell[c]{$\maxfreq$-frequency\\bounded Delphic set\\streams} \\
\hline
\makecell[c]{\\ \mbox{}$F_0$\\ \\} 
  & \makecell[c]{$\Theta(\varepsilon^{-2}+\log |\univ|)$} 
  & \makecell[c]{$\tilde O(\varepsilon^{-2}\log^2|\univ|)$~\cite{NandiVGMP024}} \\
\hline
\makecell[c]{$\F{k},\; k\in(0,1)$} 
  & \makecell[c]{$\Theta(\varepsilon^{-2}+\log\log|\univ|)^*$\\ ~\cite{cohen2017hyperlogloghyperextendedsketches}\\ $\Theta(\varepsilon^{-2}+\log|\univ|)$~\cite{JayaramRajeshWoodruff23}} 
  & \makecell[c]{$\tilde O_k\left(\tau^{\frac{5+k}{2}}{\varepsilon^{-4-\frac{3(k+1)}{2k(1-k)}}\log^5|\univ|}\right)$\\~[Our work]} \\
\hline
\makecell[c]{\\ \mbox{}$F_1$\\ \\} 
  & \makecell[c]{$\tilde\Theta(1)$ [Folklore]} 
  & \makecell[c]{\(\tilde\Theta(1)\) [Folklore]} \\
\hline
\makecell[c]{\\ \mbox{}$\F{k},\; k\in (1,2)$\\ \\} 
  & \makecell[c]{\\ \mbox{}$\Theta(\varepsilon^{-2}(\log|\univ|+\log m))$\\ ~\cite{KNW2010}\\ \\} 
  & \makecell[c]{$\tilde \Theta_k(\maxfreq^{k-1}\error^{-2})$~\cite{alon1996space}} \\
\hline
\makecell[c]{\\ \mbox{}$F_2$\\ \\} 
  & \makecell[c]{\\ \mbox{}$\Theta(\varepsilon^{-2}(\log|\univ|+\log m))$\\ ~\cite{alon1996space}\\ \\} 
  & \makecell[c]{$\tilde O(\error^{-2}\sqrt{|\Omega|}$)~\cite{alon1996space}} \\
\hline
\makecell[c]{\\ \mbox{}$\F{k},\; k> 2$\\ \\} 
  & \makecell[c]{$\tilde \Theta_k(\error^{-2}|\univ|^{1-2/k})$~\cite{alon1996space}} 
  & \makecell[c]{$\tilde \Theta_k(\maxfreq^{k-1}\error^{-2})$~\cite{alon1996space}} \\
\hline
\makecell[c]{\\ \mbox{}$\SR(r)$\\ \\} 
  & \makecell[c]{$O(\varepsilon^{-2}+\log\log|\univ|)^*$\\ ~\cite{cohen2017hyperlogloghyperextendedsketches}} 
  & \makecell[c]{$\tilde O\left(\maxfreq^2\varepsilon^{-5}\log m\log^5|\univ|\right)$ \\~[Our work]} \\
\hline
\makecell[c]{\\ \mbox{}$\SLFA$\\ \\} 
  & \makecell[c]{$O(\varepsilon^{-2}+\log\log|\univ|)^*$\\ ~\cite{cohen2017hyperlogloghyperextendedsketches}} 
  & \makecell[c]{$\tilde O\left(\maxfreq^{3}\varepsilon^{-5}\log m\log^5|\univ|\right)$ \\~[Our work]} \\
\hline
\makecell[c]{\\ \mbox{}General Bernstein function\\with non-negative L\'evy density\\ \\} 
  & \makecell[c]{$O(\varepsilon^{-2}+\log\log|\univ|)^*$\\ ~\cite{cohen2017hyperlogloghyperextendedsketches}} 
  & \makecell[c]{$\tilde O(U^{3/2}\varepsilon^{-9/2}(\maxfreq\sqrt{R_0}+\sqrt{\maxfreq R_1}+$\\$\sqrt{R_2})\log m\log(1/1-e^{-L})\log^5|\univ|$)\\~[Our work]} \\
\hline
\end{tabularx}
\end{adjustbox}
\end{table}

\subsection{Why previous approaches do not work}\label{subsec:previousnotworking}


The techniques of data streams can indeed be directly ported to the Delphic streams case if one is not worried about update time per set; one can just enumerate all the elements of the set and then use the data stream algorithms. But the challenge appears because this is not allowed in this model, since the subsets can be of size $\sim |\Omega|$, and that would cause the update time to not be polylogarithmic in the parameters. A very concrete example is axis-parallel boxes represented as a pair of diagonal points; the size of the representation is $\sim \log|\univ|$, but the box itself can contain $\sim |\univ|$ number of points in it. The breakthrough of $\F{0}$ estimation in this model was possible because of the brilliant observation that the hash function-based data stream $\F{0}$ algorithms could be simulated (in some sense) using just sampling every element of the stream by a fixed probability. Since Delphic sets offer a sampling primitive, the algorithm works. In fact, the belief is that any purely sampling-based algorithms (sampling from the stream) can most likely be ported to Delphic set streams.

The current known data stream algorithms for $\F{k}$ for $k\in (0,1)$ require one to touch every \emph{element} of the stream (see \cite{cohen2017hyperlogloghyperextendedsketches, indyk2006stable} etc), be it estimating based on hash functions or sketching based on sampling from distributions. As mentioned, this violates the update time constraint in the Delphic set case. In particular, one can ask if the techniques of \cite{indyk2006stable} can be generalized to Delphic set streams; unfortunately, they cannot. The $k$-stable distribution approach of Indyk for estimating $\F{k}$ of data streams when $k\in (0,2]$ requires one to compute a linear sketch of \emph{all elements} of the stream, which is a per-element operation, and hence cannot be computed without enumerating the sets (or assuming powerful oracles).


While recent work by by Pettie and Wang~\cite{pettie_wang_levy} utilized the L\'evy-Khintchine representation for data streams by hashing the indices to appropriate L\'evy processes, applying this analytic framework to Delphic sets is fundamentally obstructed by the inability to perform per-element updates (similar to the problems as mentioned above). Our approach bypasses this by reducing to $\F{0}$ estimation of a subsampled stream and then using polynomials and estimating via integral transforms with similar L\'evy-Khintchine representations. One of our key innovations lie in the sampling process detailed in \Cref{sec:preliminaries} where we show how to simulate Delphic oracles on subsampled sets without manifesting the subsampled set (which can take a lot of time and space).

Interestingly, in \Cref{sec:complexity-barrier}, we formalize exactly what oracle would be needed to make the approach of \cite{indyk2006stable} work for $\F{k}$ estimation. A (finite-precision) stable sketch maintains counters of the form $\sum_x \f{x}w(x)$. Therefore, when a Delphic set $S_i$ arrives, the corresponding update would require the aggregate $\sum_{x\in S_i}w(x)$ over an implicit set. We show that if such threshold aggregates could be simulated deterministically in polynomial time and linear space, then unrestricted Delphic $\F{k}$ estimation would admit polylogarithmic-space and polylogarithmic-update algorithms. Consequently, lower bounds ruling this out imply a threshold-counting time-space separation. We give some more details in \Cref{subsec:whyremovinghard}.


\subsection{Integration, measures and stream-statistics estimation : a new polynomial-based approach}




In this work, we propose a sampling-based algorithm for estimating several different statistics of streams over $\maxfreq$-Frequency-Bounded Delphic set streams (see \Cref{def:freqbound}), which uses only $\text{poly}(\maxfreq, \log\streamlen, \log|\univ|, \error^{-2})$ space. We use an analytic approach towards solving the problem, which we outline below.

\subsubsection{Our approach}\label{sub:ourapproach}

There is a main core idea and a couple of surprising observations that allow our algorithm to work. The core idea is the fact that using a certain sampling technique, one can reduce estimating these statistics to estimating $\F{0}$ over the stream, followed by application of real analytic techniques to obtain the final estimate. The observation is that the object the sampling technique reveals is a much more general object that is related to deep real analytic objects, which allows us to reduce $\F{k}$ estimation for $k\in (0,1)$ and several other stream statistics to just $\F{0}$ estimations.


The main sampling idea is simple: perform \emph{element-level subsampling} on the stream by independently retaining each element $x\in S_i$ in each set $S_i$ with probability $1-\throwprob$. Now we consider the $\F{0}$ of the resulting thinned substream. One easily checks that this quantity turns out to be $\sum_{x\in\univ}(1-\alpha^{\f{x}})$. More formally, we have the following.


\paragraph{\textsc{Expected Support Polynomial}}: Let $\stream = \langle S_1, \dots, S_m\rangle$ be a stream of sets, with each set a subset of the universe $\univ$. For each set $S_i$ and each element $x\in S_i$, independently retain $x$ in $S_i$ with probability $1-\throwprob$ (and discard it with probability $\throwprob$). The resulting thinned sub-stream is denoted $\stream_{\throwprob}$. Let $\mathsf Y$ be the $\F{0}$ of the sub-stream $\stream_{\throwprob}$. Note that $\mathsf Y$ is a random variable and the expected value of $\mathsf Y$ depends on $\throwprob$.


\begin{lemma} \label{lem:expectedpoly}
If $\mathsf Y$ be the $\F{0}$ moment of the sub-stream $\stream_{\throwprob}$ that is obtained by element-level subsampling: for each set $S_i$ in the stream $\stream$ and each element $x\in S_i$, independently retaining $x$ in $S_i$ with probability $(1-\throwprob)$, then 
    \begin{equation*}
    \mathbb{E}[\mathsf Y] = \sum_{x\in \univ}(1-\throwprob^{\freq{x}})
    \end{equation*}
    where $\f{x}(\stream)$ is the frequency of $x$ in the stream $\stream$.
\end{lemma}

The proof of the lemma is in \Cref{subsec:expectedpoly}.

\paragraph{} Now since the sets can be large and constructing the sampled set might take a lot of time and space, we will never actually do the sampling to generate the sampled set. Since the only access required of the sampled sets are queries about membership, size, and uniform sampling, we \emph{simulate} the oracles from the set itself, without ever manifesting the subsampled set. A key part of this involves looking at the structure of the $\F{0}$-estimator from \cite{NandiVGMP024}. The observations that make the simulation work is noting that membership queries are (in a way) just checking membership in the original set and tossing a coin, the size of a subsampled set follows a binomial distribution, and sampling can be done via rejection sampling appropriately (there is more nuance to the idea, but this is the core philosophy). This is constructed in detail in \Cref{subsec:simulation}. This also guarantees that the observation of \Cref{lem:expectedpoly} still holds.

\paragraph{} We now note that $\poly(\throwprob):=\sum_{x\in\univ}(1-\throwprob^{\f{x}})$ is a polynomial in $\throwprob$ (since $\f{x}$ are non-negative integers) and that $\poly(1)=0$ and $\poly(0)=\F{0}(S)$. Now noting that $\throwprob\in [0,1]$ has a natural change of variables $\throwprob\mapsto e^{-\expt}$ with $\expt\in (0,\infty)$, we can define a new function

\begin{equation*}
    \espoly(t):=\sum_{x\in\F{0}}(1-e^{-\expt\f{x}})
\end{equation*}

We call $\espoly$ the \emph{Expected Support Polynomial}.

\paragraph{\textsc{Estimating $\F{k}$ when $k\in (0,1)$}}: The first observation is that if one now defines the measure $\mu=\sum_{x\in\Omega}\delta_{\f{x}}$ over the non-negative real line\footnote{A Dirac delta measure on a set $\mathcal X$ is designated as $\delta_X$ and defined for a given $x\in \mathcal X$ and any measurable subset $\mathcal A\subseteq \mathcal X$ by $\delta_x(\mathcal A)=1$ if and only if $x\in\mathcal A$.}, that is, a sum of Dirac delta measures at each frequency that occurs in the stream, then we can write

\begin{equation*}
    \espoly(\expt)=\poly(e^{-\expt})=\int_0^\infty (1-e^{-\expt\lambda})\ \mathrm d\mu(\lambda)
\end{equation*}
which is the \textit{complementary Laplace transform} (\Cref{def:complementarylaplace}) of the measure $\mu$!

\paragraph{} The next observation is that one can write the frequency moments $\F{k}$ as moments of this measure $\mu$, that is

\begin{equation*}
    \F{k}=\int_0^\infty \lambda^k\ \mathrm d\mu(\lambda)
\end{equation*}
Now depending on $k$, the function $\phi(\lambda)=\lambda^k$ behaves differently. This motivates us to look at what these functions are and if we can find a more unified approach for other statistics that are of the same kind. We note that when $k\in (0,1)$, the function is exactly a \emph{Bernstein function} (see \Cref{subsec:BernsteinfunctionLKrepprelim}) and hence admits an integral representation against $\espoly$ (by virtue of \Cref{thm:bernsteinstheorem}).

More particularly, we have the following identity,

\begin{equation*}
    \F{k}=\frac{k}{\Gamma(1-k)}\int_0^\infty \espoly(t)t^{-k-1}\ \mathrm dt
\end{equation*}
which follows from the L\'evy-Khintchine representation of the Bernstein function $\varphi(\lambda)=\lambda^k$ for $k\in (0,1)$ (for an elementary proof of this identity, see \Cref{lem:Fkidentityintegral}, and for the Laplace transform connection, see \Cref{lem:bernsteink}). Our approach is then to estimate $\F{k}$ via estimating point estimates of $\espoly$ by reducing them to $\F{0}$ estimates (see \Cref{sec:evaluateespoly}) and then using numerical integration (specifically the trapezoidal rule, see \Cref{subsec:trapezoidalrule}) to approximate the integral. A high level overview of the technicalities is provided in \Cref{subsec:highleveloverview}. This gives us the first ever method for computing $\F{k}$ for $k\in (0,1)$ for Delphic set streams with $\text{poly}(\maxfreq, \log\streamlen, \log|\univ|)$ update time per set and space complexity, see \Cref{thm:prelimAlgorithmEstimateKcorrectnessandspace}. In the regime when $\maxfreq=\mathrm{poly}(\log|\Omega|,\log m)$, we thus achieve $\mathrm{poly}(\log|\Omega|,\log m, 1/\error,\log(1/\conf))$.

\paragraph{\textsc{Estimating other statistics}}: Several other stream statistics that are Bernstein functions can similarly be written as integrals against $\espoly$, and hence using similar techniques to estimate the integrals, we estimate the statistics themselves. We provide two examples of them, namely \emph{Saturated Richness} and \emph{Smoothed Log-Frequency Aggregate} (see \Cref{def:SRdefinition,def:SLFAdefinition}), the results are mentioned in \Cref{thm:prelimestimateSRcorrectness} and \Cref{thm:prelimestimateSLFAcorrectness} respectively. In fact, the technique generalizes for general Bernstein functions whose L\'evy densities are explicitly known and satisfy mild regularity conditions. The core identity used for all of these is that the statistics look like $\sum_{x}\varphi(\f{x})$ for some Bernstein function $\varphi$ with the representation $\varphi=\int_0^\infty (1-e^{-ts})w(t)\ \mathrm dt$ for different densities $w(t)$. We provide a table, namely \Cref{tab:bernsteinfunctionsanddensities}, for some examples of Bernstein functions $\varphi$ and their densities $w(t)$. Several other examples can be found in \cite{schilling2012bernstein}. Using similar techniques, we get the result stated in \Cref{thm:prelimestimateGBcorrectness}.

\begin{table}[htb]
\centering

\newcolumntype{C}{>{\centering\arraybackslash}X}

\renewcommand{\arraystretch}{1.8}

\begin{tabularx}{\textwidth}{|C|C|}
\hline
\textbf{Bernstein Function $\phi(x)$} & \textbf{Density $w(t)$} \\ 
\hline
$x^k, k\in (0,1)$ & $\frac{k}{\Gamma(1-k)}t^{-k-1}$ \\ 
\hline
$\frac{x}{r+x}, r > 0$ & $re^{-rt}$ \\ 
\hline
$\ln(1+x)$ & $e^{-t}/t$ \\ 
\hline
$(x+1)^k - 1$ & $\frac{ke^{-t}}{\Gamma(1-k)t^{k+1}}$ \\ 
\hline
$\ln\left(\frac{\Gamma(x+a+b)\Gamma(a)}{\Gamma(x+a)\Gamma(a+b)}\right), a,b>0$ & $\frac{e^{-at}(1-e^{-bt})}{t(1-e^{-t})}$ \\
\hline
$-xe^{ax}\mathrm{Ei}(-ax),\ a>0,\ \mathrm{Ei}(x)=\int_{-\infty}^x e^{-t}/t\ \mathrm dt$ & $1/(a+t)^2$ \\ 
\hline
\end{tabularx}

\renewcommand{\arraystretch}{1}

\caption{Examples of Bernstein functions and their corresponding densities.}
\label{tab:bernsteinfunctionsanddensities}
\end{table}

\paragraph{\textsc{Estimating $\F{k}$ when $k\in \mathbb N$}}:  We also observe that 
instead of integrating against our Expected Support Polynomial $\espoly$, if we take the $k$th derivative of $\espoly(t)$ with respect to $t$ then we obtain 
\begin{equation*}
    \F{k}=(-1)^{k-1}\espoly^{(k)}(0)
\end{equation*}
Hence, estimating $\F{k}$ (for $k \in \mathbb{N}$) boils down to estimating the value of $\espoly^{(k)}(0)$. Using, standard analytical techniques (a forward finite difference scheme \cite{atkinson2008introduction}) we can design an algorithm that estimates $\F{k}$ with $k\in\mathbb N$ for $\maxfreq$-Frequency-Bounded streams of Delphic sets with $\tilde O_k(\error^{-4k+1}\maxfreq^{4k-3})$ update time per set and space complexity.
The complexity is slightly worse than the AMS modification we provide in \Cref{sec:Fkamsanalytic} which achieves a complexity of $\tilde O_k(\error^{-2}\maxfreq^{k-1})$. However, for completeness, we provide a detailed analysis of our algorithm in \Cref{subsec:appendixFkintegeranalytic}. Thus we present a new algorithm for estimating $\F{k}$ in a $\maxfreq$-Frequency-Bounded streams (for $k\in \mathbb{N}$). This also shows the power of our technique of using analytical tools with the Expected Support Polynomial.

\subsection{High level overview of the technicalities}\label{subsec:highleveloverview}

Let us go into some more details on $\F{k}$ estimation in this subsection (the other statistics follow a very similar line of reasoning). In \Cref{sub:ourapproach}, we have established that

\[\F{k}=\frac{k}{\Gamma(1-k)}\int_0^\infty\espoly(t)t^{-k-1}\ \mathrm dt\]

If we are to compute $\F{k}$ using this identity, we are immediately faced with several challenges.

\begin{itemize}
    \item The first challenge is that we cannot directly evaluate the integral since we can only evaluate $\espoly$ at points.
    \item The second challenge is the infinite domain of integration and a (integrable) singularity at $0$.
    \item Third is that we can only evaluate $\espoly$ approximately by subsampling the incoming sets so that every element is picked independently.
    \item And finally, the problem of subsampling itself; since we are picking every element independently, we don't have enough time to subsample or enough storage to store the subsampled sets.
\end{itemize}

To solve the first and second challenge, we \emph{approximate} the integral as a Riemann sum (we will use the trapezoidal rule) since we can only evaluate $\espoly$ at points, by \emph{discretizing the integration region}. Since we also cannot do infinite sums, we need to truncate the integral at some finite point $U$; moreover, the integrand (as we shall see later in the calculations) has an integrable singularity at the origin, so we need to truncate the integral somewhere off the origin, say $L$. So we have, with equispaced (spacing $\Delta t=(U-L)/N$) evaluation points from $L$ to $U$

\[\F{k}=\frac{k}{\Gamma(1-k)}\int_0^\infty\espoly(t)t^{-k-1}\ \mathrm dt\approx \frac{k}{\Gamma(1-k)}\left[\frac12(\espoly(L)+\espoly(U))+\sum_{t=L+\Delta t}^{U-\Delta t}\espoly(t)\Delta t\right]\]

Note that this already introduces three sources of error: two from truncating the integrals on both sides and one from approximating the integral as a sum. The errors for truncation can be handled by analyzing the integrand, and the approximation error can be handled by the error of the trapezoidal rule. The third challenge can be solved by replacing $\espoly$ by estimates $\widehat{\espoly}$ and introducing another source of error, so now we have:

\begin{equation}\label{eqn:approximateFk}
    \F{k}\approx \frac{k}{\Gamma(1-k)}\left[\frac12(\widehat{\espoly}(L)+\widehat{\espoly}(U))+\sum_{t=L+\Delta t}^{U-\Delta t}\widehat\espoly(t)\right]\Delta t
\end{equation}

The error for cutting off the origin from the sum is one place where the bounded frequency $\maxfreq$ assumption comes in handy to bound the error; a second place where it is used at times is discretizing the integral, which has to be analyzed via taylor series, and hence some derivatives have to be bounded and $\maxfreq$ plays a role there (We talk more about this assumption on $\tau$ in \Cref{subsec:whyremovinghard}).

Finally, we divy up the total error into all these errors and then use a union bound.

So at a very very high level, the algorithm works as follows. Fix $U, L, N$ before the stream starts. Then when a set in the stream comes, do the subsampling for every set, estimate $\espoly$ parallely for every evaluation point (which are the probabilities for subsampling) and evaluate \Cref{eqn:approximateFk} as a running sum. At the end of the stream, output the value.

Schematically, the high level $\F{k}$ algorithm can be written as (assuming no error in the $\espoly$ evaluations:

\SetAlgoNoLine%
\begin{algorithm}[htb]
    \DontPrintSemicolon%
    \caption{High level algorithm for $\F{k}$ assuming $\espoly$ evaluations have no error}\label{alg:highlevelFk}
    
    \SetKwInOut{Input}{Input}
    \SetKwInOut{Output}{Output}
    
    \Input{stream $\stream$, moment order $k\in(0,1)$\;}
    \Output{(Approximation of the $k$-th frequency moment $\F{k}$ of $\stream$\;}
    
    $L\gets \text{Lower truncation of integral}$\;
    $U\gets \text{Upper truncation of integral}$\;
    $N \gets \text{Number of grid points}$\;
    $\Delta t=\dfrac{U-L}{N}$\;
    $c\gets 0$\;
    \BlankLine
    $\mathrm{bound}=\espoly(U)U^{-k-1}+\espoly(L)L^{-k-1}$\\
    
    \BlankLine
    
    \For{$j \gets 1$ \KwTo $N-1$}{
        $c_j\gets \espoly(L+j\cdot\Delta t)$\;
    }
    
    \BlankLine
    
    $c\gets \sum_{j=1}^{N-1}c_j$\;
    $I^{[k]}=\Delta t\cdot\left(\frac{\mathrm{bound}}2+c\right)$\;
    \BlankLine
    \Return $\frac{kI^{[k]}}{\Gamma(1-k)}$\;
    
    \end{algorithm}

The final challenge is that of subsampling itself. Since $\espoly$ at any point is essentially an expected value of a subsampled stream, na\"ively it seems that we need to actually do the subsampling, that is, when a set comes we have to sample every element from it with some probability. But this takes time about the size of the set, which can be as large as the universe (also storage of the subsampled sets is a concern).

We devise a way to handle this by understanding the structure of the $\F{0}$-estimator that we will be using (from \cite{NandiVGMP024}), and exploiting this structure to \emph{simulate} the subsampled sets. The crux of the simulation is basically that when a set comes, the $\F{0}$-estimator only asks \emph{membership queries, size queries,} and \emph{sampling queries} to the set, and has no other way to access it. The key fact we note is that the $\F{0}$-estimator from \cite{NandiVGMP024} does the following \emph{in order} and outline how we simulate them:

\begin{enumerate}
    \item the \emph{membership queries first} and \emph{asks only one membership query per universe element}, we can essentially simulate this by checking membership in the mother set, flipping a coin with the appropriate head probability, answering, and then freezing the answers for those elements.
    \item \emph{Next, it asks for a size query}, which we can answer by exploiting the fact that the size of a subsampled set follows a binomial distribution, and we can answer this by sampling from a binomial distribution after we take into account the frozen \emph{yes} answers (here we will use a decomposition property of binomials, see \Cref{lem:binomialdecompositionappendix}).
    \item \emph{Finally, it asks for uniform samples}; here we do a rejection sampling with appropriate probabilities across the sets stored and the mother set to return a sample as if we are returning a uniform sample from a subsampled set.
\end{enumerate}

Because everything happens in this particular order, we show that the $\F{0}$-algorithm is oblivious to the simulation. The details are in \Cref{subsec:simulation}.

\subsection{Why the bounded-frequency assumption is not easily removable}\label{subsec:whyremovinghard}

Complementary to the algorithm, we look at the assumption of $\maxfreq$-frequency boundedness. The dependence on $\maxfreq$ in our algorithm comes from the analytic reconstruction step. For example, when $0<k<1$, the identity underlying our estimator is
\[\F{k}=\frac{k}{\Gamma(1-k)}\int_0^\infty \espoly(\expt)\expt^{-k-1}\,d\expt,\]
and controlling the contribution near $\expt=0$ uses the inequality $\F{1}\le \maxfreq^{1-k}\F{k}$. Thus, our method genuinely uses a frequency cap to truncate the integral at a scale that can be handled with polylogarithmic resources.

However, this does not mean that the assumption is necessary for every possible method. But in fact, proving such a necessity appears to be hard. The standard data-stream algorithms for $\F{k}$, such as stable sketches \cite{indyk2006stable}, fail in the Delphic model only because they need per-element linear updates (as mentioned in \Cref{subsec:previousnotworking}). If one had an oracle that could compute the aggregate contribution of an implicit set to each linear sketch coordinate, then the stable-sketch algorithm would lift immediately to Delphic set streams. In \Cref{sec:complexity-barrier}, we formalize this observation and show that a lower bound ruling out such an unrestricted algorithm would imply the separation $\mathrm{LinPP}\not\subseteq\mathrm{DTISP}(\mathrm{poly},\mathrm{LINSPACE})$.

Thus, our $\maxfreq$-bounded results should be viewed as algorithmic feasibility results under a natural structural restriction, while unconditional lower bounds for the unrestricted model appear to require major progress on threshold-counting time-space lower bounds.

\subsubsection{Our results}

Combining the above ideas and techniques, we have the following results\footnote{For definitions of the statistics and terms, see \Cref{sec:preliminaries}. For definition of the complexity classes, see \Cref{sec:complexity-barrier}.}:


\begin{theorem}\label{thm:prelimAlgorithmEstimateKcorrectnessandspace}
    Given a $\maxfreq$-Frequency-Bounded Delphic set stream
    $\stream$, there exists a one-pass streaming algorithm that outputs an $(\error,\conf)$-approximation of $\F{k}(\stream)$ when $k\in(0,1)$ and takes 
    
    \begin{equation*}
        \tilde O_k\left(\tau^{\frac{5+k}{2}}{\varepsilon^{-4-\frac{3(k+1)}{2k(1-k)}}}\log(1/\conf)[\log|\univ|+\log\maxfreq+\log(1/\conf)]\log^4|\univ|\right)
    \end{equation*}
    space and update time per set.    
\end{theorem}

\vspace{1em}

As examples of other statistics that we can compute, we show the following.

\begin{theorem}\label{thm:prelimestimateSRcorrectness}
    Given a $\maxfreq$-Frequency-Bounded Delphic set stream $\stream$, there exists a one-pass streaming algorithm that outputs an $(\error,\conf)$-approximation of $\SR(r)$ of $\stream$ and takes
    
    \begin{equation*}
        \tilde O\left(\maxfreq^2\error^{-5}\ln^{3/2}(1/\error)(\log |\univ|+\log\maxfreq+\log(1/\error)+\log(1/\conf))\log^4|\univ|\right)
    \end{equation*}
    space and update time.
\end{theorem}

\vspace{1em}

\begin{theorem}\label{thm:prelimestimateSLFAcorrectness}
    Given a $\maxfreq$-Frequency-Bounded Delphic set stream $\stream$, there exists a one-pass streaming algorithm that outputs an $(\error,\conf)$-approximation of $\SLFA$ of $\stream$ and takes
    
    \begin{equation*}
        \tilde O\left(\maxfreq^3\error^{-5}\ln^{3/2}(1/\error)(\log |\univ|+\log\maxfreq+\log(1/\error)+\log(1/\conf))\log^4|\univ|\right)
    \end{equation*}
    space and update time.
\end{theorem}

\begin{theorem}\label{thm:prelimestimateGBcorrectness}
    Let $\stream$ be a $\maxfreq$-Frequency-Bounded Delphic set stream. Let $\varphi$ be a Bernstein function with L\'evy-Khintchine representation

    \begin{equation*}
        \varphi(s)=\int_0^\infty (1-e^{-ts})w(t)\ \mathrm dt
    \end{equation*}
    where $w(s)\ge 0$ is the L\'evy density of $\varphi$. Let $I_{\mathrm{GB}}=\sum_{x\in\F{0}}\varphi(\f{x})$ be the statistic to be estimated. Then

    \begin{equation*}
        I_{\mathrm{GB}}=\int_0^\infty \espoly(t)w(t)\ \mathrm dt
    \end{equation*}
    Moreover, if $S_1(y):=\int_0^y sw(s)\ \mathrm ds<\infty$ for all $y>0$, and $r$ is $C^2$ away from $0$ then there exists a one-pass streaming algorithm that outputs an $(\error,\conf)$-approximation of $I$ and takes

    \begin{equation*}
    \tilde O\left(\sqrt{\frac{U^{3}}{\error^{9}}}\left(\maxfreq\sqrt{R_0}+\sqrt{\maxfreq R_1}+\sqrt{R_2}\right)\cdot 
    \log\left(|\univ|\cdot U\cdot \maxfreq\cdot R_0\cdot R_1\cdot R_2\cdot \frac1{\error}\cdot  \frac1{\conf} \cdot (1-e^{-L})\right)     \cdot 
    (1-e^{-L})\log^4|\univ|\right)
    \end{equation*}
    space and update time per set, where $L$ and $U$ are such that $S_1(L)\le\error\varphi(1)/5\maxfreq$ and $\int_U^\infty w(t)\ \mathrm dt\le\error\varphi(1)/5$ respectively and $R_j=\sup_{[L,U]}|r^{(j)}(t)|$ for $j\in \{0,1,2\}$.

\end{theorem}

Finally, we complement this with a complexity theoretic barrier result.

\begin{theorem}[Informal complexity-theoretic barrier]
\label{thm:intro-complexity-barrier}
Fix $0<k\le 2$. Suppose one proves that unrestricted Delphic $\F{k}$ estimation cannot be solved by any one-pass randomized streaming algorithm with space and per-set update time polynomial in $\log|\univ|$, $\log \streamlen$, $\error^{-1}$, and $\log(1/\conf)$. Then
\[\mathrm{LinPP}\not\subseteq\mathrm{DTISP}(\mathrm{poly},\mathrm{LINSPACE}),\]
where $\mathrm{LinPP}$ is the polynomial-time, linear-space analogue of a PP or gap-counting threshold class.

Equivalently, if $\mathrm{LinPP}\subseteq\mathrm{DTISP}(\mathrm{poly},\mathrm{LINSPACE})$, then unrestricted Delphic $\F{k}$ estimation admits a one-pass randomized algorithm with polylogarithmic space and per-set update time.
\end{theorem}

\subsection{Organization of the Paper}


The remainder of the paper is organized as follows. \Cref{sec:preliminaries} fixes notation, defines the Delphic set stream model and the $\maxfreq$-bounded frequency assumption, and recalls the complementary Laplace transform and numerical-quadrature facts used throughout. \Cref{sec:evaluateespoly} gives the primitive $\evaluate$ routine for estimating the Expected Support Polynomial $\espoly$ at any given point via sampled substreams and black-box $\F{0}$ algorithms. \Cref{sec:Fkforkinzeroone} applies $\evaluate$ to construct a one-pass streaming estimator for $\F{k}$ with $k\in (0,1)$. \Cref{sec:otherstatistics} demonstrates similar estimators for several other Bernstein-type statistics (including \emph{Saturated Richness} and \emph{Smoothed Log-Frequency Aggregate}). \Cref{sec:generalbernsteinestimationsection} places the approach in the general Bernstein/L\'evy framework, discusses when the L\'evy density is explicit versus when numerical inversion is required, and addresses stability considerations. Finally, \Cref{sec:complexity-barrier} gives the oracle simulation and threshold-counting barrier for lower bounds in the unrestricted Delphic model.

\section{Preliminaries}\label{sec:preliminaries}

We will use $[n]$ to denote the set $\{1, \dots, n\}$. For any event $\mathsf{E}$ we will use $\mathbb{P}(\mathsf{E})$ to denote the probability of the event $\mathsf{E}$ and we will use $\mathbbm{1}_{\mathsf{E}}$ to denote the Boolean valued random variable that takes value $1$ if event $\mathsf{E}$ happens and $0$ otherwise. For any random variable $\mathsf{X}$ we will use $\mathbb{E}[\mathsf{X}]$ to denote the expected value of the random variable $\mathsf{X}$.

We use the Gamma function, defined as $\Gamma(z)=\int_0^\infty t^{z-1}e^{-t}\ \mathrm dt$ for $\Re(z)>0$ and meromorphically extend it to the left (except to the negative even integers and 0) if required. We will also need the digamma function $\psi(z)=\Gamma'(z)/\Gamma(z)$; in particular, we will use that for $n\in\mathbb N$, it satisfies $\psi(n)=H_{n-1}-\gamma$ where $H_k=\sum_{j=1}^k (1/j)$ is the $k$-th harmonic sum and $\gamma$ is the Euler-Mascheroni constant (which equals $\lim_{k\to\infty} (H_k-\ln k)$).




We will be working with (Delphic) \emph{set streams} in this paper. A set stream of items coming from a universe $\univ$ will be denoted by $\stream = \langle S_1, \dots, S_{\streamlen}\rangle$, where $S_j \subseteq \Omega$. The stream length will refer to the number of items in the stream, that is $\streamlen$. In \cite{CVM,CVMapprox}, the concept of \textit{Delphic sets} was introduced to formalize the types of properties that can reasonably be assumed for the sets received in the stream.

\begin{definition}[Delphic Families]\label{def:Delphicsetdefinition}
     A set $S\subseteq\univ$ belongs to a Delphic Family if the following can be done in time $O(\log |\univ|)$:
        (1) Getting the size of $S$.
        (2) For any element $x\in\univ$ checking membership in $S$, that is, answering the question \textit{Is $x\in S$?}
        (3) Sampling an element uniformly from $S$.
\end{definition}


For \Cref{def:Fkdefinition,def:SRdefinition,def:SLFAdefinition}, let $\stream=\langle S_1,\dots,S_m\rangle$ be a stream of Delphic sets where $S_j\subseteq\univ$ and for any $x\in\univ$, let $\f{x}=|\{j\in [m]\mid x\in S_j\}|$, that is, $\f{x}$ is the number of sets in the stream where $x$ belongs. 

\begin{definition}[$\F{k}$ for Delphic set streams]\label{def:Fkdefinition}
    $\F{k}$ of $\stream$ is defined as
    \begin{equation*}
        \F{k}=\sum_{x\in\univ}\f{x}^k
    \end{equation*}
\end{definition}

Note that, for $k>0$, this can also be written as $\F{k}=\sum_{x\in\F{0}}\f{x}^k$ because for the other elements $\f{x}=0$ where we abuse a little notation to denote $\F{0}$ as the set of distinct elements in the stream.

\begin{definition}[$\SR(r)$ of a stream]\label{def:SRdefinition}
    For any $r>0$, the \emph{Saturated Richness} $\SR(r)$ of $\stream$ is defined as
    \begin{equation*}
        \SR(r)=\sum_{x\in\univ}\frac{\f{x}}{r+\f{x}}
    \end{equation*}
\end{definition}

\begin{definition}[$\SLFA$ of a stream]\label{def:SLFAdefinition}
    The \emph{Smoothed Log-Frequency Aggregate} $\SLFA$ of $\stream$ is defined as
    \begin{equation*}
        \SLFA)=\sum_{x\in\univ}\ln(1+\f{x})
    \end{equation*}
\end{definition}

\begin{definition}[$\maxfreq$-Frequency-Bounded set streams] \label{def:freqbound}
    A set stream is called $\maxfreq$-Frequency-Bounded if for every element $x\in\univ$, $x$ occurs in at most $\maxfreq$ many sets in the stream. Equivalently, the frequency of $x$, namely $\freq{x}$, satisfies $\freq{x}\le \maxfreq$. 
\end{definition}

\begin{definition}[$(\error,\conf)$-approximation]\label{def:epsilondelta}
    Let $X$ be a value. We say a random variable $\mathsf{Z}$ is an $(\error,\conf)$-approximation of $X$ if $\mathbb P(|X-\mathsf{Z}|\le \error X)\ge 1-\conf$.
\end{definition}

\subsection{Subsampling from the sets in the stream and simulating the oracles}\label{subsec:simulation}


In our algorithms we will ideally like to do the following random process to select a substream.

\ 

True model: \textit{For each set $S_i$ in the stream and each element $x\in S_i$, independently retain $x$ in $S_i$ with probability $(1-\throwprob)$ and discard it with probability $\throwprob$. The resulting thinned set is denoted $S_i^{(\throwprob)}$, and the sub-stream of thinned sets is denoted $\stream_{\throwprob}$. The goal is to answer membership queries, size queries, and get a uniform sample from the thinned sets.}

\ 

Note that element $x$ appears in the thinned sub-stream (i.e., $x\in\bigcup_i S_i^{(\throwprob)}$) if and only if at least one of its $\freq{x}$ independent coins lands ``keep'', which happens with probability $1-\throwprob^{\freq{x}}$. Crucially, since distinct elements $x\ne y$ have disjoint sets of coins, their survival indicators are independent.

If we run the above random process $k$ times (with fresh independent coins each time) we will denote the $k$ sub-streams by $\stream_{\throwprob}^{[1]}, \dots, \stream_{\throwprob}^{[k]}$. In most cases, $\throwprob$ will be equal to $e^{-t}$.

Now we only want to guarantee we can answer membership queries, size queries, and get a uniform sample from the thinned sets. Since constructing the thin sets can take a lot of time and space, we simulate the required oracles for the thinned set by \textit{simulating} the True model. To define this, we need a crucial assumption: \emph{all size queries are asked after all membership queries}. So formally, we define:

\ 

Simulated model: \textit{For each set $S_i$ in the stream, define the following simulated oracles for the thinned set $S_i'$:
\begin{enumerate}
    \item \textbf{Membership oracle:} On query $x$, check whether $x\in S_i$ using the Delphic membership oracle for $S_i$. If $x\notin S_i$, return \textsc{No}. If $x\in S_i$, flip an independent coin $C_x\sim\Ber(1-\alpha)$; return \textsc{Yes} if $C_x=1$, \textsc{No} otherwise.
    \item \textbf{Size oracle (called after all membership queries):} Let $Q_i\subseteq S_i$ be the set of elements for which the membership oracle was invoked and returned a definite answer (i.e., elements $x$ with $x\in S_i$ that were queried). Store the sets $Q_{i,\mathrm{yes}}:=\{x\in Q_i:C_x=1\}$ and $Q_{i,\mathrm{no}}:=\{x\in Q_i:C_x=0\}$ Let $r=|Q_i|$ and $r_+=|Q_{i,\mathrm{yes}}|$. Return 
    \begin{equation*}
        K_i \;=\; r_+ \;+\; K',\qquad K'_i\sim\mathrm{Bin}(|S_i|-r,\;1-\alpha),
    \end{equation*}
    where $K'_i$ is drawn independently of all prior coin flips.
    \item \textbf{Sampling oracle:} Initialize $D_i=\varnothing$ before the first sampling query, where $D_i\subseteq S_i\setminus Q_i$ stores the previously sampled elements from the unqueried part. If $K_i=0$, return $\bot$. Otherwise, with probability
    \[\frac{r_+ + |D_i|}{K_i},\]
    return an element chosen uniformly at random from $Q_{i,\mathrm{yes}}\cup D_i$; or with probability
    \[\frac{K'_i-|D_i|}{K_i},\]
    return an element uniformly at random from $S_i\setminus(Q_i\cup D_i)$, and add the returned element to $D_i$. To sample an element uniformly at random from $S_i\setminus(Q_i\cup D_i)$, we do the following:
    \begin{enumerate}
        \item Draw a uniform sample $y$ from $S_i$ using the Delphic oracle.
        \item Check if $y\in Q_i\cup D_i$. If yes, reject and return to Step a. Else, add $y$ to $D_i$ and return $y$.
    \end{enumerate}
\end{enumerate}}

\

To show that the simulated model is equivalent to the true model under our requirements, we prove \Cref{thm:delphicsimulation}. In words, the theorem says that if an algorithm interacts with a subsampled set $S'$ (of $S$) via \emph{making at most one membership query per element of the universe}, and \emph{asks the membership queries before asking the size query, and asks the size query before any sampling query}, and \emph{the number of distinct samples asked is at most the size query answer}, then \emph{the algorithm is statistically oblivious to whether it is interacting with a subsampled set or a simulation of the subsampled set (as per the simulation notion above)}. Formally, we state this as follows.

\begin{theorem}[Simulation of element-level subsampling via Delphic oracles]\label{thm:delphicsimulation}
Let $S\subseteq\univ$ be a Delphic set with $|S|=n$, and let $p=1-\throwprob\in[0,1]$. Consider the element-level subsampled set $S'=\{x\in S: B_x=1\}$ where $\{B_x\}_{x\in S}$ are independent $\Ber(p)$ random variables. This is the true model. Let the simulated model be defined as above.

Let $\mathcal{A}$ be a (possibly randomized, adaptive) algorithm that interacts with $S'$ through the above three oracle types and satisfies the following two properties:
\begin{enumerate}
    \item[(P1)] $\mathcal{A}$ makes at most one membership query per element $x\in\univ$, and asks at most $T$ many membership queries.
    \item[(P2)] $\mathcal{A}$ completes all membership queries before invoking the size oracle, (possibly) invokes the size oracle before any sampling query, asks at most $T$ sampling queries in total, and the number of distinct samples obtained is at most the answer of the size oracle.
\end{enumerate}
Then:
\begin{enumerate}
    \item[(i)] The joint distribution of the membership answers and the size $K$ observed by $\mathcal A$ under the simulated oracles is identical to that under the true oracles of $S'$.
    \item[(ii)] Conditional on the membership answers and the size query answer, the joint distribution of the entire sequence of answers from the simulated sampling oracle is identical to that obtained by repeated uniform sampling from the fixed true set $S'$. In particular, any procedure for obtaining distinct samples by rejecting duplicates has the same distribution in the two models, and the simulation never returns more distinct elements than the size oracle answer.
    \item[(iii)] The simulated membership and size oracles each require exactly one Delphic oracle call on $S$ plus $O(1)$ additional work, taking $O(\log|\univ|)$ time. Moreover, for any $\delta\in(0,1)$, over the entire sampling phase the simulated sampling oracle uses
    \[ O\!\left(T\left(\log T+\log\frac1\delta\right)\right) \]
    Delphic sampling calls with probability at least $1-\delta$. Consequently, the entire sampling phase takes
    \[ O\!\left(T\left(\log T+\log\frac1\delta\right)\log|\univ|\right) \]
    time with probability at least $1-\delta$. It uses $O(T\log|\univ|)$ bits for storing $Q$ and $D$, up to lower-order counters.
\end{enumerate}
\end{theorem}

We provide the proof of \Cref{thm:delphicsimulation} in \Cref{subsec:simulationappendix}. The proof idea is basically that membership query answers are generated by flipping coins, and the same heads probability is used for the size oracle simulation; thus by freezing the answers of the memberships, we can use a binomial decomposition argument (\Cref{lem:binomialdecompositionappendix}) to get the same joint distribution of membership answers and size query answer in the simulated model as in the true model. For the sampling oracle, the rejection sampling protocol analysis guarantees that samples are generated with the correct uniform probabilities as in the true model. The complexities are analysed via tail bounds. 

\begin{remark}
    We now crucially observe that the expected value established in \Cref{lem:expectedpoly} applies perfectly to the simulated stream generated via \Cref{thm:delphicsimulation}. Because the joint distribution of oracle responses under our simulated model is identical to that of the true model, any $\F{0}$-estimation algorithm satisfying the assumptions of \Cref{thm:delphicsimulation} is statistically oblivious to the substitution. Thus, it estimates the exact target expectation $\sum_{x\in \univ}(1-\throwprob^{\freq{x}})$.
\end{remark}

Now, any algorithm that queries the thinned sets itself can have a failure probability, which now needs to take into account that there is a chance the sampling oracle above fails. To account for this, we have the obvious corollary.

\begin{corollary}[Global Runtime Union Bound]\label{cor:errorforlasvegassampling}
    Let $\mathcal{A}$ be a randomized algorithm interacting with the simulated sampling oracle from \Cref{thm:delphicsimulation} (in particular, satisfies the setup and assumptions of \Cref{thm:delphicsimulation}). Suppose that, over its entire execution, $\mathcal A$ carries out at most $M$ sampling phases (where one sampling phase means all sampling queries made to one simulated set by one algorithm instance). To guarantee that every sampling phase satisfies the high-probability runtime bound of
    \Cref{thm:delphicsimulation}(iii) with overall probability at least
    $1-\delta_{\mathrm{time}}$, it suffices to instantiate each phase with failure probability $\delta'=\frac{\delta_{\mathrm{time}}}{M}$.
\end{corollary}

\begin{proof}
Apply \Cref{thm:delphicsimulation}(iii) to each of the at most $M$
sampling phases with failure probability $\delta'$. A union bound gives
total runtime-failure probability at most $M\delta'=\delta_{\mathrm{time}}$.
\end{proof}


\subsection{Laplace transforms and measures}

\begin{definition}[Laplace transform]\label{def:laplace}
    Let $f:\mathbb R_{\ge 0}\to\mathbb R_{\ge 0}$ be a function. Then its Laplace transform $\mathcal L[f(t)](s)$ is defined as

    \begin{equation*}
        \mathcal L[f(t)](s):=\int_0^\infty f(t)e^{-st}\ \mathrm dt
    \end{equation*}
    where $\mathrm dt$ is the Lebesgue measure on $\mathbb R$. More generally, if $\mu$ is a non-negative (Borel) measure on $[0,\infty)$, then the Laplace transform of $\mu$ is defined as

    \begin{equation*}
        \mathcal L[\mu](s):=\int_0^\infty e^{-st}\ \mathrm d\mu
    \end{equation*}
\end{definition}

There is a related notion of the Laplace transform which we will use.

\begin{definition}[Complementary Laplace transform and expected support]\label{def:complementarylaplace}
    Let $f:\mathbb R_{\ge 0}\to\mathbb R_{\ge 0}$ be a function. Then its complementary Laplace transform $\mathcal L^c[f(t)](s)$ is defined as

    \begin{equation*}
        \mathcal L^c[f(t)](s):=\int_0^\infty f(t)(1-e^{-st})\ \mathrm dt
    \end{equation*}
    where $\mathrm dt$ is the Lebesgue measure on $\mathbb R$. More generally, if $\mu$ is a non-negative (Borel) measure on $[0,\infty)$, then the complementary Laplace transform of $\mu$ is defined as

    \begin{equation*}
        \mathcal L^c[\mu](s):=\int_0^\infty (1-e^{-st})\ \mathrm d\mu
    \end{equation*}
    We may call the complementary Laplace transform of a measure $\mu$ as the \textit{Expected Support} of $\mu$.
\end{definition}

This is called the \emph{complementary} Laplace transform because from \Cref{def:laplace}, we have $\mathcal L^c[\mu](s)=\int_0^\infty (1-e^{st})\ \mathrm d\mu=\mu([0,\infty))-\mathcal L[\mu](s)$.

\subsection{Bernstein functions and L\'evy-Khintchine representation}\label{subsec:BernsteinfunctionLKrepprelim}

\begin{definition}[Bernstein function]\label{def:Bernsteinfunction}
    A continuous function $\varphi:\mathbb R_{\ge 0}\to\mathbb R_{\ge 0}$ is called a Bernstein function if $(-1)^k\varphi^{(k)}(x)\le 0$ for $x>0$ and $k=1,2,\dots$.
\end{definition}

Bernstein functions are those that are \emph{completely monotone} in their derivatives and naturally represent the transforms of positive measures, making them ideal for capturing many useful statistics.

\begin{theorem}[Bernstein's theorem]\label{thm:bernsteinstheorem}
    $\varphi$ is a Bernstein function if and only if

    \begin{equation*}
        \varphi(x)=a+bx+\int_0^\infty (1-e^{-tx})\ \mathrm d\mu(t)
    \end{equation*}
    for some $a,b\ge 0$ and a (Radon) measure $\mu$ such that $\int_0^\infty \min(1, t)\ \mathrm d\mu(t)<\infty$.
\end{theorem}

This representation of $\varphi$ is called the \emph{L\'evy-Khintchine representation} of $\varphi$ and $\mu$ is called the L\'evy measure of $\varphi$. Here $a$ is called the killing term and $b$ is called the drift term. Note that if $\mu$ has a density $r$ (called the L\'evy density of $\varphi$), then $\mathrm d\mu(t)=w(t)\ \mathrm dt$, which will be useful for us.

\subsection{Trapezoidal rule for estimating integrals}\label{subsec:trapezoidalrule}

A definite integral $\int_a^b f(x)\ \mathrm dx$ can we approximated by choosing a partitioning $\{x_k\}_{k=0}^N$ of $[a,b]$ such that $a=x_0<x_1<\cdots<x_{N-1}<x_N=b$, letting $\Delta x_k=x_k-x_{k-1}$ and computing

\begin{equation*}
    \sum_{k=1}^N \frac{f(x_{k-1})+f(x_k)}{2}\Delta x_k
\end{equation*}
This is called the \emph{trapezoidal rule} \cite{atkinson2008introduction}. We will be using a uniform partition, so $x_k=a+k\Delta x$ where $\Delta x=(b-a)/N$ and the trapezoidal rule reads

\begin{equation*}
    \Delta x\left(\frac{f(x_N)+f(x_0)}{2}+\sum_{k=1}^{N-1}f(x_k)\right)
\end{equation*}
Moreover for the uniform partition, if $f\in C^2([a,b])$, then there exists $\xi\in (a,b)$ such that

\begin{equation*}
    \int_a^b f(x)\ \mathrm dx -\Delta x\left(\frac{f(x_N)+f(x_0)}{2}+\sum_{k=1}^{N-1}f(x_k)\right)=-\frac{(b-a)^3}{12N^2}f''(\xi)
\end{equation*}

\subsection{Distributions and concentration inequalities}

We define a random variable $\mathsf{X}$ to follow a Bernoulli distribution with success probability $p$ (denoted as $\mathsf X\sim\Ber(p)$) if
\begin{equation*}
    \mathsf{X}=\begin{cases}1&\text{ with probability $p$}\\
    0&\text{ with probability $1-p$}
    \end{cases}
\end{equation*}
If $\mathsf{X}\sim\Ber(p)$ then $\mathbb E[\mathsf{X}]=p$ (that is, the expectation of a Bernoulli random variable with success probability $p$ is $p$).

\paragraph{} The following concentration inequalities would be used in this paper. 

\begin{theorem}[Chernoff Bounds]
    Let $\mathsf{X}$ be a random variable. Then the following holds:
    
    \begin{equation*}
        \mathbb P(\mathsf X\ge a)\le\inf_{t>0}\mathbb E[e^{t(\mathsf X-a)}]
    \end{equation*} 
    \begin{equation*}
        \mathbb P(\mathsf X\le a)\le\inf_{t<0}\mathbb E[e^{t(\mathsf X-a)}]
    \end{equation*}
\end{theorem}


The following corollaries of the above theorem would be used in our proofs.

\begin{corollary}\label{cor:chernoffcorollary}
    Let $\mathsf X_1,\mathsf X_2,\cdots,\mathsf X_n$ be independent Bernoulli random variables. Also let $\mathsf X=\sum_{i=1}^n \mathsf X_i$ and $\mu=\mathbb E[\mathsf X]$. Then for any $\error\in (0,1)$ 
    \begin{equation*}
        \mathbb P(|\mathsf X-\mu|<\error\mu)\ge 1-2e^{-\error^2\mu/3}
    \end{equation*}
\end{corollary}

\section{Estimating the evaluation of $\espoly$ at a point}\label{sec:evaluateespoly}

To slowly build our algorithms, we will first show how to evaluate the function $\espoly$ (via the polynomial $\poly$) at a given point. Since it is unrealistic to calculate the exact evaluation of the function at a point $t$, we present an algorithm $\evaluate$ which can give an approximate evaluation of the function at $t$.



We show how to evaluate $\espoly$ at a point $t\in [0,\infty)$. Note that evaluating $\espoly$ at $\expt$ is the same as evaluating $\poly$ at $\alpha=e^{-t}$. So showing that we can approximately evaluate $\poly$ at any point $\alpha$ suffices. That is exactly what $\evaluate$ does. The correctness is proven in \Cref{thm:evaluate} and the required corollary for $\espoly$ is proven in \Cref{cor:espolyestimatecorollary}. 


The algorithm $\evaluate$ takes as input a stream $\stream = \langle a_1, \dots, a_{\streamlen}\rangle$, the point of evaluation  $\throwprob$, error parameter $\error$ and a confidence parameter $\conf$ and outputs a value that is an $(\error, \conf)$ estimate of $\poly(\throwprob)$. For this, we will need to use the following sub-routine:

\paragraph{$\estimateZ$ :} The algorithm $\estimateZ$ is the algorithm in \cite{NandiVGMP024}, which takes as input an error parameter $\error$, confidence parameter $\conf$ and on input a stream $\stream$ outputs an $(\error, \conf)$ estimate of the $\F{0}(\stream)$. It also satisfies hypotheses~(P1) and~(P2) of \Cref{thm:delphicsimulation}. Also per set, it does $\tilde O(\error^{-2}\log|\univ|\log^2(1/\error))$ number of sampling queries.

In the algorithm $\evaluate$ we will use multiple parallel instantiations of the algorithm  $\estimateZ$ on different streams. We will denote the $j$th instantiation as  $\estimateZ^{[j]}$.

\



\SetAlgoNoLine%
\begin{algorithm}[htb]
    \DontPrintSemicolon%
    \caption{$\helpevaluate(\stream = \langle a_1, \dots, a_{\streamlen}\rangle, \throwprob,\error)$}
    \SetKwInOut{Input}{Input}
    \SetKwInOut{Output}{Output}
    
    \Input{stream $\stream$, throwing probability $\alpha$, overall error $\error$\;}
    \Output{$(\error,1/3)$-approximation of the expected number of distinct sampled elements;}

    \BlankLine
    
    $k \gets \left\lceil\frac{48\ln(12)}{\error^2(1-\alpha)}\right\rceil$\\
    \For{$i \gets 1$ \KwTo $\streamlen$}{
        \For{$j \gets 1$ \KwTo $k$}{
            Send the thinned set $a_i^{(\throwprob,j)}$ to $\estimateZ^{[j]}(\frac{\error}{4}, \frac{1}{6k})$\; \label{ln:estimateZ}
            \tcp*{element-level subsampling via simulation}
        }
    }
    \BlankLine
    \Return $\frac1k\sum_{j=1}^k(\text{value returned by }\estimateZ^{[j]}(\frac{\error}{4}, \frac{1}{6k}))$\;
    
\end{algorithm}

\SetAlgoNoLine%
\begin{algorithm}[htb]
    \DontPrintSemicolon%
    \caption{$\evaluate(\stream = \langle a_1, \dots, a_{\streamlen}\rangle, \throwprob,\error,\conf)$}
    \SetKwInOut{Input}{Input}
    \SetKwInOut{Output}{Output}
    
    \Input{stream $\stream$, throwing probability $\alpha$, overall error $\error$\;}
    \Output{$(\error,\conf)$-approximation of the expected number of distinct sampled elements;}

    \BlankLine
    
    \If{$\alpha=1$}{
    \Return $0$\;
    }
    $R \gets \lceil48\ln(2/\delta)\rceil$\\
    \For{$r \gets 1$ \KwTo $R$}{
        $Y_r\gets \helpevaluate(\stream,\alpha,\varepsilon)$ 
    }
    \BlankLine
    \Return $\mathrm{median}_r\{Y_r\}$\;
\end{algorithm}

The algorithm $\evaluate$ uses the subroutines $\estimateZ^{[j]}$ inside the helper algorithm $\helpevaluate$ and the correctness of its output is guaranteed in \Cref{thm:evaluate}.


\begin{theorem}\label{thm:evaluate}
    Given any real number $0<\error<1$ and a stream $\stream=\langle a_1,\ a_2\dots,\ a_{\streamlen}\rangle$ and a point of evaluation $\throwprob\in[0,1]$ the algorithm $\evaluate(\stream = \langle a_1, \dots, a_{\streamlen}\rangle, \throwprob,\error,\conf)$
    outputs a $(\error,\conf)$-approximation of $\mathbb E[|\F{0}(S_\throwprob)|]$. For $\throwprob\in[0,1)$, the time and space complexity of the algorithm is
    \[ \tilde O\left(\frac{\log^4|\univ|}{\error^4(1-\throwprob)}\log(1/\error)\log(1/\conf)\left[\log m+\log(1/\error)+\log(1/\conf)+\log\frac1{1-\throwprob}\right]\right). \]
    For $\throwprob=1$, the algorithm returns $0$ exactly.
\end{theorem}

\begin{proof} To prove this theorem, we show that $\helpevaluate$ gives a $(\error,1/3)$-approximation of $\mathbb E[|\F{0}(S_\throwprob)|]$. Then we show that $\evaluate$ achieves the stated guarantees of the theorem.

Let $\stream_{\throwprob}^{[j]}$ be the substream that is send to $\estimateZ^{[j]}(\error, \conf)$ and let $Y^{[j]}_{\throwprob}$ be the exact value of the $\F{0}(\stream_{\throwprob}^{[j]})$. Note that from \Cref{lem:expectedpoly} $Y^{[j]}_{\throwprob}$ is a random variable with $\mathbb{E}[Y^{[j]}_{\throwprob}] = \sum_{x} (1 - \throwprob^{\freq{x}})$, for all $j$.


Let the value returned by $\estimateZ^{[j]}(\frac{\error}{4}, \frac{1}{6k}))$ in the end be $\hat Y_{\throwprob}^{[j]}$ and hence the output of $\evaluate(\stream = \langle a_1, \dots, a_{\streamlen}\rangle, \throwprob,\error,\conf)$ is

\begin{equation*}
     \overline{\hat{Y_\throwprob}}=\frac1k\sum_{j=1}^k\hat Y_\throwprob^{[j]}
\end{equation*}
By the guarantee we have of the subroutine $\estimateZ^{[j]}(\frac{\error}{4}, \frac{1}{6k}))$  we have for any $j \in [k]$,
\[\mathbb P\left(|\hat Y_\throwprob^{[j]}-Y_\throwprob^{[j]}|\le\frac{\error}{4} Y_\throwprob^{[j]}\right)\ge 1-\frac{1}{6k}\]
Thus by union bound with probability at least $(1-\frac16)=5/6$ for all $j$ we have $|\hat Y_\throwprob^{[j]}-Y_\throwprob^{[j]}|\le\error Y_\throwprob^{[j]}/4$  and hence 
$\bigg|\overline{\hat{Y}_\throwprob}-\overline{Y_\throwprob}\bigg|\le\error \overline{Y_\throwprob}/4$,
where $\overline{Y_\throwprob} = \frac{1}{k}\sum_{j=1}^k Y_\throwprob^{[j]}$ and $\overline{\hat{Y}_\throwprob} = \frac{1}{k}\sum_{j=1}^k \hat{Y}_\throwprob^{[j]}$. Thus, with probability at least $5/6$
\begin{equation}\label{eq:hatY}
    \left(1-\frac{\error}{4}\right) \overline{Y_\throwprob} \leq \overline{\hat{Y}_\throwprob}  \leq \left(1+\frac{\error}{4}\right) \overline{Y_\throwprob}
\end{equation}




Now note that for all $j$, $Y_{\throwprob}^{[j]} = \sum_{x\in \univ} \mathbbm{1}_{x\in \F{0}(\stream^{[j]}_\throwprob)}$, where $\mathbbm{1}_{x\in \F{0}(\stream^{[j]}_\throwprob)}$ is the indicator that element $x$ appears in at least one thinned set of the sub-stream $\stream_{\throwprob}^{[j]}$. Under element-level subsampling, for each $x$, the indicator $\mathbbm{1}_{x\in \F{0}(\stream^{[j]}_\throwprob)}$ depends only on the $\freq{x}$ independent coin flips deciding whether $x$ is retained from each set containing it. For distinct elements $x\ne y$, these coin flip sets are disjoint, and coins across different copies $j\ne j'$ are also fresh. Hence the collection $\{\mathbbm{1}_{x\in \F{0}(\stream^{[j]}_\throwprob)}\}_{x\in\univ, j\in[k]}$ consists of mutually independent Bernoulli random variables, and in particular $\sum_{j=1}^k Y_{\throwprob}^{[j]}$ is a sum of independent Bernoulli random variables.\footnote{The fact that element-level subsampling can be efficiently simulated using the Delphic oracles of the original sets (without enumerating them), and that the $\F{0}$ estimation algorithm receives the correct joint distribution of oracle responses under this simulation, is guaranteed by \Cref{thm:delphicsimulation}. Hence we keep arguing with the true model for the correctness, and for the time and space bounds use the simulated model.} Hence, if $\mu_{\throwprob} = \mathbb{E}[\overline Y_\throwprob]$ (which is $\mathbb{E}[\F{0}(\stream_{\throwprob})]$), then from \Cref{cor:chernoffcorollary}

\begin{equation}\label{eq:Y}
\mathbb P\left(|\overline Y_\throwprob-\mu_\throwprob|>\frac{\error}{4}\mu_\throwprob\right)=\mathbb P\left(\left|\sum_{j=1}^k Y_{\throwprob}^{[j]}-k\mu_\throwprob\right|>\frac{\error}{4}k\mu_\throwprob\right)\le2\exp\left(-\frac{(\error/4)^2k\mu_\throwprob}{3}\right)\le \frac16.
\end{equation}
where the first inequality follows from the Chernoff Bound and the fact that $\mathbb{E}[\sum_{j=1}^k Y_{\throwprob}^{[j]}] = k\mu_{\throwprob}$. If the stream is nonempty, then there exists some $x$ with $\freq{x}\ge 1$, and therefore
\[ \mu_\throwprob = \sum_x(1-\throwprob^{\freq{x}}) \ge 1-\throwprob.\]
Thus, by the choice $k\ge \frac{48\ln(12)} {\error^2(1-\throwprob)}$, the second inequality follows. If the stream is empty, then
$\mu_\throwprob=0$ and all the $Y_\throwprob^{[j]}$ are identically
zero, so the claim is immediate.Thus with probability at least $5/6$, 
\begin{equation}\label{eq:Y2}
\left(1-\frac{\error}{4}\right)\mu_\throwprob \le \overline Y_\throwprob
\le \left(1+\frac{\error}{4}\right)\mu_\throwprob.
\end{equation}

On the intersection of the events in
Equations~\ref{eq:hatY} and~\ref{eq:Y2}, which has probability at
least $2/3$, we have
\[ \left(1-\frac{\error}{4}\right)^2\mu_\throwprob \le \overline{\hat Y_\throwprob} \le \left(1+\frac{\error}{4}\right)^2\mu_\throwprob. \]
Since $0<\error<1$,
\[ \left(1-\frac{\error}{4}\right)^2\ge 1-\error \qquad\text{and}\qquad
\left(1+\frac{\error}{4}\right)^2\le 1+\error. \]
Hence the output is a $(\error,1/3)$-approximation of $\mathbb E[|\F{0}(\stream_\throwprob)|]$.

Finally, we analyze the space and update time complexity. The algorithm
$\evaluate$ repeats $\helpevaluate$ exactly $R=O(\log(1/\conf))$
times in parallel. Inside each copy of $\helpevaluate$, we run $k=O\left(\frac{1}{\error^2(1-\throwprob)}\right)$ parallel instantiations of $\estimateZ$. By the construction and analysis of the $\F{0}$ estimator in
\cite{NandiVGMP024}, for a fixed arriving set and a fixed
$\estimateZ$ instance,
\[ T=\tilde O\left(\error^{-2}\log^2|\univ|\log^2(1/\error)\right) \]
is a common upper bound on the number of membership and sampling queries
made to that set. The same estimator uses
$\tilde O(\error^{-2}\log^2|\univ|)$ space, up to logarithmic factors
in the error and confidence parameters. By \Cref{thm:delphicsimulation}(iii), if the failure
probability allocated to this set--instance interaction is $\conf'$, its
entire simulated sampling phase takes $\tilde O\left(T\log(1/\conf')\log|\univ|\right)$ time with probability at least $1-\conf'$.

There are at most $M_{\mathrm{total}}=mRk$ such set-instance interactions over the entire execution of $\evaluate$. We allocate $\conf/2$ to the accuracy failure probability and $\conf/2$ to the runtime failure probability, and set $\conf'=\frac{\conf}{2M_{\mathrm{total}}}$. By \Cref{cor:errorforlasvegassampling}, with probability at least $1-\conf/2$, all simulated sampling phases satisfy their stated runtime bounds. Moreover,
\[ \log\frac1{\conf'} = \tilde O\left(\log m+\log\frac1\error+\log\frac1\conf+\log\frac1{1-\throwprob}\right) \]
up to lower-order logarithmic terms. Multiplying the cost of one
set-instance interaction by the $Rk$ parallel instances gives a
worst-case per-set update time, with the stated high probability, of
\[ \tilde O\left(\frac{\log^4|\univ|}{\error^4(1-\throwprob)}\log(1/\error)\log(1/\conf)\left[\log m+\log(1/\error)+\log(1/\conf)+\log\frac1{1-\throwprob}\right]\right).\]
The total space is bounded by the same expression.
\end{proof}



\begin{corollary}\label{cor:espolyestimatecorollary}
Given any real number $0<\error<1$, a stream
$\stream=\langle a_1,\dots,a_{\streamlen}\rangle$, and a point $t>0$,
\[ \evaluate(\stream=\langle a_1,\dots,a_{\streamlen}\rangle, e^{-t},\error,\conf) \]
outputs an $(\error,\conf)$-approximation of $\espoly(t)$ with update
time per set and space complexity
\[ \tilde O\left(\frac{\log^4|\univ|}{\error^4(1-e^{-t})}\log(1/\error)\log(1/\conf)\left[\log m+\log(1/\error)+\log(1/\conf)+\log\frac1{1-e^{-t}}\right]\right). \]
For $t=0$, $\espoly(0)=0$ and the algorithm returns $0$ exactly.
\end{corollary}

\begin{proof}
Substitute $\throwprob=e^{-t}$ in \Cref{thm:evaluate} and use
\Cref{lem:expectedpoly}, which gives $\mathbb E[|\F{0}(S_{e^{-t}})|]=\espoly(t)$. For $t=0$, we have $e^{-t}=1$ and hence $\espoly(0)=0$ exactly.
\end{proof}

\section{Frequency moments $\F{k}$ - estimation}\label{sec:Fkforkinzeroone}


In this section, we will show how to compute various $\F{k}$ for $k\in (0,1)$.


\subsection{$\F{k}$ when $k\in (0,1)$}\label{sec:Fk01}

There is a driving identity that helps design our algorithm. Let us start with a little intuition first.
\paragraph{\textsc{Intuition:}} Let us think of a jar of marbles with $\f{x}$ marbles of type $x$. A thinning at strength $t$ can be thought of as a \textit{focal length} for the \textit{magnifying glass} that is our probe (that is, the probe scale $t$ now is thought of as a focal length): each marble survives independently with probability $e^{-t}$, and the probe returns \textit{something of type $x$ survived} with expectation $1-e^{-t\f{x}}$. Pool the whole family of probes by averaging these responses against a weight function $w(t)$; mathematically we study

\begin{equation*}
    I(\f{x})=\int_0^\infty (1-e^{-t\f{x}})w(t)\ \mathrm dt
\end{equation*}
What should we choose for our weights $w(t)$? The motivation comes from a simple experiment. Let us put our probes at geometric focal lengths $t_j=2^j,j\in\mathbb Z$; probe $j$ fires roughly when $2^j\lesssim 1/\f{x}$, so the contributing probes are those with $j\le -\log_2\f{x}$. If probe $j$ is given weight $w_j$, the total signal is roughly the tail sum $\sum_{j\le-\log_2\f{x}}w_j$. Scaling $\f{x}\to c\f{x}$ shifts the cutoff by $-\log_2c$; to make the shift always multiply the tail by $c^k$ for every $\f{x}$ (since we are looking at the moments and scaling $\f{x}$ by $c$ scales the moments by $c^k$) the weights must be geometric : $w_j\propto q^j$. A shift by $-\log_2c$ multiplies the tail by $q^{-\log_2c}=c^{-\log_2q}$; setting $-\log_2q=k$ forces $q=2^{-k}$. This \textit{geometric-on-log-grid} corresponds to a continuous density $w(t)\propto t^{-k-1}$ in $t$ (the extra factor comes from a Jacobian). Thus the power law is the unique scale-free choice that turns log-translations (rescaling $\f{x}$) into exact multiplicative rescaling of the output.

In fact, that this weight function works can now be intuitively verified. For $t\ll 1/\f{x}$ we have $1-e^{t\f{x}}\sim t\f{x}$, so the small-$t$ contribution is $\f{x}\int_0^{1/\f{x}}tw(t)\appropto f{x}^k$, and for $t\gg 1/\f{x}$ the integrand saturates at $1$ so the large-$t$ tail contributes $\int_{1/\f{x}}^\infty w(t)\appropto \f{x}^k$. The proportionality constant can be found by setting the case if the jar had only one marble of each type.

\paragraph{} We formalize this now. We provide an elementary proof of \Cref{lem:Fkidentityintegral} and in \Cref{sec:deeperreasonsappendix} we discuss about its deeper connections to the Laplace transform.

\begin{lemma}\label{lem:Fkidentityintegral}
    Let $k\in(0,1)$. Let
    
    \begin{equation*}
        I^{[k]}:=\int_0^\infty \espoly(t)t^{-k-1}\ \mathrm dt
    \end{equation*}
    Then the following identity holds,

    \begin{equation*}
        \F{k}=\frac{kI^{[k]}}{\Gamma(1-k)}
    \end{equation*}
\end{lemma}

\begin{proof}
    To aid us in proving this, let us define

    \begin{equation*}
        I(x):=\int_0^\infty (1-e^{tx})t^{-k-1}\ \mathrm dt
    \end{equation*}
    In \Cref{lem:Fkintegralidentityandconveregence}, we show that $I(x)$ converges and equates to $x^k\frac{\Gamma(1-k)}{k}$ where $\Gamma$ is the Gamma function defined in \Cref{sec:preliminaries}. Assuming this integral identity, we note that

    \begin{align*}
        \int_0^\infty \espoly(t)t^{-k-1}\ \mathrm dt &=\int_0^\infty \sum_{x\in\F{0}}(1-e^{-t\f{x}})t^{-k-1}\ \mathrm dt
        =\sum_{x\in\F{0}}\int_0^\infty (1-e^{-t\f{x}})t^{-k-1}\ \mathrm dt\\
        &=\sum_{x\in\F{0}}I(\f{x})=\sum_{x\in\F{0}}\frac{\Gamma(1-k)}{k}\f{x}^k=\frac{\Gamma(1-k)}{k}\F{k}
    \end{align*}
    This finishes the proof.
\end{proof} 


We now provide an algorithm to estimate $\F{k}$ called $\estimateK$. \Cref{thm:AlgorithmEstimateKcorrectnessandspace} then proves the correctness of the algorithm along with the space and update time complexity.

\SetAlgoNoLine%
\begin{algorithm}[htb]
    \DontPrintSemicolon%
    \caption{$\estimateK(\stream = \langle a_1, \dots, a_{\streamlen}\rangle,\error,\conf)$}
    \algorithmfootnote{Remark : Although $\estimateK$ is written as if $\evaluate(\stream,\dots)$ is called sequentially for each grid point, the implementation is in fact single-pass. As the stream $\stream$ is read, each arriving update is fed simultaneously to the parallel data structures underlying all Evaluate instances (for endpoints and grid points). Thus the algorithm maintains all sketches concurrently and requires only one pass over the stream.}
    
    \SetKwInOut{Input}{Input}
    \SetKwInOut{Output}{Output}
    
    \Input{stream $\stream$, moment order $k\in(0,1)$, overall error $\error$, overall failure probability $\conf$\;}
    \Output{($(\error,\conf)$-approximation of the $k$-th frequency moment $\F{k}$ of $\stream$\;}
    
    $L\gets \dfrac1\maxfreq\left(\dfrac{\varepsilon\Gamma(2-k)}{5k}\right)^{\dfrac1{1-k}}$\;
    $U\gets \left(\dfrac{10(1-k)}{\error}\right)^{1/k}$\;
    $M_k\gets L^{-k-3}\left[(k+1)(k+2)+\dfrac{2(k+1)}e+\dfrac4{e^2}\right]$\;
    $N \gets \bigg\lceil\sqrt{\dfrac{5k(1-k)M_kU^3}{6\error}}\bigg\rceil$\;
    $\Delta t=\dfrac{U-L}{N}$\;
    $c\gets 0$\;
    \BlankLine
    $\label{line:boundFk}\mathrm{bound}=\evaluate(\stream, e^{-U}, \error/5, \conf/(N+1))U^{-k-1}+\evaluate(\stream, e^{-L}, \error/5, \conf/(N+1))L^{-k-1}$\\
    
    \BlankLine
    
    \For{$j \gets 1$ \KwTo $N-1$\label{line:forloopFk}}{
        $c_j\gets \evaluate(\stream, e^{-(L+j\cdot\Delta t)},\error/5,\conf/(N+1))$\label{line:forloopinsideFk}\;
    }
    
    \BlankLine
    
    $c\gets \sum_{j=1}^{N-1}c_j$\;
    $\hat I_{\mathrm{est}}^{[k]}=\Delta t\cdot\left(\frac{\mathrm{bound}}2+c\right)$\;
    \BlankLine
    \label{line:Fkoutput}\Return $\frac{k\hat I_{\mathrm{est}}^{[k]}}{\Gamma(1-k)}$\;
    
    \end{algorithm}


\begin{theorem}\label{thm:AlgorithmEstimateKcorrectnessandspace}
    The algorithm $\estimateK$ outputs an $(\error,\conf)$-approximation of $\F{k}$ when $k\in(0,1)$. $\estimateK$ takes 
    
    \begin{equation*}
        \tilde O\left(a(k)\error^{-4-\frac{3(k+1)}{2k(1-k)}}\maxfreq^{\frac{5+k}{2}}(\log(1/\conf)+\log N)\log(1/\error)[\log m+\log(1/\error)+\log N+\log(1/\conf)]\log^4|\Omega|\right)
    \end{equation*}
    space and update time where $c=\inf_{y\in(1,2]}\Gamma(y)$ which approximately equals $0.88560325\dots$ and
    
    \begin{equation*}
        \log N=\tilde O\left(\frac{k^2+k+6}{2k(1-k)}+\frac{k+3}{2k(1-k)}\log(1/\error)+\frac{3+k}{2}\log\maxfreq+\log k\right)
    \end{equation*}

    \begin{equation*}
        \tilde O_k\left(\tau^{\frac{5+k}{2}}\error^{-4-\frac{3(k+1)}{2k(1-k)}}\log(1/\conf)[\log m+\log(1/\conf)]\log^4|\univ|\right).
    \end{equation*}

\end{theorem}

\begin{proof}
    We start with proving the correctness of the algorithm $\estimateK$, that is, proving that the output of the algorithm $\estimateK$ is an $(\error,\conf)$ estimate of $\F{k}$.
    
    For the sake of the proof, let us define the function $f(t)=\espoly(t)t^{-k-1}$. Let $U$ and $L$ be as defined in the algorithm.
    
    Observe that the output of the algorithm is $\frac{k\hat I_{\mathrm{est}}^{[k]}}{\Gamma(1-k)}$ (Line~\ref{line:Fkoutput}). In Line~\ref{line:boundFk}, the algorithm $\estimateK$ makes two calls to the subroutine $\evaluate$. From \Cref{cor:espolyestimatecorollary} we obtain that the first $\evaluate$ returns an $(\error/5, \conf/(N+1))$-approximation of $\espoly(U)$ and hence the first term is an $(\error/5, \conf/(N+1))$-approximation of $f(U)$ (let us call it $\hat f(U)$) while the second $\evaluate$ returns an $(\error/5, \conf/(N+1))$-approximation of of $\espoly(L)$ and hence the second term is an $(\error/5, \conf/(N+1))$-approximation of $f(L)$ (let us call this $\hat f(L)$). Thus the value of $\mathrm{bound}$ in Line~\ref{line:boundFk} is
    
    \begin{equation} \label{eq:bound}
        \mathrm{bound} = \hat f(U)+\hat f(L)
    \end{equation}
    which is an $(\error/5, 2\conf/(N+1))$-approximation of $f(U)+f(L)$ by an union bound. 

    Similarly in the \textbf{for} loop in Line~\ref{line:forloopFk}-\ref{line:forloopinsideFk}, we have $N-1$ many calls to $\evaluate$, and for each $j$, $c_j=\hat f(L+j\Delta t)$ is an $(\error/5, \conf/(N+1))$-approximation of $f(L+j\Delta t)$.

    Hence $\hat I_{\mathrm{est}}^{[k]}$ is given by  

    \begin{equation}\label{eq:ikest}
    \hat I_{\mathrm{est}}^{[k]}=\Delta t\left(\frac{\hat f(U)+\hat f(L)}{2}+\sum_{j=1}^{N-1}\hat f(t_j)\right)
    \end{equation}
    Finally, the algorithm outputs $\frac{k\hat I_{\mathrm{est}}^{[k]}}{\Gamma(1-k)}$, while the goal of the algorithm is the estimate $\F{k}$ which is (from \Cref{lem:Fkidentityintegral}) equal to $\frac{k\hat I^{[k]}}{\Gamma(1-k)}$. Thus if we show that with probability at least $(1-\delta)$
    
    \begin{equation}\label{eqn:Fkrequiredbound}
    |I^{[k]} - I_{\mathrm{est}}^{[k]}| \leq \error I^{[k]}
    \end{equation}
    then the output of the algorithm is an $(\error, \conf)$ of $\F{k}$. We now prove \Cref{eqn:Fkrequiredbound}. 
    
    $I^{[k]}$ has a singularity at $0$ because near $0$ we have $\espoly(t)\sim \F{1}t$ and thus the integrand behaves like $t^{-1/2}$ which blows up at $0$. Thus to compute $I^{[k]}$ need to isolate the singularity. Hence, we break the integral into three parts : a small part isolating $0$, a part that we estimate numerically, and a tail part

    \begin{equation*}
        I^{[k]}=\underbrace{\int_0^L(\dots)\ \mathrm dt}_{I_{\mathrm{small}}^{[{k}]}}+\underbrace{\int_L^U(\dots)\ \mathrm dt}_{I_{\mathrm{mid}}^{[k]}}+\underbrace{\int_U^\infty(\dots)\ \mathrm dt}_{I_{\mathrm{tail}}^{[k]}}
    \end{equation*}
    We deal with the three parts separately. In \Cref{lem:smallpartFk} we analytically upper bound the value of $I^{[k]}_{\mathrm{small}}$ and obtain  
    
    \begin{equation}\label{eqn:newsmallmain}
        I_{\mathrm{small}}^{[k]}\le (\error/5)I^{[k]}
    \end{equation}
    Then in \Cref{lem:tailpartFk} we again analytically upper bound the value of $I^{[k]}_{\mathrm{tail}}$ and obtain
    
    \begin{equation}\label{eqn:newtailmain}
        I_{\mathrm{tail}}^{[k]}\le (\error/5)I^{[k]}
    \end{equation}
    Finally in \Cref{lem:Emidbound} we show that the distance between $I^{[k]}_{\mathrm{mid}}$ and the $\hat I^{[k]}_{\mathrm{est}}$ is small with respect to $I^{[k]}$; more precisely we show that

    \begin{equation}\label{eqn:newIkIestbound}
        |I^{[k]}_{\mathrm{mid}}-\hat I^{[k]}_{\mathrm{est}}|\le \left(\frac{2\error}{5}+\frac{\error^2}{25}\right)I^{[k]}
    \end{equation}
    Using the triangle inequality and \Cref{eqn:newsmallmain}, \Cref{eqn:newtailmain} and \Cref{eqn:newIkIestbound} we obtain 
    
    \begin{align*}
        \frac{\bigg|I^{[k]}-\hat{I}_{\mathrm{est}}^{[k]}\bigg|}{I^{[k]}}
        &=\frac{|I_{\mathrm{small}}^{[k]}+I_{\mathrm{mid}}^{[k]}+I_{\mathrm{tail}}^{[k]}-\hat I_{\mathrm{est}}^{[k]}|}{I^{[k]}}
        \le \frac{I_{\mathrm{small}}^{[k]}}{I^{[k]}}+\frac{|I_{\mathrm{mid}}^{[k]}-\hat I_{\mathrm{est}}^{[k]}|}{I^{[k]}}+\frac{I_{\mathrm{tail}}^{[k]}}{I^{[k]}}\\
        &\le \frac{\error}{5}+\left(\frac{2\error}{5}+\frac{\error^2}{25}\right)+\frac{\error}{5}\\
        &=\frac{4\error}{5}+\frac{\error^2}{25}
        \le\error
    \end{align*}
    Thus $\hat I^{[k]}_{\mathrm{est}}$ is an $(\error,\conf)$ estimator of $I^{[k]}$ and hence the output of algorithm $\estimateK$ is an $(\error,\conf)$ estimator of $\F{k}$.

    Now we analyse the space and update time complexity of $\estimateK$. Note that the number of evaluations of our polynomial (which is the number of $\evaluate$ calls) required is $O(N)$ which evaluates to (for fixed $k\in (0,1)$ and $\varepsilon\to 0$) 
    
    \begin{equation*}
        N=O\left(k^{\frac32}(1-k)^{\frac{3+k}{2k}}5^{\frac3{2k}+\frac{3+k}{2(1-k)}}\varepsilon^{-\frac{k+3}{2k(1-k)}}\tau^{\frac{3+k}{2}}c^{-\frac{k+3}{2k(1-k)}}\right)
    \end{equation*}
    where $c=\inf_{y\in(1,2]}\Gamma(y)$ which approximately equals $0.88560325\dots$. Moreover, since every call to $\evaluate$ is made at a point $t\ge L$,
using $\Gamma(2-k)\ge c$ and
$\frac{1}{1-e^{-t}}\le 1+\frac1t$, we have
\[\frac{1}{1-e^{-t}}\le\left(1+\left(\frac{5k}{c}\right)^{\frac1{1-k}}\right)\tau\error^{-\frac1{1-k}}.\]
The additional logarithmic term $\log(1/(1-e^{-t}))$ from
\Cref{cor:espolyestimatecorollary} is $\tilde O_k(\log\maxfreq+\log(1/\error))$ and is absorbed by the
logarithmic factors below.

The space used and update time then is $O(N)$ times the space used and update time of $\evaluate$ (proven in \Cref{cor:espolyestimatecorollary}), which is



    
    \begin{equation*}
        \tilde O\left(a(k)\varepsilon^{-4-\frac{k+3}{2k(1-k)}}\tau^{\frac{3+k}{2}}(\log(1/\conf)+\log N)\log(1/\error)[\log |\univ|+\log\maxfreq+\log(1/\error\conf)+\log N]\log^4|\Omega|\right)
    \end{equation*}
    with 
    
    \begin{equation*}
        \log N=\tilde O\left(\frac{k^2+k+6}{2k(1-k)}+\frac{k+3}{2k(1-k)}\log(1/\error)+\frac{3+k}{2}\log\maxfreq+\log k\right)
    \end{equation*}
    
    \begin{equation*}
        a(k)=\left(1+\left(\frac{5k}{c}\right)^{\frac1{1-k}}\right)k^{\frac32}(1-k)^{\frac{3+k}{2k}}5^{\frac{3}{2k}+\frac{3+k}{2(1-k)}}c^{-\frac{k+3}{2k(1-k)}}.
    \end{equation*}
    For readability if we hide the factors of $k,\log(1/\error)$ and $\log(1/\maxfreq)$, this reads as 
    
    \begin{equation*}
        \tilde O_k\left(\tau^{\frac{5+k}{2}}\error^{-4-\frac{3(k+1)}{2k(1-k)}}\log(1/\conf)[\log m+\log(1/\conf)]\log^4|\univ|\right).
    \end{equation*}

\end{proof}

To completely finish the proof of \Cref{thm:AlgorithmEstimateKcorrectnessandspace}, we now prove the lemmas used in the proof.

\begin{lemma}\label{lem:smallpartFk}
    By choosing $L=\frac1{\maxfreq}\left[\frac{\error\Gamma(2-k)}{5k}\right]^{\frac1{1-k}}$, we have $I_{\mathrm{small}}^{[k]}\le (\error/5)I^{[k]}$.
\end{lemma}

\begin{proof}
    For this part, note that

    \begin{equation*}
        0\le I_{\mathrm{small}}^{[k]}=\int_0^L\espoly(t)t^{-k-1}\ \mathrm dt
        \le\int_0^Lmt\cdot t^{-k-1}\ \mathrm dt
        =m\int_0^Lt^{-k}\ \mathrm dt
        =\frac{mL^{1-k}}{1-k}
    \end{equation*}
    where the second inequality follows since $1-e^{-x}\le x$. Now also note that
    
    \begin{equation*}
        m=\sum_{x\in\F{0}} \f{x}=\sum_{x\in\F{0}}\f{x}^{1-k}\f{x}^k\le \maxfreq^{1-k}\sum_{x\in\F{0}}\f{x}^k=\maxfreq^{1-k}\F{k}
    \end{equation*}
    and hence $m/\F{k}\le \maxfreq^{1-k}$. Thus using $L=\frac1{\maxfreq}\left[\frac{\error\Gamma(2-k)}{5k}\right]^{\frac1{1-k}}$ we have $I_{\mathrm{small}}^{[k]}\le (\error/5)I^{[k]}$.
    
\end{proof}

\begin{lemma}\label{lem:tailpartFk}
    By choosing $U=\left[\frac{10(1-k)}{\error}\right]^{1/k}$, we have $I_{\mathrm{tail}}^{[k]}\le (\error/5)I^{[k]}$.
\end{lemma}

\begin{proof}
    For the tail bound, we begin by noting that
    
    \begin{equation*}
        \espoly(t)=\sum_{x\in\F{0}}(1-e^{-t\f{x}})\le \sum_{x\in\F{0}}1=\F{0}
    \end{equation*}
    Then
    
    \begin{equation*}
        I_{\mathrm{tail}}^{[k]}=\int_U^\infty\espoly(t)t^{-k-1}\ \mathrm dt
        \le |\F{0}|\int_U^\infty t^{-k-1}\ \mathrm dt
        =\frac{|\F{0}|}{kU^k}
    \end{equation*}
    Also
    
    \begin{equation}\label{eqn:upperboundIfork}
        I^{[k]}=\frac{\Gamma(1-k)}k\F{k}=\frac{\Gamma(1-k)}k\sum_{x\in\F{0}}\f{x}^k\ge \frac{\Gamma(1-k)}k\sum_{x\in\F{0}}1=\frac{\Gamma(1-k)}k|\F{0}|\ge \frac1{2k(1-k)}|\F{0}|
    \end{equation}
    where the last equality follows from \Cref{lem:Gammabound}. And thus using $U=\left[\frac{10(1-k)}{\error}\right]^{1/k}$, we get $$I_{\mathrm{tail}}^{[k]}\le (\error/5)I^{[k]}$$
    
\end{proof}

\begin{lemma}\label{lem:Emidbound}
    
    
    By choosing $N=\bigg\lceil\sqrt{\frac{5k(1-k)M_kU^3}{6\error}}\bigg\rceil$ as in $\estimateK$, we have $\bigg|I^{[k]}_{\mathrm{mid}}-\hat I^{[k]}_{\mathrm{est}}\bigg|\le \left(\frac{2\error}5+\frac{\error^2}{25}\right)I^{[k]} $.
\end{lemma}

\begin{proof}
    Recall from \Cref{eq:ikest} that

    \begin{equation*}
    \hat I_{\mathrm{est}}^{[k]}=\Delta t\left(\frac{\hat f(U)+\hat f(L)}{2}+\sum_{j=1}^{N-1}\hat f(t_j)\right)
    \end{equation*}
    Since $\hat f(x)$ is an $(\error/5,\conf/(N+1))$ approximation of $f(x)$ for $x\in \{L,U\}\cup \{L+tj\}_{t\in [N-1]}$, $\hat I^{[k]}_{\mathrm{est}}$ is an $(\error,\conf)$-approximation of

    \begin{equation*}
        \hat I^{[k]}_{\mathrm{mid}}:=\Delta t\left(\frac{f(U)+f(L)}{2}+\sum_{j=1}^{N-1}f(t_j)\right)
    \end{equation*}
    by the union bound. That is, with probability at last $1-\conf$
    
    \begin{equation}\label{eqn:boundforhatmidandest}
        |\hat I^{[k]}_{\mathrm{mid}}-\hat I^{[k]}_{\mathrm{est}}|\le \frac{\error}5\hat I^{[k]}_{\mathrm{mid}}
    \end{equation}
    Also, by the definition of trapezoidal rule in \Cref{subsec:trapezoidalrule}, there exists some $\xi\in (L,U)$ such that

    \begin{equation*}
        |\hat I^{[k]}_{\mathrm{mid}}-I^{[k]}_{\mathrm{mid}}|=\frac{(U-L)^3}{12N^2}|f''(\xi)|\le \frac{U^3}{12N^2}|f''(\xi)|
    \end{equation*}
    Then using \Cref{lem:doublederivativebound} gives us that
    
    \begin{equation*}
        |f''(t)|\le\frac1{L^{k+3}}\sum_{x\in\F{0}}\left[(k+1)(k+2)+\frac{2(k+1)}e+\frac4{e^2}\right]=M_k|\F{0}|
    \end{equation*}
    where
    
    \begin{equation*}
        M_k=L^{-k-3}\left[(k+1)(k+2)+\frac{2(k+1)}e+\frac4{e^2}\right]
    \end{equation*}
    and thus using \Cref{eqn:upperboundIfork} and setting $N=\bigg\lceil\sqrt{\frac{5k(1-k)M_kU^3}{6\error}}\bigg\rceil$ we have
    
    \begin{equation}\label{eqn:Emidboundinlemma}
        \frac{\bigg|\hat I^{[k]}_{\mathrm{mid}}-I^{[k]}_{\mathrm{mid}}\bigg|}{I^{[k]}}= \frac{U^3}{12N^2}\frac{2k(1-k)M_k|\F{0}|}{|\F{0}|}=\frac{k(1-k)M_kU^3}{6N^2}\le \frac{\error}{5}
    \end{equation}
    Then using \Cref{eqn:boundforhatmidandest}, \Cref{eqn:Emidboundinlemma} and the fact that $I^{[k]}_{\mathrm{mid}}\le I^{[k]}$, we obtain

    \begin{align*}
        |I^{[k]}_{\mathrm{mid}}-\hat I^{[k]}_{\mathrm{est}}|&=|I^{[k]}_{\mathrm{mid}}-\hat I^{[k]}_{\mathrm{mid}}+\hat I^{[k]}_{\mathrm{mid}}-\hat I^{[k]}_{\mathrm{est}}|\le |I^{[k]}_{\mathrm{mid}}-\hat I^{[k]}_{\mathrm{mid}}|+|\hat I^{[k]}_{\mathrm{mid}}-\hat I^{[k]}_{\mathrm{est}}|\\
        &\le \frac{\error}5I^{[k]}+\frac{\error}5I^{[k]}_{\mathrm{mid}}\\
        &\le \frac{\error}{5}I^{[k]}+\frac{\error}{5}\left(1+\frac{\error}{5}\right)I^{[k]}\\
        &=\left(\frac{2\error}5+\frac{\error^2}{25}\right)I^{[k]} 
    \end{align*}
\end{proof}

\section{Other statistics}\label{sec:otherstatistics}

Using similar techniques that we used for $\F{k}$ estimation in \Cref{sec:Fk01}, we can estimate several other statistics. We provide examples of Saturated Richness and Smoothed Log-Frequency Aggregate as examples, and then treat the general case in \Cref{sec:generalbernsteinestimationsection}.

\subsection{Saturated Richness}

Let us consider a $\maxfreq$-frequency bounded stream with stream length $\streamlen$. The \emph{Saturated richness} of such a stream is defined as

\begin{equation*}
    \SR=\sum_{x\in\F{0}}\frac{\f{x}}{\f{x}+1}
\end{equation*}
We will actually be looking at a more general form of this statistic, namely the \emph{$r$-Saturated richness} of a stream, defined at a scale $r>0$ as (also see \Cref{def:SRdefinition})

\begin{equation*}
    \SR(r)=\sum_{x\in\F{0}}\frac{\f{x}}{\f{x}+r}
\end{equation*}
We chose the name \emph{Saturated Richness} because it builds on the basic idea of "richness" in streaming algorithms - the number of unique items, or $\F{0}$, which is key for stuff like spotting heavy hitters or estimating counts in massive data streams. The name “saturated” reflects that each type’s contribution is capped (saturates) as its frequency grows: rare types contribute about $1/2$ while very frequent types contribute nearly $1$. Note that $\SR(1)=\SR$, which prompts the generalization.

\paragraph{} To facilitate developing our algorithm, we start again with an intuition.

\paragraph{\textsc{Intuition:}} Let us consider the jar of marbles again, and this time let us pick a random probe scale $t\sim \mathrm{Exp}(r)$. Then for a marble type $x$, the average response we get from such sensors is $\int_0^\infty (1-e^{-t\f{x}})re^{-rt}\ \mathrm dt$. Thought of in another way, this is actually a race; attach a clock to each marble of type $x$, where the clock follows a $\mathrm{Exp}(1)$ distribution (so these clocks have rate $1$), and there is a global clock (the probe scale) moving at rate $r$ - a marble survived if its own clock goes off before the global clock. Then by standard probability, the probability that at least one marble of type $x$ has survived after the global clock goes off is $\f{x}/(r+\f{x})$, which is also the expected value of the sensor output at scale $t$ (since the sensor replies if a marble survives or not at that strength and we aggregate those values). Thus we have the identity

\begin{equation*}
    \int_0^\infty (1-e^{-t\f{x}})re^{-rt}\ \mathrm dt=\frac{\f{x}}{\f{x}+r}
\end{equation*}

Now we formalize this. We provide an elementary proof of \Cref{lem:SRintegralidentity} and in \Cref{sec:deeperreasonsappendix} we discuss about its deeper connections to the Laplace transform.

\begin{lemma}\label{lem:SRintegralidentity}
    Let

    \begin{equation*}
        I_{\SR}=\int_0^\infty \espoly(t)re^{-rt}\ \mathrm dt
    \end{equation*}
    Then the following identity holds,

    \begin{equation*}
        \SR(r)=I_{\SR(r)}
    \end{equation*}
\end{lemma}

\begin{proof}
   Note that
    
    \begin{align*}
        \int_0^\infty (1-e^{-\f{x}t})re^{-rt}\ \mathrm dt&=r\int_0^\infty e^{-rt}\ \mathrm dt-r\int_0^\infty e^{-(\f{x}+r)t}\ \mathrm dt=1-\frac{r}{\f{x}+r}=\frac{\f{x}}{\f{x}+r}
    \end{align*}
    Then

    \begin{align*}
        I_{\SR(r)}&=\int_0^\infty\espoly(t)re^{-rt}\ \mathrm dt=\int_0^\infty \sum_{x\in\F{0}}(1-e^{-\f{x}t})re^{-rt}\ \mathrm dt=\sum_{x\in\F{0}}\int_0^\infty (1-e^{-\f{x}t})re^{-rt}\ \mathrm dt\\
        &=\sum_{x\in\F{0}}\frac{\f{x}}{\f{x}+r}=\SR(r)
    \end{align*}
\end{proof}

We will need a small lemma in the proof of correctness (\Cref{thm:estimateSRcorrectness}, specifically \Cref{lem:SRtailbound})

\begin{lemma}\label{lem:increasing}
    The functon $\frac{x}{x+r}$ is increasing for $x\ge 0$ and fixed $r>0$.
\end{lemma}

\begin{proof}
    Let $h(x)$ be the function. Then $h'(x)=\frac{(x+r)-x}{(x+r)^2}=\frac{r}{(x+r)^2}>0$ and thus the result holds.
\end{proof}

We now provide an algorithm to compute $\SR(r)$ called $\estimateSR$. \Cref{thm:estimateSRcorrectness} then proves the correctness of the algorithm along with the space and update time complexity.

\SetAlgoNoLine%
\begin{algorithm}[htb]
    \DontPrintSemicolon%
    \caption{$\estimateSR(\stream = \langle a_1, \dots, a_{\streamlen}\rangle,r,\error,\conf)$}
    
    \SetKwInOut{Input}{Input}
    \SetKwInOut{Output}{Output}
    
    \Input{stream $\stream$, $r>0$, overall error $\error$, overall failure probability $\conf$\;}
    \Output{($(\error,\conf)$-approximation of the $r$-Saturated richness ($\SR(r)$) of $\stream$\;}
    
    $U\gets \dfrac1r\ln\left(\frac{4(r+1)}{\error}\right)$\;
    $N \gets \bigg\lceil\sqrt{\dfrac{U^3r(r+1)(r^2+(r+\maxfreq)^2)}{3\error}}\bigg\rceil$\;
    $\Delta t=\dfrac{U}{N}$\;
    $c\gets 0$\;
    \BlankLine
    $\label{line:boundSR}\mathrm{bound}=\evaluate(\stream, e^{-U}, \error/4, \conf/(N+1))re^{-rU}$\\
    
    \BlankLine
    
    \For{$j \gets 1$ \KwTo $N-1$\label{line:forloopSR}}{
        $c_j\gets \evaluate(\stream, e^{-j\cdot\Delta t},\error/4,\conf/(N+1))$\label{line:forloopinsideSR}\;
    }
    
    \BlankLine
    
    $c\gets \sum_{j=1}^{N-1}c_j$\;
    $\hat I_{\mathrm{est}}=\Delta t\cdot\left(\frac{\mathrm{bound}}2+c\right)$\;
    \BlankLine
    \label{line:SRoutput}\Return $\hat I_{\mathrm{est}}$\;
    
\end{algorithm}

\begin{theorem}\label{thm:estimateSRcorrectness}
    The algorithm $\estimateSR$ outputs an $(\varepsilon,\delta)$-approximation of $\SR(r)$. $\estimateSR$ takes $\tilde O\left(\maxfreq^2\error^{-5}\ln^2(1/\error)(\log|\univ|+\log\maxfreq+\log(1/\error)+\log(1/\conf))\log^4|\univ|\right)$ space and update time.
\end{theorem}

\begin{proof}
    We start with proving the correctness of the algorithm $\estimateSR$, that is, proving that the output of the algorithm $\estimateSR$ is an $(\error,\conf)$ estimate of $\SR(r)$.
    
    For the sake of the proof, let us define the function $f(t)=\espoly(t)re^{-rt}$. Let $U$ be as defined in the algorithm.
    
    Observe that the output of the algorithm is $\hat I_{\mathrm{est}}$ (Line~\ref{line:SRoutput}). In Line~\ref{line:boundSR}, the algorithm $\estimateSR$ makes a call to the subroutine $\evaluate$. From \Cref{cor:espolyestimatecorollary} we obtain that the $\evaluate$ call returns an $(\error/4, \conf/(N+1))$-approximation of $\espoly(U)$ and hence the first term is an $(\error/4, \conf/(N+1))$-approximation of $f(U)$ (let us call it $\hat f(U)$). Thus the value of $\mathrm{bound}$ in Line~\ref{line:boundSR} is
    
    \begin{equation} \label{eq:boundSR}
        \mathrm{bound} = \hat f(U)
    \end{equation}
    which is an $(\error/5, \conf/(N+1))$-approximation of $f(U)$. 

    Similarly in the \textbf{for} loop in Line~\ref{line:forloopSR}-\ref{line:forloopinsideSR}, we have $N-1$ many calls to $\evaluate$, and for each $j$, $c_j=\hat f(j\Delta t)$ is an $(\error/5, \conf/(N+1))$-approximation of $f(j\Delta t)$.

    Hence $I_{\mathrm{est}}$ is given by  

    \begin{equation}\label{eq:IestSR}
    \hat I_{\mathrm{est}}=\Delta t\left(\frac{\hat f(U)}{2}+\sum_{j=1}^{N-1}\hat f(t_j)\right)
    \end{equation}
    Finally, the algorithm outputs $\hat I_{\mathrm{est}}$, while the goal of the algorithm is the estimate $\SR(r)$ which is (from \Cref{lem:SRintegralidentity}) equal to $I_{\SR(r)}$. Thus if we show that with probability at least $(1-\delta)$
    
    \begin{equation}\label{eqn:SRrequiredbound}
    |I_{\SR(r)}- I_{\mathrm{est}}| \leq \error
    \end{equation}
    then the output of the algorithm is an $(\error, \conf)$ of $\F{k}$. We now prove \Cref{eqn:SRrequiredbound}.

    $I_{\SR(r)}$ has no singularity at $0$, so it suffices to break it into two pieces, one a tail piece and the remaining piece we estimate using the trapezoidal rule. Concretely, we do

    \begin{equation*}
        I_{\SR(r)}=\underbrace{\int_0^U(\dots)\ \mathrm dt}_{I_{\mathrm{trap}}}+\underbrace{\int_U^\infty(\dots)\ \mathrm dt}_{I_{\mathrm{tail}}}
    \end{equation*}
    We deal with the two parts separately. In \Cref{lem:SRtailbound} we analytically upper bound the value of $I_{\mathrm{tail}}$ and obtain

    \begin{equation}\label{eqn:tailboundSR}
        I_{\mathrm{tail}}\le \frac{\error}4I_{\SR(r)}
    \end{equation}
    Then in \Cref{lem:ItrapIestboundSR} we show that the distance between $I_{\mathrm{trap}}$ and $\hat I_{\mathrm{est}}$ is small with respect to $I_{\SR(r)}$; more precisely we show that

    \begin{equation}\label{eqn:IkIestboundSR}
        |I_{\mathrm{trap}}-\hat I_{\mathrm{est}}|\le \left(\frac{\error}2+\frac{\error^2}{16}\right)I_{\SR(r)}
    \end{equation}
    Using the triangle inequality and \Cref{eqn:tailboundSR} and \Cref{eqn:IkIestboundSR}, we obtain

    \begin{align*}
        \frac{|I_{\SR(r)}-\hat I_{\mathrm{est}}|}{I_{\SR(r)}}&=\frac{|I_{\mathrm{trap}}+I_{\mathrm{tail}}-\hat I_{\mathrm{est}}|}{I_{\SR(r)}}\le \frac{|I_{\mathrm{trap}}-\hat I_{\mathrm{est}}|}{I_{\SR(r)}}-\frac{I_{\mathrm{tail}}}{I_{\SR(r)}}\\
        &\le \left(\frac{\error}2+\frac{\error^2}{16}\right)+\frac{\error}4=\frac{3\error}4+\frac{\error^2}{16}\le \error
    \end{align*}
    Thus $I_{\mathrm{est}}$ is an $(\error,\conf)$ estimator of $I_{\mathrm{\SR}}$ and hence the output of algorithm $\estimateSR$ is an $(\error,\conf)$ estimator of $\SR(r)$.
    
    Now we analyse the space and update time complexity of $\estimateSR$. Note that the number of evaluations of our polynomial (which is the number of $\evaluate$ calls) is $O(N)$ which is $O(\maxfreq\error^{-1/2}\ln^{3/2}(1/\error))$. Plugging this in along with the space and update time of one $\evaluate$ call (which is proven in \Cref{cor:espolyestimatecorollary}) gives the space and update time complexity of $\estimateSR$ as

    \begin{equation*}
        \tilde O\left(\maxfreq^2\error^{-5}\ln^2(1/\error)(\log|\univ|+\log\maxfreq+\log(1/\error)+\log(1/\conf))\log^4|\univ|\right)
    \end{equation*}

\end{proof}

To completely finish the proof of \Cref{thm:estimateSRcorrectness}, we now prove the lemmas used in the proof.

\begin{lemma}\label{lem:SRtailbound}
    By choosing $U=\frac1r\ln(4(r+1)/\error)$, we have $I_{\mathrm{tail}}\le (\error/4)I_{\SR(r)}$.
\end{lemma}

\begin{proof}
    Observe that

    \begin{equation*}
        I_{\SR(r)}=\SR(r)=\sum_{x\in\F{0}}\frac{\f{x}}{\f{x}+r}\ge\sum_{x\in\F{0}}\frac1{1+r}=\frac{|\F{0}|}{1+r}
    \end{equation*}
    where for the inequality, we used the fact that $\frac{\f{x}}{\f{x}+r}$ is increasing in $\f{x}$ from \Cref{lem:increasing}. Now for $I_{\mathrm{tail}}$, using $(1-e^{-\lambda t})\le 1$ we have

    \begin{equation*}
        I_{\mathrm{tail}}=\int_U^\infty\espoly(t)re^{-rt}\ \mathrm dt\le |\F{0}|\int_U^\infty re^{-rt}\ \mathrm dt=|\F{0}|e^{-rU}
    \end{equation*}
    Thus using $U=\frac1r\ln(4(r+1)/\error)$, we have

    \begin{equation*}
        \frac{I_{\mathrm{tail}}}{I_{\SR(r)}}\le \frac{\error}{4}
    \end{equation*}
\end{proof}

\begin{lemma}\label{lem:ItrapIestboundSR}
    By choosing $N=\bigg\lceil\sqrt{\frac{U^3r(r+1)(r^2+(r+\maxfreq)^2)}{3\error}}\bigg\rceil$ as in $\estimateSR$, we have $|I_{\mathrm{trap}}-\hat I_{\mathrm{est}}|\le \left(\frac{\error}2+\frac{\error^2}{16}\right)I_{\SR(r)}$
\end{lemma}

\begin{proof}
    Recall from \Cref{eqn:IkIestboundSR} that

    \begin{equation*}
        \hat I_{\mathrm{est}}=\Delta t\left(\frac{\hat f(U)}2+\sum_{j=1}^{N-1}\hat f(t_j)\right)
    \end{equation*}
    Since $\hat f(x)$ is an $(\error/4,\conf/(N+1))$ approximation of $f(x)$ for $x\in \{U\}\cup \{j\Delta t\}_{j\in [N-1]}$, $\hat I_{\mathrm{est}}$ is an $(\error,\conf)$-approximation of

    \begin{equation*}
        \hat I_{\mathrm{trap}}:=\Delta t\left(\frac{f(U)}{2}+\sum_{j=1}^{N-1}f(t_j)\right)
    \end{equation*}
    by the union bound. That is, with probability

    \begin{equation}\label{eqn:midestboundSR}
        |\hat I_{\mathrm{trap}}-\hat I_{\mathrm{est}}|\le \frac{\error}4\hat I_{\mathrm{mid}}
    \end{equation}
    Also, by the definition of trapezoidal rule in \Cref{subsec:trapezoidalrule}, there exists some $\xi\in (L,U)$ such that

    \begin{equation*}
        |\hat I_{\mathrm{trap}}-I_{\mathrm{\SR}}|=\frac{(U-L)^3}{12N^2}|f''(\xi)|\le \frac{U^3}{12N^2}|f''(\xi)|
    \end{equation*}
    Computation of the upper bound on the absolute value of the double derivative of $f$ from \Cref{lem:SRfdoublederivbound} gives us that

    \begin{equation*}
        |f''(t)|\le r\sum_{x\in\F{0}}(r^2+(r+\maxfreq)^2)=|\F{0}|r(r^2+(r+\maxfreq)^2)
    \end{equation*}
    Hence setting $N=\bigg\lceil\sqrt{\frac{U^3r(r+1)(r^2+(r+\maxfreq)^2)}{3\error}}\bigg\rceil$, we have

    \begin{equation}\label{eqn:ItrapboundSR}
        \frac{\bigg|I_{\mathrm{trap}}-\hat I_{\mathrm{trap}}\bigg|}{I}\le\frac{U^3}{12N^2}\frac{|\F{0}|r(r^2+(r+\maxfreq)^2)(1+r)}{|\F{0}|}\le \frac{\error}{4}\Rightarrow |I_{\mathrm{trap}}-\hat I_{\mathrm{trap}}|\le \frac{\error}4I
    \end{equation}
    Then using \Cref{eqn:midestboundSR}, \Cref{eqn:ItrapboundSR} and the fact that $I_{\mathrm{trap}}\le I_{\SR(r)}$, we obtain

    \begin{align*}
        |I_{\mathrm{trap}}-\hat I_{\mathrm{est}}|&=|I_{\mathrm{trap}}-\hat I_{\mathrm{trap}}-\hat I_{\mathrm{trap}}+I_{\mathrm{est}}|\le |I_{\mathrm{trap}}-\hat I_{\mathrm{trap}}|+|\hat I_{\mathrm{trap}}+I_{\mathrm{est}}|\\
        &\le \frac{\error}{4}I_{\SR(r)}+\frac{\error}{4}\hat I_{\mathrm{trap}}\\
        &\le \frac{\error}{4}I_{\SR(r)}+\frac{\error}{4}\left(1+\frac{\error}{4}\right)I_{\mathrm{\SR}}\\
        &=\left(\frac{\error}{2}+\frac{\error^2}{16}\right)I_{\SR(r)}
    \end{align*}    
\end{proof}

\subsection{Smoothed Log-Frequency Aggregate}

Again, let us consider a $\maxfreq$-frequency bounded stream with stream length $m$. Recall from \Cref{def:SLFAdefinition} that the \emph{Smoothed Log-Frequency Aggregate} (or \emph{SLFA} in short) of such a stream is defined as

\begin{equation*}
    \SLFA=\sum_{x\in\F{0}}\ln(1+\f{x})
\end{equation*}
We named this statistic the \emph{Smoothed Log-Frequency Aggregate} ($\SLFA$) because the "$+1$" employs Laplace smoothing to address zero frequencies and prevent undefined logarithms, the logarithmic transformation mitigates the effect of extremely high frequencies for a more balanced representation, and \emph{aggregate} reflects the summation that yields a single measure of the overall smoothed frequency distribution—effectively capturing the "total effective logarithmic scale" of the data in a set. This formulation draws from established practices in information retrieval, where $\log(1 + tf)$ is used for sublinear term frequency weighting in tf-idf, and summing such values can serve as an $L_1$ norm in document vector similarity computations and sentiment analysis, see for example, \cite{addiga2022sentiment}. In ecology, sums of $\log(\mathrm{abundance} + 1)$ normalize skewed species counts for community-level analyzes in biodiversity assessments; such analyzes appear in microbiome ecology reviews, where $\log(\mathrm{abundance} + 1)$ help to discern features in rarely occurring objects \cite{mcknight2019methods}. In Bayesian statistics, it appears in contexts like pre-processing for count data models, where $\log(y + 1/2)$ or similar helps in proposal distributions for sampling algorithms like in \cite{hoff2009first}.

\paragraph{} Again, we start with an intuition to facilitate the development of our algorithm.
\paragraph{\textsc{Intuition:}} The jar of marbles will help us again. But now, let us add a special marble to the jar. We pick an integer $J\in\{-N,\dots,N\}$ and set a special probe strength $t=2^J$; this probe reports the indicator that the special marble has survived but at least one marble of type $x$ has not at strength $t$ (thus this is a collection of our previous probes). The probability of this happening is $e^{-t}-e^{-(\f{x}+1)t}$. The average response of this probe over the random choice of $J$ gives a Riemann sum approximation

\begin{equation*}
    \frac1{2N+1}\sum_{J=-N}^N(e^{-2^J}-e^{-(\f{x}+1)2^J})\approx \int_{2^{-N}}^{2^N}\frac{e^{-t}-e^{-(\f{x}+1)t}}{t}\ \mathrm dt
\end{equation*}
because one sensor per multiplicative step corresponds to the measure $\mathrm d(\log t)=\mathrm dt/t$. Letting $N\to \infty$ (one sensor per multiplicative step over all scales) converts the discrete average into the continuous per-log integral; the value produced by this multiplicative-sensor experiment experiment is therefore

\begin{equation*}
    \int_0^\infty\frac{e^{-t}-e^{-(\f{x}+1)t}}{t}\ \mathrm dt=\int_0^\infty(1-e^{-\f{x}t})\frac{e^{-t}}{t}\ \mathrm dt
\end{equation*}
We can again think of this as a race where the special marble is given a clock $S_0\sim \mathrm{Exp}(1)$ and the combined system a clock $S_{\f{x}+1}\sim\mathrm{Exp}(1+\f{x})$. The difference in the expected log-waiting-times for the clocks to go off, namely $\mathbb E[\ln S_0]-\mathbb E[\ln S_{\f{x}+1}]$, measures how many multiplicative time-scales the jar shaves off the waiting time on average; writing that expectation as an integral over log-scales yields precisely the same per-log integrand and this can easily seen to be $\ln(1+\f{x})$. Hence we have the identity

\begin{equation*}
    \int_0^\infty(1-e^{-\f{x}t})\frac{e^{-t}}{t}\ \mathrm dt=\ln(1+\f{x})
\end{equation*}
We now prove this identity formally. We provide an elementary proof of \Cref{lem:SLFAintegralidentity} and in \Cref{sec:deeperreasonsappendix} we discuss about its deeper connections to the Laplace transform.

\begin{lemma}\label{lem:SLFAintegralidentity}
    Let

    \begin{equation*}
        I_{\SLFA}=\int_0^\infty \espoly(t)\frac{e^{-t}}{t}\ \mathrm dt
    \end{equation*}
    Then the following identity holds : $\SLFA=I_{\SLFA}$.
    
\end{lemma}

\begin{proof}
    Let us consider the integral

    \begin{equation*}
        I(x):=\int_0^\infty(1-e^{-xt})\frac{e^{-t}}{t}\ \mathrm dt=\int_0^\infty \frac{e^{-t}-e^{-(x+1)t}}{t}\ \mathrm dt
    \end{equation*}
    Assuming $x$ is a variable here, we can apply Feynman's trick. We differentiate $I(x)$ with respect to $x$ and then bring the derivative under the integral sign.

    \begin{equation*}
        I'(x)=\ddt{f}\int_0^\infty \frac{e^{-t}-e^{-(x+1)t}}{t}\ \mathrm dt=\int_0^\infty\frac{\partial}{\partial x}\left(\frac{e^{-t}-e^{-(x+1)t}}{t}\right)\ \mathrm dt=\int_0^\infty e^{-(x+1)t}\ \mathrm dt=\frac1{x+1}
    \end{equation*}
    Thus $I(x)=\ln(x+1)+C$ where $C$ is a constant with respect to $x$. Now note that $I(0)=0$ and hence $C=0$. Thus $I(x)=\ln(x+1)$. Then

    \begin{align*}
        I_{\SLFA}&=\int_0^\infty \espoly(t)\frac{e^{-t}}{t}\ \mathrm dt=\int_0^\infty \sum_{x\in\F{0}}(1-e^{-\f{x}t})\frac{e^{-t}}{t}\ \mathrm dt=\sum_{x\in\F{0}}\int_0^\infty (1-e^{-\f{x}t})\frac{e^{-t}}{t}\ \mathrm dt\\
        &=\sum_{x\in\F{0}}\ln(\f{x}+1)=\SLFA
    \end{align*}
    
\end{proof}

To aid us in the proof of correctness in a moment, we would also need another lemma.

\begin{lemma}\label{lem:SLFAfderivativebound}
    Let $f(t)=\espoly(t)\frac{e^{-t}}{t}$. Then $|f''(t)|\le |\F{0}|\tau(1+\tau)^2$ for $t\in [0,T]$ for any $T>0$.
\end{lemma}

We defer the proof of \Cref{lem:SLFAfderivativebound} to the appendix, see \Cref{lem:SLFAfderivativeboundappendix}

\paragraph{}We now provide an algorithm to estimate $\SLFA$ called $\estimateSLFA$. \Cref{thm:estimateSLFAcorrectness} then proves the correctness of the algorithm along with the space and update time complexity.

\SetAlgoNoLine%
\begin{algorithm}[htb]
    \DontPrintSemicolon%
    \caption{$\estimateSLFA(\stream = \langle a_1, \dots, a_{\streamlen}\rangle,\error,\conf)$}
    
    \SetKwInOut{Input}{Input}
    \SetKwInOut{Output}{Output}
    
    \Input{stream $\stream$, overall error $\error$, overall failure probability $\conf$\;}
    \Output{$(\error,\conf)$-approximation of the Smoothed Log-Frequency Aggregate ($\SLFA$) of $\stream$\;}
    
    $U\gets \ln\left(\dfrac{4}{\error\ln 2}\right)$\;
    $N \gets \bigg\lceil\sqrt{\dfrac{U^3\maxfreq(1+\maxfreq)^2}{3\error\ln 2}}\bigg\rceil$\;
    $\Delta t=\dfrac{U}{N}$\;
    $c\gets 0$\;
    \BlankLine
    \label{line:boundSLFA}$\mathrm{bound}=\evaluate(\stream, e^{-U}, \error/4, \conf/(N+1))e^{-U}/U+\widehat{m}$\\

    \BlankLine
    
    \tcp{$\hat m$ is an $(\error/4, \conf/(N+1))$-approximation of $m=\F{1}$, for example using a Morris counter \cite{Morris1978}}
    
    \BlankLine
    
    \For{$j \gets 1$ \KwTo $N-1$\label{line:forloopSLFA}}{
        $c_j\gets \evaluate(\stream, e^{-j\cdot\Delta t},\error/4,\conf/(N+1))$\label{line:forloopinsideSLFA}\;
    }
    
    \BlankLine
    
    $c\gets \sum_{j=1}^{N-1}c_j$\;
    $\hat I_{\mathrm{est}}=\Delta t\cdot\left(\frac{\mathrm{bound}}2+c\right)$\;
    \BlankLine
    \label{line:SLFAoutput}\Return $\hat I_{\mathrm{est}}$\;
    
\end{algorithm}

\begin{theorem}\label{thm:estimateSLFAcorrectness}
    The algorithm $\estimateSLFA$ outputs an $(\error,\conf)$-approximation of $\SLFA$. $\estimateSLFA$ takes $\tilde O\left(\maxfreq^3\error^{-5}\ln^2(1/\error)(\log|\univ|+\log\maxfreq+\log(1/\error)+\log(1/\conf))\log^4|\univ|\right)$ space and update time.
\end{theorem}

\begin{proof}
    We start with proving the correctness of the algorithm $\estimateSLFA$, that is, proving that the output of the algorithm $\estimateSR$ is an $(\error,\conf)$ estimate of $\SLFA$.
    
    For the sake of the proof, let us define the function $f(t)=\espoly(t)e^{-t}/t$. Note that $f(0)$ is not defined as is, but since
    \begin{equation*}
        \lim_{t\to 0^+}f(t)=\lim_{t\to 0^+}\espoly(t)\frac{e^{-t}}{t}=\sum_{x\in\F{0}}\lim_{t\to 0^+}(1-e^{-t\f{x}})\frac{e^{-t}}{t}=\sum_{x\in\F{0}}\f{x}=\F{1}=m
    \end{equation*}
    We can extend $f$ to 0 from the right using this value $m$ and call that $f$. Also let $U$ be as defined in the algorithm.
    
    Observe that the output of the algorithm is $\hat I_{\mathrm{est}}$ (Line~\ref{line:SLFAoutput}). In Line~\ref{line:boundSLFA}, the algorithm $\estimateSLFA$ makes a call to the subroutine $\evaluate$ and has a term $\hat m$. From \Cref{cor:espolyestimatecorollary} we obtain that the $\evaluate$ call returns an $(\error/4, \conf/(N+1))$-approximation of $\espoly(U)$ and hence the first term is an $(\error/4, \conf/(N+1))$-approximation of $f(U)$ (let us call it $\hat f(U)$). The term $\hat m$ is an estimate of the stream length $\streamlen=\F{1}$ and is an $(\error/4, \conf/(N+1))$-approximation of $f(0)$ (let us call it $\hat f(0)$). Thus the value of $\mathrm{bound}$ in Line~\ref{line:boundSLFA} is
    
    \begin{equation} \label{eq:boundSLFA}
        \mathrm{bound} = \hat f(U)+\hat f(0)
    \end{equation}
    which is an $(\error/5, \conf/(N+1))$-approximation of $f(U)+f(0)$. 

    Similarly in the \textbf{for} loop in Line~\ref{line:forloopSLFA}-\ref{line:forloopinsideSLFA}, we have $N-1$ many calls to $\evaluate$, and for each $j$, $c_j=\hat f(j\Delta t)$ is an $(\error/5, \conf/(N+1))$-approximation of $f(j\Delta t)$.

    Hence $I_{\mathrm{est}}$ is given by  

    \begin{equation}\label{eq:IestSLFA}
    \hat I_{\mathrm{est}}=\Delta t\left(\frac{\hat f(U)+\hat f(0)}{2}+\sum_{j=1}^{N-1}\hat f(t_j)\right)
    \end{equation}
    Finally, the algorithm outputs $\hat I_{\mathrm{est}}$, while the goal of the algorithm is the estimate $\SLFA$ which is (from \Cref{lem:SLFAintegralidentity}) equal to $I_{\SLFA}$. Thus if we show that with probability at least $(1-\delta)$
    
    \begin{equation}\label{eqn:SLFArequiredbound}
    |I_{\SLFA}- I_{\mathrm{est}}| \leq \error
    \end{equation}
    then the output of the algorithm is an $(\error, \conf)$ of $\F{k}$. We now prove \Cref{eqn:SLFArequiredbound}.
    $I_{\SLFA}$ has no singularity at $0$, so it suffices to break it into two pieces, one a tail piece and the remaining piece we estimate using the trapezoidal rule. Concretely, we do

    \begin{equation*}
        I=\underbrace{\int_0^U(\dots)\ \mathrm dt}_{I_{\mathrm{trap}}}+\underbrace{\int_U^\infty(\dots)\ \mathrm dt}_{I_{\mathrm{tail}}}
    \end{equation*}
    We deal with the two parts separately. In \Cref{lem:SLFAtailbound} we analytically upper bound the value of $I_{\mathrm{tail}}$ and obtain

    \begin{equation}\label{eqn:tailboundSLFA}
        \frac{I_{\mathrm{tail}}}{I}\le \frac{\error}4
    \end{equation}
    Then in \Cref{lem:ItrapIestboundSLFA} we show that the distance between $I_{\mathrm{trap}}$ and $\hat I_{\mathrm{est}}$ is small with respect to $I_{\SR(r)}$; more precisely we show that

    \begin{equation}\label{eqn:IkIestboundSLFA}
        |I_{\mathrm{trap}}-\hat I_{\mathrm{est}}|\le \left(\frac{\error}2+\frac{\error^2}{16}\right)I_{\SR(r)}
    \end{equation}
    Using the triangle inequality and \Cref{eqn:tailboundSLFA} and \Cref{eqn:IkIestboundSLFA}, we obtain

    \begin{align*}
        \frac{|I_{\SLFA}-\hat I_{\mathrm{est}}|}{I_{\SLFA}}&=\frac{|I_{\mathrm{trap}}+I_{\mathrm{tail}}-\hat I_{\mathrm{est}}|}{I_{\SLFA}}\le \frac{|I_{\mathrm{trap}}-\hat I_{\mathrm{est}}|}{I_{\SR(r)}}-\frac{I_{\mathrm{tail}}}{I_{\SR(r)}}\\
        &\le \left(\frac{\error}2+\frac{\error^2}{16}\right)+\frac{\error}4=\frac{3\error}4+\frac{\error^2}{16}\le \error
    \end{align*}
    Thus $I_{\mathrm{est}}$ is an $(\error,\conf)$ estimator of $I_{\mathrm{\SLFA}}$ and hence the output of algorithm $\estimateK$ is an $(\error,\conf)$ estimator of $\SR$.
    
    Now we analyse the space and update time complexity of $\estimateSLFA$. Note that the number of evaluations of our polynomial (which is the number of $\evaluate$ calls) is $O(N)$ which is $O(\maxfreq^{3/2}\error^{-1/2}\ln^{3/2}(1/\error))$. Plugging this in along with the space and update time of one $\evaluate$ call (which is proven in \Cref{cor:espolyestimatecorollary}) gives the space and update time complexity of $\estimateSLFA$ as

    \begin{equation*}
        \tilde O\left(\maxfreq^3\error^{-5}\ln^2(1/\error)(\log|\univ|+\log\maxfreq+\log(1/\error)+\log(1/\conf))\log^4|\univ|\right)
    \end{equation*}
    
 \end{proof}

To completely finish the proof of \Cref{thm:estimateSLFAcorrectness}, we now prove the lemmas used in the proof.

\begin{lemma}\label{lem:SLFAtailbound}
    By choosing $U=\ln\left(\frac4{\error\ln 2}\right)$, we have $I_{\mathrm{tail}}\le (\error/4)I_{\SR(r)}$.
\end{lemma}

\begin{proof}
    Observe that

    \begin{equation*}
        I_{\SLFA}=\sum_{x\in\F{0}}\ln(1+\f{x})\ge\sum_{x\in\F{0}}\ln(1+1)=|\F{0}|\ln2
    \end{equation*}
    Now for $I_{\mathrm{tail}}$, using $(1-e^{-\lambda t})\le 1$ we have

    \begin{equation*}
        I_{\mathrm{tail}}=\int_U^\infty\espoly(t)\frac{e^{-t}}{t}\ \mathrm dt\le |\F{0}|\int_U^\infty \frac{e^{-t}}{t}\ \mathrm dt\le|\F{0}|\frac{e^{-U}}{U}
    \end{equation*}
    Thus using $U=\ln\left(\frac4{\error\ln2}\right)$, we have

    \begin{equation*}
        \frac{I_{\mathrm{tail}}}{I}\le \frac{\error}{4}
    \end{equation*}
\end{proof}

\begin{lemma}\label{lem:ItrapIestboundSLFA}
    By choosing $N=\bigg\lceil\sqrt{\frac{U^3\maxfreq(1+\maxfreq)^2}{3\error\ln 2}}\bigg\rceil$ as in $\estimateSLFA$, we have $|I_{\mathrm{trap}}-\hat I_{\mathrm{est}}|\le \left(\frac{\error}2+\frac{\error^2}{16}\right)I_{\SR(r)}$
\end{lemma}

\begin{proof}
    Recall from \Cref{eqn:IkIestboundSLFA} that

    \begin{equation*}
        \hat I_{\mathrm{est}}=\Delta t\left(\frac{\hat f(U)+\hat f(0)}2+\sum_{j=1}^{N-1}\hat f(t_j)\right)
    \end{equation*}
    Since $\hat f(x)$ is an $(\error/4,\conf/(N+1))$ approximation of $f(x)$ for $x\in \{U\}\cup \{j\Delta t\}_{j\in [N-1]}$, $\hat I_{\mathrm{est}}$ is an $(\error,\conf)$-approximation of

    \begin{equation*}
        \hat I_{\mathrm{trap}}:=\Delta t\left(\frac{f(U)+f(0)}{2}+\sum_{j=1}^{N-1}f(t_j)\right)
    \end{equation*}
    by the union bound. That is, with probability

    \begin{equation}\label{eqn:midestboundSLFA}
        |\hat I_{\mathrm{mid}}-\hat I_{\mathrm{est}}|\le \frac{\error}4\hat I_{\mathrm{mid}}
    \end{equation}
    Also, by the definition of trapezoidal rule in \Cref{subsec:trapezoidalrule}, there exists some $\xi\in (L,U)$ such that

    \begin{equation*}
        |\hat I_{\mathrm{trap}}-I_{\mathrm{\SR}}|=\frac{(U-L)^3}{12N^2}|f''(\xi)|\le \frac{U^3}{12N^2}|f''(\xi)|
    \end{equation*}
    Computation of the upper bound on the absolute value of the double derivative of $f$ from \Cref{lem:SLFAfderivativebound} gives us that

    \begin{equation*}
        |f''(t)|\le r\sum_{x\in\F{0}}\maxfreq(1+\maxfreq)^2
    \end{equation*}
    Hence setting $N=\bigg\lceil\sqrt{\frac{U^3\maxfreq(1+\maxfreq)^2}{3\error\ln 2}}\bigg\rceil$, we have

    \begin{equation}\label{eqn:ItrapboundSLFA}
        \frac{ \bigg|I_{\mathrm{trap}}-\hat I_{\mathrm{trap}}\bigg|}{I_{\SLFA}}\le\frac{U^3}{12N^2}\frac{|\F{0}|\tau(1+\tau)^2}{|\F{0}|\ln 2}\le \frac{\error}{4}\Rightarrow |I_{\mathrm{trap}}-\hat I_{\mathrm{trap}}|\le \frac{\error}4I_{\SLFA}
    \end{equation}
    Then using \Cref{eqn:midestboundSLFA}, \Cref{eqn:ItrapboundSLFA} and the fact that $I_{\mathrm{trap}}\le I_{\SLFA}$, we obtain

    \begin{align*}
        |I_{\mathrm{trap}}-\hat I_{\mathrm{est}}|&=|I_{\mathrm{trap}}-\hat I_{\mathrm{trap}}-\hat I_{\mathrm{trap}}+I_{\mathrm{est}}|\le |I_{\mathrm{trap}}-\hat I_{\mathrm{trap}}|+|\hat I_{\mathrm{trap}}+I_{\mathrm{est}}|\\
        &\le \frac{\error}{4}I_{\SLFA}+\frac{\error}{4}\hat I_{\mathrm{trap}}\\
        &\le \frac{\error}{4}I_{\SLFA}+\frac{\error}{4}\left(1+\frac{\error}{4}\right)I_{\mathrm{\SLFA}}\\
        &=\left(\frac{\error}{2}+\frac{\error^2}{16}\right)I_{\SLFA}
    \end{align*}
\end{proof}

\section{Estimation of a general Bernstein function}\label{sec:generalbernsteinestimationsection}

Here we use the full theory of the L\'evy-Khintchine representation of Bernstein functions from \Cref{subsec:BernsteinfunctionLKrepprelim} to give a general algorithm and theorem. For the sake of completeness, we later talk about what one can practically do if the weight function is not explicitly known in \Cref{subsec:unknowndensity}.

\subsection{Algorithm for when the L\'evy density is explicitly given}

Consider the following algorithm $\estimateGB$ for estimating the statistic of an arbitrary statistic that comes from a Bernstein function with non-negative L\'evy density.

\SetAlgoNoLine%
\begin{algorithm}[htb]
    \DontPrintSemicolon%
    \caption{$\estimateGB(\stream = \langle a_1, \dots, a_{\streamlen}\rangle,\varphi,w, L, U, R_0, R_1, R_2, \error,\conf)$}
    
    \SetKwInOut{Input}{Input}
    \SetKwInOut{Output}{Output}
    
    \Input{
    
    \begin{itemize}
        \item stream $\stream$, Bernstein function $\varphi$, L\'evy density $w$ of $\varphi$
        \item overall error $\error$, overall failure probability $\conf$
        \item lower bound on estimation interval $L$ such that $\int_0^L sw(s)\ \mathrm ds\le \error\varphi(1)/5\maxfreq$
        \item upper bound on estimation interval $U$ such that $\int_U^\infty w(t)\ \mathrm dt\le\error\varphi(1)/5$
        \item $\sup_{[L,U]}|w(t)|=R_0$, $\sup_{[L,U]}|w'(t)|=R_1$, $\sup_{[L,U]}|w''(t)|=R_2$
    \end{itemize}\;}
    
    \BlankLine
    
    \Output{$(\error,\conf)$-approximation of the statistic $\sum_{x\in \F{0}}\varphi(\f{x})$ of $\stream$\;}
    
    \BlankLine
    
    $N \gets \bigg\lceil\sqrt{\dfrac{5U^3(\maxfreq^2R_0+\maxfreq R_1+R_2)}{12\error\varphi(1)}}\bigg\rceil$\;
    $\mathrm{grid}=\dfrac{U-L}{N}$\;
    $c\gets 0$\;
    \BlankLine
    \label{line:boundGB}$\mathrm{bound}=\evaluate(\stream, e^{-U}, \error/5, \conf/(N+1))w(U)+\evaluate(\stream, e^{-L}, \error/5, \conf/(N+1))w(L)$\\
    
    \BlankLine
    
    \For{$j \gets 1$ \KwTo $N-1$\label{line:forloopGB}}{
        $c_j\gets \evaluate(\stream, e^{-(L+j\cdot\mathrm{grid})},\error/5,\conf/(N+1))$\label{line:forloopinsideGB}\;
    }
    
    \BlankLine
    
    $c\gets \sum_{j=1}^{N-1}c_j$\;
    $\hat I_{\mathrm{est}}=\mathrm{grid}\cdot\left(\frac{\mathrm{bound}}2+c\right)$\;
    \BlankLine
    \label{line:GBoutput}\Return $I_{\mathrm{est}}$\;
    
    \end{algorithm}

We can then write down the following general theorem.

\begin{theorem}\label{thm:algorithmestimateGBcorrectnessandspace}
    Let $\varphi$ be a Bernstein function with L\'evy-Khintchine representation (which exists from \Cref{subsec:BernsteinfunctionLKrepprelim})

    \begin{equation*}
        \varphi(s)=\int_0^\infty (1-e^{-ts})w(t)\ \mathrm dt
    \end{equation*}
    where $w(s)\ge 0$ is the L\'evy density of $\varphi$. Let $\GBS=\sum_{x\in\F{0}}\varphi(\f{x})$ be the statistic to be estimated and let $I_{\mathrm{GB}}=\int_0^\infty \espoly(t)w(t)\ \mathrm dt$. Then

    \begin{equation}\label{eqn:generalBernsteinidentity}
        \GBS=I_{\mathrm{GB}}
    \end{equation}
    Moreover, if $S_1(y):=\int_0^y sw(s)\ \mathrm ds<\infty$ for all $y>0$, and $r$ is $C^2$ away from $0$ then $\estimateGB$ outputs an $(\error,\conf)$-approximation of $\GBS$.
\end{theorem}

\begin{proof}
    We first prove the identity in \Cref{eqn:generalBernsteinidentity}. We have

    \begin{align*}
        \GBS=\sum_{x\in\F{0}}\varphi(\f{x})
        =\sum_{x\in\F{0}}\int_0^\infty (1-e^{-t\f{x}})w(t)\ \mathrm dt
        =\int_0^\infty\sum_{x\in\F{0}}(1-e^{-t\f{x}})w(t)\ \mathrm dt
        =\int_0^\infty\espoly(t)w(t)\ \mathrm dt=I_{\mathrm{GB}}
    \end{align*}
    Now we start with the proof of the correctness of the algorithm $\estimateGB$, that is, proving that the output of the algorithm $\estimateGB$ is an $(\error,\conf)$ estimate of $\GBS$. For the sake of the proof, let us define the function $f(t)=\espoly(t)w(t)$. Let $U$ and $L$ be as defined in the algorithm.
    
    Observe that the output of the algorithm is $I_{\mathrm{est}}$ (Line~\ref{line:GBoutput}). In Line~\ref{line:boundGB}, the algorithm $\estimateK$ makes two calls to the subroutine $\evaluate$. From \Cref{cor:espolyestimatecorollary} we obtain that the first $\evaluate$ returns an $(\error/5, \conf/(N+1))$-approximation of $\espoly(U)$ and hence the first term is an $(\error/5, \conf/(N+1))$-approximation of $f(U)$ (let us call it $\hat f(U)$) while the second $\evaluate$ returns an $(\error/5, \conf/(N+1))$-approximation of of $\espoly(L)$ and hence the second term is an $(\error/5, \conf/(N+1))$-approximation of $f(L)$ (let us call this $\hat f(L)$). Thus the value of $\mathrm{bound}$ in Line~\ref{line:boundGB} is
    
    \begin{equation} \label{eq:GBbound}
        \mathrm{bound} = \hat f(U)+\hat f(L)
    \end{equation}
    which is an $(\error/5, 2\conf/(N+1))$-approximation of $f(U)+f(L)$ by an union bound. 

    Similarly in the \textbf{for} loop in Line~\ref{line:forloopGB}-\ref{line:forloopinsideGB}, we have $N-1$ many calls to $\evaluate$, and for each $j$, $c_j=\hat f(L+j\Delta t)$ is an $(\error/5, \conf/(N+1))$-approximation of $f(L+j\Delta t)$.

    Hence $I_{\mathrm{est}}$ is given by  

    \begin{equation}\label{eq:ikestGB}
    \hat I_{\mathrm{est}}=\Delta t\left(\frac{\hat f(U)+\hat f(L)}{2}+\sum_{j=1}^{N-1}\hat f(t_j)\right)
    \end{equation}
    Finally, the algorithm outputs $I_{\mathrm{est}}$, while the goal of the algorithm is the estimate $\GBS$ which is (from \Cref{eqn:generalBernsteinidentity}) equal to 
    $I_{\mathrm{GB}}$. Thus if we show that with probability at least $(1-\delta)$
    
    \begin{equation}\label{eqn:GBrequiredbound}
    |I_{\mathrm{GB}} - I_{\mathrm{est}}| \leq \error I_{\mathrm{GB}}
    \end{equation}
    then the output of the algorithm is an $(\error, \conf)$ of $\GBS$. We now prove \Cref{eqn:GBrequiredbound}. 
    
    Since we have not assumed that $I_{\mathrm{GB}}$ does not have a singularity near $0$, we need to chop off a little part around $0$ first (similar to the proof of \Cref{thm:AlgorithmEstimateKcorrectnessandspace}) and a tail part, and estimate the remaining part. So let us write $I_{\mathrm{GB}}$ as

    \begin{equation*}
        I_{\mathrm{GB}}=\underbrace{\int_0^L (\dots)}_{I_{\mathrm{small}}}+\underbrace{\int_L^U (\dots)}_{I_{\mathrm{trap}}}+\underbrace{\int_U^\infty (\dots)}_{I_{\mathrm{tail}}}
    \end{equation*}
    We deal with the three parts separately. In \Cref{lem:smallpartGB} we analytically upper bound the value of $I_{\mathrm{small}}$ and obtain  

    \begin{equation}\label{eqn:newsmallmainGB}
        I_{\mathrm{small}}\le (\error/5)I_{\mathrm{GB}}
    \end{equation}
    Then in \Cref{lem:tailpartGB} we again analytically upper bound the value of $I_{\mathrm{tail}}$ and obtain
    
    \begin{equation}\label{eqn:newtailmainGB}
        I_{\mathrm{tail}}\le (\error/5)I_{\mathrm{GB}}
    \end{equation}
    Finally in \Cref{lem:EtrapboundGB} we show that the distance between $I_{\mathrm{trap}}$ and $I_{\mathrm{est}}$ is small with respect to $I_{\mathrm{GB}}$; more precisely we show that

    \begin{equation}\label{eqn:newIkIestboundGB}
        |I_{\mathrm{trap}}-I_{\mathrm{est}}|\le \left(\frac{2\error}{5}+\frac{\error^2}{25}\right)I_{\mathrm{GB}}
    \end{equation}
    Using the triangle inequality and \Cref{eqn:newsmallmainGB}, \Cref{eqn:newtailmainGB} and \Cref{eqn:newIkIestboundGB} we obtain 
    
    \begin{align*}
        \frac{\bigg|I_{\mathrm{GB}}-I_{\mathrm{est}}\bigg|}{I_{\mathrm{GB}}}
        &=\frac{|I_{\mathrm{small}}+I_{\mathrm{trap}}+I_{\mathrm{tail}}- I_{\mathrm{est}}|}{I_{\mathrm{GB}}}
        \le \frac{I_{\mathrm{small}}}{I_{\mathrm{GB}}}+\frac{|I_{\mathrm{trap}}- I_{\mathrm{est}}|}{I_{\mathrm{GB}}}+\frac{I_{\mathrm{tail}}}{I_{\mathrm{GB}}}\\
        &\le \frac{\error}{5}+\left(\frac{2\error}{5}+\frac{\error^2}{25}\right)+\frac{\error}{5}\\
        &=\frac{4\error}{5}+\frac{\error^2}{25}
        \le\error
    \end{align*}
    Thus $I_{\mathrm{est}}$ is an $(\error,\conf)$ estimator of $I_{\mathrm{GB}}$ and hence the output of algorithm $\estimateGB$ is an $(\error,\conf)$ estimator of $\GBS$.
    
\end{proof}

The space usage of this algorithm depends on the particulars of the L\'evy density $r$. But at a general level, one can write the space usage and update time of the algorithm as $\tilde O\left(N\error^{-4}(\log(\conf)+\log(N))\log^2|\univ|\right)$ which more explicitly is

\begin{equation*}
    \tilde O\left(\frac{U^{3/2}}{(1-e^{-L})\error^{9/2}}\left(\maxfreq\sqrt{R_0}+\maxfreq^{1/2}\sqrt{R_1}+\sqrt{R_2}\right)\left(\log(|\univ|U\maxfreq)+\log(R_0R_1R_2)+\log\frac1{\error\conf}+\log\frac1{1-e^{-L}}\right)\log^4|\univ|\right).
\end{equation*}

To completely finish the proof of \Cref{thm:algorithmestimateGBcorrectnessandspace}, we now give the proofs of the lemmas used in the proof.

\begin{lemma}\label{lem:smallpartGB}
    By choosing $L$ such that $S_1(L)\le \error\varphi(1)/5\maxfreq$ (as in the algorithm $\estimateGB$), we have $I_{\mathrm{small}}\le (\error/5)I_{\mathrm{GB}}$.
\end{lemma}

\begin{proof}
    Note that such an $L$ is guaranteed to exist as $tw(t)\ge 0$ and thus $S_1(y)\to 0$ as $y\downarrow 0$. Then observe that we can use the fact that $\varphi$ is non-decreasing (since Bernstein functions satisfy $\varphi'(x)\ge 0$ (see \Cref{def:Bernsteinfunction}) and hence are non-decreasing by definition) to conclude
    
    \begin{equation}\label{eqn:GBintegralbound}
        I=\sum_{x\in\F{0}}\varphi(\f{x})\ge |\F{0}|\varphi(1)
    \end{equation}
    Now for $I_{\mathrm{small}}$, using the inequality $1-e^{st}\le st$ for $s,t\ge 0$, we get 
    
    \begin{equation*}
        \espoly(t)=\sum_{x\in\F{0}}(1-e^{-t{\f{x}}})\le \sum_{x\in\F{0}}t\f{x}\le |\F{0}|\maxfreq t
    \end{equation*}
    and hence

    \begin{equation*}
        I_{\mathrm{small}}\le |\F{0}|\tau \int_0^Ltw(t)\ \mathrm dt=|\F{0}|\tau S_1(L).
    \end{equation*}
    By our choice of $L$ and \Cref{eqn:GBintegralbound}, we have $S_1(L)\le\error\varphi(1)/5\maxfreq$ and hence $I_{\mathrm{small}}\le\error I/5$.   
\end{proof}

\begin{lemma}\label{lem:tailpartGB}
    Let $T_0(U):=\int_U^\infty w(t)\ \mathrm dt$. Then by choosing $U$ such that $T_0(U)\le \error\varphi(1)/5$ (as in the algorithm $\estimateGB$), we have $I_{\mathrm{tail}}\le (\error/5)I_{\mathrm{GB}}$.
\end{lemma}

\begin{proof}
    We start by noting that such an $U$ exists as $T_0(y)\to 0$ as $y\to\infty$. Then note that
    
    \begin{equation*}
        I_{\mathrm{tail}}=\int_U^\infty \espoly(t)w(t)\ \mathrm dt\le |\F{0}|\int_U^\infty w(t)\ \mathrm dt=|\F{0}|T_0(U)
    \end{equation*}
    By our choice of $U$ and \Cref{eqn:GBintegralbound}, we have $T_0(U)\le\error\varphi(1)/5$ and hence $I_{\mathrm{tail}}\le \error I/5$.
    
\end{proof}

\begin{lemma}\label{lem:EtrapboundGB}
    
    
    By choosing $N=\bigg\lceil\sqrt{\frac{5U^3(\maxfreq^2R_0+\maxfreq R_1+R_2)}{12\error\varphi(1)}}\bigg\rceil$ as in $\estimateGB$, we have $\bigg|I_{\mathrm{trap}}- I_{\mathrm{est}}\bigg|\le \left(\frac{2\error}5+\frac{\error^2}{25}\right)I_{\mathrm{GB}} $.
\end{lemma}

\begin{proof}
    Recall from \Cref{eq:ikestGB} that

    \begin{equation*}
    I_{\mathrm{est}}=\Delta t\left(\frac{\hat f(U)+\hat f(L)}{2}+\sum_{j=1}^{N-1}\hat f(t_j)\right)
    \end{equation*}
    Since $\hat f(x)$ is an $(\error/5,\conf/(N+1))$ approximation of $f(x)$ for $x\in \{L,U\}\cup \{L+tj\}_{t\in [N-1]}$, $\hat I^{[k]}_{\mathrm{est}}$ is an $(\error,\conf)$-approximation of

    \begin{equation*}
        \hat I_{\mathrm{trap}}:=\Delta t\left(\frac{f(U)+f(L)}{2}+\sum_{j=1}^{N-1}f(t_j)\right)
    \end{equation*}
    by the union bound. That is, with probability at last $1-\conf$
    
    \begin{equation}\label{eqn:boundforhattrapandestGB}
        |\hat I_{\mathrm{trap}}-I_{\mathrm{est}}|\le \frac{\error}5\hat I_{\mathrm{trap}}
    \end{equation}
    Also, by the definition of trapezoidal rule in \Cref{subsec:trapezoidalrule}, there exists some $\xi\in (L,U)$ such that

    \begin{equation*}
        |\hat I_{\mathrm{trap}}-I_{\mathrm{trap}}|=\frac{(U-L)^3}{12N^2}|f''(\xi)|\le \frac{U^3}{12N^2}|f''(\xi)|
    \end{equation*}
    Now since $f''(t)=\espoly''(t)w(t)+2\espoly'(t)w'(t)+\espoly(t)w''(t)$, we get
    
    \begin{equation*}
        |f''(t)|\le |\espoly''(t)w(t)|+2|\espoly'(t)w'(t)|+|\espoly(t)w''(t)|
    \end{equation*}
    and since $|\espoly(t)|\le|\F{0}|, |\espoly'(t)|\le |\F{0}|\maxfreq, |\espoly''(t)|\le |\F{0}|\maxfreq^2$, we have
    
    \begin{equation*}
        |f''(t)|\le |\F{0}|(\tau^2|w(t)|+\tau|w'(t)|+|w''(t)|)
    \end{equation*}
    Since the algorithm sets $R_0=\sup_{[L,U]}|w(t)|, R_1=\sup_{[L,U]}|w'(t)|, R_2=\sup_{[L,U]}|w''(t)|$ (all of which exist and is finite since $r\in C^2([L-\eta,U+\eta])$ for some $\eta$ small enough and $[L-\eta,U+\eta]$ is compact), we have

    \begin{equation*}
        |f''(t)|\le |\F{0}|(\maxfreq^2R_0+\maxfreq R_1+R_2)
    \end{equation*}
    and thus using \Cref{eqn:GBintegralbound} and setting $N=\bigg\lceil\sqrt{\frac{5U^3(\maxfreq^2R_0+\maxfreq R_1+R_2)}{12\error\varphi(1)}}\bigg\rceil$ we have
    
    \begin{equation}\label{eqn:GBEtrapboundinlemma}
        |I_{\mathrm{trap}}-\hat I_{\mathrm{trap}}|\le \error I_{\mathrm{GB}}/5
    \end{equation}
    Then using \Cref{eqn:boundforhattrapandestGB}, \Cref{eqn:GBEtrapboundinlemma} and the fact that $I_{\mathrm{trap}}\le I_{\mathrm{GB}}$, we obtain

    \begin{align*}
        |I_{\mathrm{trap}}-I_{\mathrm{est}}|&=|I_{\mathrm{trap}}-\hat I_{\mathrm{trap}}+\hat I_{\mathrm{trap}}-\hat I_{\mathrm{est}}|\le |I_{\mathrm{trap}}-\hat I_{\mathrm{trap}}|+|\hat I_{\mathrm{trap}}-I_{\mathrm{est}}|\\
        &\le \frac{\error}5I_{\mathrm{GB}}+\frac{\error}5I_{\mathrm{trap}}\\
        &\le \frac{\error}{5}I_{\mathrm{GB}}+\frac{\error}{5}\left(1+\frac{\error}{5}\right)I_{\mathrm{GB}}\\
        &=\left(\frac{2\error}5+\frac{\error^2}{25}\right)I_{\mathrm{GB}} 
    \end{align*}
\end{proof}

\subsection{Remarks on if the L\'evy density is not explicitly known}\label{subsec:unknowndensity}

We add this section for the sake of completeness and keep the exposition brief. If the L\'evy density $w(t)$ is not explicitly given for a Bernstein function $\varphi$, then one has to do a Laplace inversion to get the density. More formally, note that from \Cref{subsec:BernsteinfunctionLKrepprelim}, in particular from \Cref{thm:bernsteinstheorem} and the definition of the complementary Laplace transform \Cref{def:complementarylaplace}, we have

\begin{equation*}
    \varphi(t)=\int_0^\infty(1-e^{-ts})w(s)\ \mathrm ds=\mathcal L^c[w(s)](t)
\end{equation*}
Differentiating both sides with respect to $t$, swapping integral and differentiation (we assume $\int_0^\infty sw(s)\ \mathrm ds<\infty$ here and hence this is allowed [the weaker assumption of $sw(s)$ being locally integrable also suffices]), we get that

\begin{equation*}
    \ddt{t}\mathcal L^c[w(s)](t)=\int_0^\infty se^{-ts}w(s)\ \mathrm ds\Rightarrow w(s)=\frac1s\mathcal L^{-1}\left\{\ddt{t}\mathcal L^c[w(s)](t)\right\}(s)\Rightarrow w(s)=\frac1s\mathcal L^{-1}\left\{\varphi'(t)\right\}(s)
\end{equation*}
So we have proved the following.

\begin{theorem}\label{thm:BernsteininverseLaplace}
    Let $\varphi$ be a Bernstein function with L\'evy-Khintchine representation

    \begin{equation*}
        \varphi(t)=\int_0^\infty (1-e^{-ts})w(s)\ \mathrm ds
    \end{equation*}
    Then we have $w(s)=\frac1s\mathcal L^{-1}\left\{\varphi'(t)\right\}(s)$ and consequently, we can write

    \begin{equation*}
        \varphi(t)=\int_0^\infty (1-e^{-ts})\frac1s\mathcal L^{-1}\left\{\varphi'(t)\right\}(s)\ \mathrm ds
    \end{equation*}    
\end{theorem}

As an example, consider the function $\varphi(t)=\frac{t}{1+t}$. We have $\varphi'(t)=\frac1{(1+t)^2}$ and $\mathcal L^{-1}\{\varphi'(t)\}(s)=se^{-s}$ and hence $w(s)=e^{-s}$, agreeing with our calculation for $\SR$.

This suggests an immediate modification of $\estimateGB$ in the following way: everywhere, we replace evaluations of $w(s)$ with $\frac1s\mathcal L^{-1}\left\{\varphi'(t)\right\}(s)$. For most practical functions, one can analytically compute the Laplace inversion by hand and use the algorithm from the previous section. But in cases where such analytic computations are not possible, one can do numerical inversion, but this loses theoretical guarantees since there is no known way to get $\varepsilon$-multiplicative estimates of a Laplace inversion (unless the function is particularly well looking). For theoretical purposes, there are two main ways to handle this (but again, there are no theoretical guarantees on these methods): 

\begin{enumerate}
    \item One is to use the Bromwich integral to compute the Laplace inversion, which in all its glory reads (in our context)

    \begin{equation*}
        \mathcal L^{-1}\{\varphi'(t)\}(s)=\frac1{2\pi s\iota}\lim_{T\to\infty}\int_{\gamma-\iota T}^{\gamma+\iota T}e^{s\lambda}\varphi'(\lambda)\ \mathrm d\lambda
    \end{equation*}
    where the integration is done along the vertical line $\Re(\lambda)=\gamma$ in the complex plane such that $\gamma$ is greater than the real part of all singularities of $\varphi'$ and $\varphi'$ is bounded on the line, for example if the contour path is in the region of convergence.

    For a proof of this, see Section 5.4 in \cite{Hunt}.

    \item The other is to use Post's inversion formula for Laplace transforms. It requires $\mathcal L[\varphi'(s)](t)$ to be of exponential order, that is, $\sup_{t>0}\frac{\mathcal L^{-1}[\varphi'(s)](t)}{e^{bt}}<\infty$ for some $b\in\mathbb R$. Then
    \begin{equation*}
        \mathcal L^{-1}\{\varphi'(t)\}(s)=\lim_{k\to\infty} \frac{(-1)^k}{k!}\left(\frac{k}{s}\right)^{k+1}\varphi^{(k+1)}(k/s)
    \end{equation*}

    For a proof of this, see \cite{rosehulman}.
\end{enumerate}

But some form of approximations of the Laplace inversions are possible with regularization/mollification, for example, see \cite{maréchal2023regularizationinverselaplacetransform}. The space usage of such algorithms are quite understudied in the literature. We keep it as future work to understand how to efficiently estimate Laplace inversions in the context of set streams.

\section{A complexity-theoretic barrier for lower bounds for $\F{k}$}
\label{sec:complexity-barrier}

In this section we explain why proving lower bounds for unrestricted Delphic $\F{k}$ estimation appears to be difficult. The point is not that the bounded-frequency assumption is known to be necessary. Rather, we show that a sufficiently strong lower bound ruling out polylogarithmic-space and polylogarithmic-update algorithms for unrestricted Delphic $\F{k}$ would imply a threshold-counting time-space separation. This is analogous to the barrier for $\F{0}$ in \cite{NandiVGMP024}, but the oracle needed for $\F{k}$ is naturally a threshold aggregate oracle rather than an existential oracle.

Throughout this section, let $\univ=\{0,1\}^n$, so $n=\log|\univ|$. Let $\stream=\angleb{S_1,\ldots,S_\streamlen}$ be a stream of Delphic sets over $\univ$. For $x\in\univ$, recall that $\f{x}=|\{i\in[\streamlen]:x\in S_i\}|$ and $\F{k}(\stream)=\sum_{x\in\univ}\f{x}^k$. We assume that every input set $S_i$ is given by a binary representation $r_i$, and that membership $x\in S_i$ can be decided from $(r_i,x)$ in time polynomial in $|r_i|+n$ and space linear in $|r_i|+n$.

The ordinary singleton-stream algorithm that we simulate is the stable-sketch algorithm of \cite{indyk2006stable}. We only need its linearity. After finite-precision truncation and discretization, the sketch can be written as follows.

\begin{definition}[Finite-precision linear singleton-stream sketch]
\label{def:finite-precision-linear-sketch}
A finite-precision linear singleton-stream sketch for $\F{k}$ over $\univ$ is specified by a random seed $\sigma$, an integer $d$, integer counters $Y_1,\ldots,Y_d$, and integer-valued functions $w_{\sigma,j}:\univ\to\mathbb Z$ for $j\in[d]$, such that a singleton update $x$ performs $Y_j\leftarrow Y_j+w_{\sigma,j}(x)$ for every $j\in[d]$. At the end, a decoder $\mathsf{Dec}$ outputs a value from $\sigma$ and the counter vector $Y=(Y_1,\ldots,Y_d)$. We assume that each $w_{\sigma,j}(x)$ is computable in time polynomial in the parameter lengths and linear space, and that $|w_{\sigma,j}(x)|\le 2^b$ for some precision parameter $b$.

We say that such a sketch is an $(\error,\conf)$-estimator for singleton-stream $\F{k}$ if, for every singleton stream over $\univ$ with frequency vector $\f{\cdot}$, the output is a $(1\pm\error)$-approximation to $\sum_{x\in\univ}\f{x}^k$ with probability at least $1-\conf$ over the choice of $\sigma$.
\end{definition}

The singleton-stream stable sketch can then be expressed in the following form. We just need a small technicality to write the algorithm cleanly.

Let $Z_k$ denote a standard symmetric $k$-stable random variable, and let
$c_k=\mathrm{median}(|Z_k|^k)$. The standard stable-sketch decoder is
\[\mathsf{StableDecode}_{k}(\sigma,Y_1,\ldots,Y_d)
=
\frac{\mathrm{median}_{j\in[d]} |Y_j|^k}{c_k}.\]
For the finite-precision implementation, $c_k$ is replaced by the corresponding calibration constant of the finite-precision weight distribution. We treat this standard decoder as part of the singleton-stream stable sketch guarantee.

\begin{algorithm}[H]
\caption{$\mathsf{StableSketch}_{k}$: finite-precision singleton-stream sketch}
\label{alg:singleton-stable-sketch}
\KwIn{Singleton stream, moment parameter $0<k\le 2$, error $\error$, confidence $\conf$}
\KwOut{$(\error,\conf)$-estimate of $\F{k}$}
Choose $d=\mathrm{poly}_k(\error^{-1},\log(1/\conf))$\;
Choose a random seed $\sigma$ defining finite-precision $k$-stable weights $w_{\sigma,j}:\univ\to\mathbb Z$ for $j\in[d]$\;
Initialize $Y_1,\ldots,Y_d\leftarrow 0$\;
\For{each singleton update $x$}{
    \For{$j=1$ \KwTo $d$}{
        $Y_j\leftarrow Y_j+w_{\sigma,j}(x)$\;
    }
}
\Return{$\mathsf{StableDecode}_{k}(\sigma,Y_1,\ldots,Y_d)$}\;
\end{algorithm}

The obstruction in the Delphic setting is now visible. If a set $S_i$ arrives, then the update that \Cref{alg:singleton-stable-sketch} would perform on the expanded stream is not a single update, but rather the aggregate update $\sum_{x\in S_i}w_{\sigma,j}(x)$ for every sketch coordinate $j$. We now define an oracle for exactly this operation.

\begin{definition}[Aggregate oracle]
\label{def:aggregate-oracle}
Fix the functions $w_{\sigma,j}$ of a finite-precision linear sketch. The value aggregate oracle $\mathsf{AGG}_w$ is the oracle that, given a Delphic set representation $r$, a seed $\sigma$, and an index $j\in[d]$, returns \[\mathsf{AGG}_w(r,\sigma,j)=\sum_{x\in S_r}w_{\sigma,j}(x).\]
The associated threshold aggregate language is \[\mathsf{THR}\text{-}\mathsf{AGG}_w=\left\{(r,\sigma,j,T):\sum_{x\in S_r}w_{\sigma,j}(x)\ge T\right\}.\]
\end{definition}

Using this oracle, we can write the Delphic version of the stable sketch explicitly.

\begin{algorithm}[H]
\caption{$\mathsf{DelphicStableSketch}_{k}^{\mathsf{AGG}}$}
\label{alg:delphic-stable-sketch-agg}
\KwIn{Delphic set stream $\stream=\angleb{S_1,\ldots,S_\streamlen}$, moment parameter $0<k\le 2$, error $\error$, confidence $\conf$}
\KwOut{$(\error,\conf)$-estimate of $\F{k}(\stream)$}
Choose $d=\mathrm{poly}_k(\error^{-1},\log(1/\conf))$\;
Choose a random seed $\sigma$ defining finite-precision $k$-stable weights $w_{\sigma,j}:\univ\to\mathbb Z$ for $j\in[d]$\;
Initialize $Y_1,\ldots,Y_d\leftarrow 0$\;
\For{each set update $S_i$ with representation $r_i$}{
    \For{$j=1$ \KwTo $d$}{
        $a_{i,j}\leftarrow \mathsf{AGG}_w(r_i,\sigma,j)$\;
        $Y_j\leftarrow Y_j+a_{i,j}$\;
    }
}
\Return{$\mathsf{StableDecode}_{k}(\sigma,Y_1,\ldots,Y_d)$}\;
\end{algorithm}

\begin{lemma}[Exact simulation of the expanded stream]
\label{lem:exact-stable-simulation}
For every fixed seed $\sigma$, \Cref{alg:delphic-stable-sketch-agg} produces exactly the same counter vector as \Cref{alg:singleton-stable-sketch} would produce on the expanded singleton stream obtained by replacing each set $S_i$ with all elements of $S_i$.
\end{lemma}

\begin{proof}
Fix the seed $\sigma$. Consider the expanded singleton stream obtained by replacing each set $S_i$ with all elements of $S_i$. If \Cref{alg:singleton-stable-sketch} were run on this expanded stream, then its $j$-th counter at the end would be \[Y_j^{\mathrm{exp}}=\sum_{i=1}^{\streamlen}\sum_{x\in S_i}w_{\sigma,j}(x)=\sum_{x\in\univ}\f{x}w_{\sigma,j}(x).\]

On the other hand, \Cref{alg:delphic-stable-sketch-agg} does not expand the sets. When the set representation $r_i$ arrives, it queries \[a_{i,j}=\mathsf{AGG}_w(r_i,\sigma,j)=\sum_{x\in S_i}w_{\sigma,j}(x)\] for every $j\in[d]$, and then updates $Y_j\leftarrow Y_j+a_{i,j}$.

Therefore, after the whole stream has been processed, \[Y_j=\sum_{i=1}^{\streamlen}a_{i,j}=\sum_{i=1}^{\streamlen}\sum_{x\in S_i}w_{\sigma,j}(x)=Y_j^{\mathrm{exp}}.\]
This holds for every $j\in[d]$, and hence the whole counter vector is identical to the counter vector obtained by running the singleton-stream sketch on the expanded stream.
\end{proof}

\begin{theorem}[Oracle upper bound for unrestricted Delphic $\F{k}$]
\label{thm:oracle-upper-bound-delphic-fk}
Fix $0<k\le 2$. Suppose we have oracle access to the value aggregate oracle for the finite-precision weights of a standard $k$-stable singleton-stream sketch. Then unrestricted Delphic $\F{k}$ estimation has a one-pass algorithm with space and per-set update time polynomial in $\log|\univ|$, $\log \streamlen$, $\error^{-1}$, and $\log(1/\conf)$, with no bounded-frequency assumption.
\end{theorem}

\begin{proof}
By \Cref{lem:exact-stable-simulation}, for every fixed seed $\sigma$, \Cref{alg:delphic-stable-sketch-agg} produces exactly the same sketch vector as \Cref{alg:singleton-stable-sketch} would produce on the expanded singleton stream. The expanded singleton stream has frequency vector $\f{\cdot}$, and hence its target value is exactly $\F{k}(\stream)=\sum_{x\in\univ}\f{x}^k$. Therefore the correctness guarantee follows from the correctness of the finite-precision stable sketch on singleton streams.

The space used is the space needed to store the seed and the $d$ counters, plus lower-order bookkeeping. The update uses one aggregate query per sketch coordinate. Thus, with oracle access to $\mathsf{AGG}_w$, the update time is polynomial in the sketch dimension and the precision parameters, and hence polynomial in $\log|\univ|$, $\log \streamlen$, $\error^{-1}$, and $\log(1/\conf)$ for the usual finite-precision implementation.
\end{proof}

The theorem above assumes access to the value oracle $\mathsf{AGG}_w$. We next relate this oracle to a threshold-counting class. First observe that since $|S_r|\le 2^n$ and $|w_{\sigma,j}(x)|\le 2^b$, the aggregate value lies in the interval $[-2^{n+b},2^{n+b}]$. Therefore a value query to $\mathsf{AGG}_w$ can be recovered exactly from $O(n+b)$ queries to $\mathsf{THR}\text{-}\mathsf{AGG}_w$ by binary search.

We now define the relevant complexity classes. Let $\#\mathrm{NTISP}(\mathrm{poly},\mathrm{lin})$ be the class of functions $g:\{0,1\}^*\to\mathbb N$ for which there is a nondeterministic Turing machine $M$ such that $g(z)$ is the number of accepting computation paths of $M$ on input $z$, and on inputs of length $N$, the machine $M$ runs in time $N^{O(1)}$ and space $O(N)$. Define $\mathrm{GapNTISP}(\mathrm{poly},\mathrm{lin})$ to be the class of differences $g_1-g_2$, where $g_1,g_2\in\#\mathrm{NTISP}(\mathrm{poly},\mathrm{lin})$.

\begin{definition}[The class $\mathrm{LinPP}$]
\label{def:linpp}
A language $L$ belongs to $\mathrm{LinPP}$ if there exists $h\in\mathrm{GapNTISP}(\mathrm{poly},\mathrm{lin})$ such that, for every input $z$, we have $z\in L$ if and only if $h(z)\ge 0$.
\end{definition}

Thus $\mathrm{LinPP}$ is the polynomial-time, linear-space analogue of a PP or gap-counting threshold class.

\begin{lemma}[Threshold aggregates are in $\mathrm{LinPP}$]
\label{lem:threshold-aggregates-in-linpp}
Assume that membership $x\in S_r$ is decidable in deterministic polynomial time and linear space from $(r,x)$, and that $w_{\sigma,j}(x)$ is computable in deterministic polynomial time and linear space. Then $\mathsf{THR}\text{-}\mathsf{AGG}_w$ belongs to $\mathrm{LinPP}$.
\end{lemma}

\begin{proof}
Write $w_{\sigma,j}(x)=w^+_{\sigma,j}(x)-w^-_{\sigma,j}(x)$, where $w^+_{\sigma,j}(x)=\max(w_{\sigma,j}(x),0)$ and $w^-_{\sigma,j}(x)=\max(-w_{\sigma,j}(x),0)$. Similarly write $T=T^+-T^-$, where $T^+=\max(T,0)$ and $T^-=\max(-T,0)$.

The inequality $\sum_{x\in S_r}w_{\sigma,j}(x)\ge T$ is equivalent to \[\sum_{x\in S_r}w^+_{\sigma,j}(x)+T^- \ge \sum_{x\in S_r}w^-_{\sigma,j}(x)+T^+.\]

Define two functions \[A(r,\sigma,j,T)=\sum_{x\in S_r}w^+_{\sigma,j}(x)+T^-\quad \text{and}\quad B(r,\sigma,j,T)=\sum_{x\in S_r}w^-_{\sigma,j}(x)+T^+.\]

We show that $A$ and $B$ belong to $\#\mathrm{NTISP}(\mathrm{poly},\mathrm{lin})$. For $A$, a nondeterministic machine has two types of accepting branches. In the first type, it guesses $x\in\univ$ and an auxiliary integer $u$, accepts if $x\in S_r$ and $u<w^+_{\sigma,j}(x)$, and rejects otherwise. For each fixed $x\in S_r$, this contributes exactly $w^+_{\sigma,j}(x)$ accepting branches. In the second type, it guesses an auxiliary integer $u$ and accepts if $u<T^-$. This contributes exactly $T^-$ accepting branches. Hence the total number of accepting paths is exactly $A(r,\sigma,j,T)$.

The machine uses only polynomial time and linear space: it stores $x$, $u$, the input, and the workspace needed to decide membership and compute $w_{\sigma,j}(x)$. The bit length of $u$ is at most the bit length of the aggregate bound, which is $O(n+b+\log |T|)$. Therefore this is a $\#\mathrm{NTISP}(\mathrm{poly},\mathrm{lin})$ computation. The same argument gives $B\in\#\mathrm{NTISP}(\mathrm{poly},\mathrm{lin})$.

Thus $A-B\in\mathrm{GapNTISP}(\mathrm{poly},\mathrm{lin})$, and \[(r,\sigma,j,T)\in\mathsf{THR}\text{-}\mathsf{AGG}_w \quad \text{if and only if} \quad A(r,\sigma,j,T)-B(r,\sigma,j,T)\ge 0.\]
Therefore $\mathsf{THR}\text{-}\mathsf{AGG}_w\in\mathrm{LinPP}$.
\end{proof}

\begin{theorem}[Collapse theorem]
\label{thm:linpp-collapse}
Fix $0<k\le 2$. Suppose $\mathrm{LinPP}\subseteq\mathrm{DTISP}(\mathrm{poly},\mathrm{LINSPACE})$. Then unrestricted Delphic $\F{k}$ estimation has a one-pass randomized algorithm with space and per-set update time polynomial in $\log|\univ|$, $\log \streamlen$, $\error^{-1}$, and $\log(1/\conf)$, up to the standard finite-precision parameters of the underlying singleton-stream $k$-stable sketch.
\end{theorem}

\begin{proof}
By \Cref{lem:threshold-aggregates-in-linpp}, the threshold aggregate language $\mathsf{THR}\text{-}\mathsf{AGG}_w$ belongs to $\mathrm{LinPP}$. Under the assumed collapse $\mathrm{LinPP}\subseteq\mathrm{DTISP}(\mathrm{poly},\mathrm{LINSPACE})$, this threshold language can be decided deterministically in polynomial time and linear space.

A value aggregate $\mathsf{AGG}_w(r,\sigma,j)$ lies in $[-2^{n+b},2^{n+b}]$, and hence can be recovered exactly using $O(n+b)$ threshold queries by binary search. Each such query has length polynomial in $|r|$, $|\sigma|$, $n$, $b$, and $\log(1/\error)+\log(1/\conf)+\log \streamlen$. Therefore, under the collapse assumption, each aggregate value can be computed in polynomial time and linear space.

Now run \Cref{alg:delphic-stable-sketch-agg}, computing each call to $\mathsf{AGG}_w$ by this deterministic simulation. By \Cref{lem:exact-stable-simulation}, the resulting algorithm computes exactly the same sketch vector as \Cref{alg:singleton-stable-sketch} would compute on the expanded singleton stream. Hence it inherits the stable-sketch correctness guarantee. Since the number of sketch coordinates and the precision are polynomial in the usual parameters, and since the set representations have length polynomial in $\log|\univ|$ for the Delphic families considered here, the resulting space and per-set update time are polynomial in $\log|\univ|$, $\log \streamlen$, $\error^{-1}$, and $\log(1/\conf)$.
\end{proof}

\begin{corollary}[Lower bounds imply a threshold-counting separation]
\label{cor:lower-bound-implies-linpp-separation}
Fix $0<k\le 2$. Suppose one proves that unrestricted Delphic $\F{k}$ estimation cannot be solved by any one-pass randomized streaming algorithm with space and per-set update time polynomial in $\log|\univ|$, $\log \streamlen$, $\error^{-1}$, and $\log(1/\conf)$, for the same representation class of Delphic sets. Then \[\mathrm{LinPP}\not\subseteq\mathrm{DTISP}(\mathrm{poly},\mathrm{LINSPACE}).\]
\end{corollary}

\begin{proof}
Assume, toward contradiction, that $\mathrm{LinPP}\subseteq\mathrm{DTISP}(\mathrm{poly},\mathrm{LINSPACE})$. Then \Cref{thm:linpp-collapse} gives a one-pass randomized algorithm for unrestricted Delphic $\F{k}$ estimation with space and per-set update time polynomial in $\log|\univ|$, $\log \streamlen$, $\error^{-1}$, and $\log(1/\conf)$. This contradicts the assumed lower bound. Therefore $\mathrm{LinPP}\not\subseteq\mathrm{DTISP}(\mathrm{poly},\mathrm{LINSPACE})$.
\end{proof}

\section{Conclusion}

In this paper, we introduced a novel framework for estimating non-integer frequency moments $\F{k}$ for $k \in (0,1)$ and related Bernstein-type statistics in the Delphic set-stream model under bounded frequency assumptions. Motivated by practical applications such as sensor coverage in IoT systems and Klee's measure problems, we addressed the longstanding challenge of achieving low space and update times without enumerating large sets. Our core insight leverages multi-rate sampling to construct the Expected Support Polynomial (ESP), which we interpret as a complementary Laplace transform of the frequency distribution measure. By reducing the problem to distinct-count ($\F{0}$) estimations on sampled substreams and employing controlled numerical integration (via the trapezoidal method), we achieve the first one-pass algorithms with space and update time polynomial in $\maxfreq$ and $\error^{-1}$. In the regime where $\maxfreq = \mathrm{polylog}(|\univ|, \streamlen)$ (where $\streamlen$ is the stream length), this yields polylogarithmic bounds in the universe size and stream length, marking a significant advancement over prior techniques that either required per-element operations or failed to handle the Delphic model's update-time constraints.

\paragraph{}

Complementary to the algorithms, we also showed that showing an unconditional (that is, independent of $\tau$) lower bound of polylog resources seems to be hard as it implies a linear-space threshold-counting complexity class separation.

\paragraph{}
Our approach unifies the estimation of a broad class of statistics, including Saturated Richness (SR), Smoothed Log-Frequency Aggregate (SLFA), and a class of Bernstein functions with non-negative Lévy densities. Rather than treating the underlying $F_0$ estimators as strict black boxes, our reductions deeply integrate their structural properties - specifically their precise sampling complexities and structure of querying the sets - to rigorously simulate subsampling from the stream with independence at the element level. We provided rigorous analyses of approximation guarantees, including error bounds from discretization and noise. Empirical motivations, such as real-time monitoring of coverage redundancy, underscore the practical relevance of our results, particularly in bounded-frequency settings arising from rate limits or physical constraints.
\paragraph{}
While our work establishes the feasibility of efficient estimation in this model, several intriguing directions remain for future exploration and advancements; particularly in handling unbounded frequencies and tighter complexities. Below, we outline several open problems and directions for future research.

\begin{enumerate}

\item \emph{Removing the Bounded-Frequency Assumption and Lower bounds}: Can we design algorithms for estimating $\F{k}$ for $k \in (0,1)$ in general Delphic set streams without assuming $\maxfreq = O(\mathrm{polylog}(|\univ|, \streamlen)$? What are the fundamental space lower bounds in the unbounded-frequency case? Can one possibly use communication complexity or information-theoretic arguments tailored to set oracles? This question is now particularly tantalizing in view of the complexity class separation reduction.

\item \emph{Improving the dependence on $\tau$ and $\varepsilon$}: Can we improve the dependency on $\maxfreq$ and $\error$ for $\F{k}$ when $k\in (0,1)$ and other statistics? Maybe by handling the singularities better or by using targeted quadratures for the particular weight functions? Perhaps using adaptive quadrature or machine learning-based approximations in streaming contexts?


\item \emph{Connections to Other Transforms}: Is it possible to use other integral transforms (e.g., Fourier or Mellin) for estimating frequency-based statistics? And identifying the classes those transforms apply to?

\item \emph{Developing efficient algorithms for estimating the inverse Laplace transform in set streams}: As hinted in our discussion of Bernstein representations, numerical Laplace inversion currently do not have much theoretical guarantees. Can we come up with efficient streaming friendly inversion algorithms that are approximately stable to handle a wider class of integral transforms?

\item \emph{Distributed streaming}: An interesting line of work would be to combine our approach with distributed streaming models (e.g., multi-party or coordinator models) for Delphic sets, addressing scalability in massive datasets.

\end{enumerate}



\newpage
\begingroup
\bibliographystyle{alpha}
\bibliography{references.bib}
\endgroup

\newpage

\appendix

\section{Some proofs left from the main body}

\subsection{Proof of \Cref{lem:expectedpoly}}\label{subsec:expectedpoly}

We restate the lemma here for convenience.

\begin{lemma}
If $\mathsf Y$ be the $\F{0}$ moment of the sub-stream $\stream_{\throwprob}$ that is obtained by \emph{element-level subsampling}: for each set $S_i$ in $\stream$ and each element $x\in S_i$, independently retaining $x$ in $S_i$ with probability $(1-\throwprob)$, then 
    \begin{equation*}
        \mathbb{E}[\mathsf Y] = \sum_{x\in \univ}(1-\throwprob^{\freq{x}})
    \end{equation*}
where $\f{x}(\stream)$ is the frequency of $x$ in the stream $\stream$.
\end{lemma}

\begin{proof}
    Let $\mathbbm 1_x$ be the indicator random variable that is $1$ if $x$ appears in the thinned sub-stream $\stream_{\throwprob}$ and $0$ otherwise. Under element-level subsampling, $x$ has $\freq{x}$ independent coins (one for each set containing $x$), each landing ``keep'' with probability $(1-\throwprob)$. Thus $x$ fails to appear if and only if all $\freq{x}$ coins land ``discard'', so
    \begin{align*}
    \mathbb{P}(\mathbbm 1_x=1) &= 1 - \mathbb{P}(\text{all } \freq{x} \text{ coins discard}) 
    = 1 - \throwprob^{\freq{x}}
    \end{align*} 
    and thus $\mathbbm 1_x\sim\Ber(p_x)$ with $p_x=1-\throwprob^{\freq{x}}$. Thus 
    $\mathsf Y = \F{0}(\stream_{\throwprob}) = \sum_{x\in \F{0}(S)} \mathbbm 1_x.$
    By the linearity of expectation, this implies that\footnote{We may sometimes abuse notation to write $\F{0}(\stream)$ to denote the set of all distinct elements of the stream $\stream$ as well as the size of this set when there is no confusion.} 
    $\mathbb E[\mathsf Y]=\sum_{x\in \F{0}(\stream)}(1-\throwprob^{\f{x}}).$ Now if $x\notin \F{0}(\stream)$, then $\f{x}=0$ and hence\footnote{We take the convention that $0^0=1$.} $(1-\throwprob^{\f{x}})=0$ and thus we can write
    \begin{equation*}
        \mathsf Y=\sum_{x\in \F{0}(\stream)}(1-\throwprob^{\f{x}})+0=\sum_{x\in \univ}(1-\throwprob^{\f{x}})
    \end{equation*}
\end{proof}

\subsection{Proof of \Cref{thm:delphicsimulation}}\label{subsec:simulationappendix}

We restate the lemma for convenience.

\begin{theorem}[Simulation of element-level subsampling via Delphic oracles]\label{thm:delphicsimulationappendix}
Let $S\subseteq\univ$ be a Delphic set with $|S|=n$, and let $p=1-\throwprob\in[0,1]$. Consider the element-level subsampled set $S'=\{x\in S: B_x=1\}$ where $\{B_x\}_{x\in S}$ are independent $\Ber(p)$ random variables. This is the true model. Let the simulated model be defined as above.

Let $\mathcal{A}$ be a (possibly randomized, adaptive) algorithm that interacts with $S'$ through the above three oracle types and satisfies the following two properties:
\begin{enumerate}
    \item[(P1)] $\mathcal{A}$ makes at most one membership query per element $x\in\univ$, and asks at most $T$ many membership queries.
    \item[(P2)] $\mathcal{A}$ completes all membership queries before invoking the size oracle, (possibly) invokes the size oracle before any sampling query, asks at most $T$ sampling queries in total, and the number of distinct samples obtained is at most the answer of the size oracle.
\end{enumerate}
Then:
\begin{enumerate}
    \item[(i)] The joint distribution of the membership answers and the size $K$ observed by $\mathcal A$ under the simulated oracles is identical to that under the true oracles of $S'$.
    \item[(ii)] Conditional on the membership answers and the size query answer, the joint distribution of the entire sequence of answers from the simulated sampling oracle is identical to that obtained by repeated uniform sampling from the fixed true set $S'$. In particular, any procedure for obtaining distinct samples by rejecting duplicates has the same distribution in the two models, and the simulation never returns more distinct elements than the size oracle answer.
    \item[(iii)] The simulated membership and size oracles each require exactly one Delphic oracle call on $S$ plus $O(1)$ additional work, taking $O(\log|\univ|)$ time. Moreover, for any $\delta\in(0,1)$, over the entire sampling phase the simulated sampling oracle uses
    \[ O\!\left(T\left(\log T+\log\frac1\delta\right)\right) \]
    Delphic sampling calls with probability at least $1-\delta$. Consequently, the entire sampling phase takes
    \[ O\!\left(T\left(\log T+\log\frac1\delta\right)\log|\univ|\right) \]
    time with probability at least $1-\delta$. It uses $O(T\log|\univ|)$ bits for storing $Q$ and $D$, up to lower-order counters.
\end{enumerate}
\end{theorem}

To prove this, we need a small lemma first.

\begin{lemma}[Binomial decomposition]\label{lem:binomialdecompositionappendix}
    Let $n\in\mathbb{N}$, $p\in [0,1]$, and let $I\subseteq [n]$ with $|I|=r$. Suppose $\{B_i\}_{i\in [n]}$ are independent $\Ber(p)$ random variables, and let $K=\sum_{i=1}^n B_i\sim\mathrm{Bin}(n,p)$. Then, conditional on the values $\{B_i\}_{i\in I}$, the remaining sum $\sum_{i\notin I}B_i$ is independent of $\{B_i\}_{i\in I}$ and satisfies
    \begin{equation*}
        \sum_{i\notin I}B_i \sim \mathrm{Bin}(n-r,\, p).
    \end{equation*}
    Consequently, letting $r_+=\sum_{i\in I}B_i$, we have
    \begin{equation*}
        K \;\Big|\; \{B_i\}_{i\in I} \;=\; r_+ + \mathrm{Bin}(n-r,\, p)
    \end{equation*}
    and in particular $K\mid \{B_i\}_{i\in I}$ depends on $\{B_i\}_{i\in I}$ only through $r_+$.
\end{lemma}

\begin{proof}
    Since the $\{B_i\}_{i\in[n]}$ are mutually independent, the sub-collections $\{B_i\}_{i\in I}$ and $\{B_i\}_{i\notin I}$ are independent. Hence conditioning on any realization of $\{B_i\}_{i\in I}$ does not affect the joint distribution of $\{B_i\}_{i\notin I}$, which remains a product of independent $\Ber(p)$ variables. Their sum $\sum_{i\notin I}B_i$ is therefore $\mathrm{Bin}(n-r,p)$, independent of $\{B_i\}_{i\in I}$. Since $K=\sum_{i\in I}B_i+\sum_{i\notin I}B_i=r_++\sum_{i\notin I}B_i$, the result follows.
\end{proof}

Now we give the proof of \Cref{thm:delphicsimulationappendix}.

\begin{proof}[Proof of \Cref{thm:delphicsimulation}]
    \textbf{(i): Joint distribution of membership answers and size.}

    Fix an arbitrary adaptive strategy of $\mathcal{A}$ that satisfies (P1) and (P2). Suppose $\mathcal{A}$ queries the membership of elements $x_1,\dots,x_r\in S$ (the queries for elements not in $S$ return \textsc{No} deterministically in both models, so we ignore them). 
    
    \emph{True model.} For each $x_i\in S$, the membership answer is $B_{x_i}\sim\Ber(p)$, and the answers $B_{x_1},\dots,B_{x_r}$ are mutually independent (since these are distinct elements by (P1)). After observing these answers, the true size of $S'$ is $|S'|=\sum_{x\in S}B_x$. By \Cref{lem:binomialdecompositionappendix} applied with $I=\{x_1,\dots,x_r\}$, we have
    \begin{equation}\label{eq:trueconditionalsize}
        |S'|\;\Big|\; B_{x_1},\dots,B_{x_r} \;=\; r_+ + \mathrm{Bin}(n-r,\,p)
    \end{equation}
    where $r_+=\sum_{i=1}^r B_{x_i}$, and the $\mathrm{Bin}(n-r,p)$ term is independent of $B_{x_1},\dots,B_{x_r}$.
    
    \emph{Simulated model.} The membership answers are $C_{x_1},\dots,C_{x_r}$, which are independent $\Ber(p)$ by construction. The size returned is $K=r_++K'$ with $K'\sim\mathrm{Bin}(n-r,p)$ independent of $C_{x_1},\dots,C_{x_r}$.
    
    The joint distributions thus satisfy
    \begin{equation*}
        (C_{x_1},\dots,C_{x_r},\,K) \;\stackrel{d}{=}\; (B_{x_1},\dots,B_{x_r},\, |S'|)
    \end{equation*}
    since in both cases the first $r$ components are i.i.d.\ $\Ber(p)$ and the size, conditioned on these components, has the same distribution $r_++\mathrm{Bin}(n-r,p)$ by \eqref{eq:trueconditionalsize}. This holds for every adaptive choice of which elements to query (since the choice depends only on prior answers, and both models produce the same distribution of prior answers by induction on the number of queries).

    \textbf{(ii): Sampling oracle.}

    We prove this by evaluating the conditional distributions in each model separately and showing they are mathematically identical. Let $Q$ be the set of membership queries asked.

    \emph{True Model:} In the true model, the thinned set $S'$ is formed upfront via the independent coins $\{B_x\}_{x \in S}$. When $\mathcal{A}$ makes its sequence of membership queries $Q$ (with $r=|Q|$), the oracle simply reveals the realized values of $\{B_x\}_{x \in Q}$. When $\mathcal{A}$ queries the size, it observes $|S'| = \sum_{x \in S} B_x$. To evaluate the distribution of a subsequent sample, we must condition on the exact information transcript observed by $\mathcal{A}$ up to this point. We condition on the queries $Q$ and their answers $\{B_x\}_{x \in Q}$, which fixes the subset of queried survivors $Q^* = \{x \in Q : B_x = 1\}$ and its size $r_+ = |Q^*|$. Since $|S'| = r_+ + \sum_{x \in S \setminus Q} B_x$, conditioning on the observed size $|S'|$ is equivalent to further conditioning on the number of unqueried survivors $K' = \sum_{x \in S \setminus Q} B_x$. Let $\mathcal{E}_{\mathrm{True}}$ denote this joint conditioning on the observed transcript. If $K = r_+ + K' = 0$, the true sampling oracle deterministically returns $\bot$. Assume $K \ge 1$.

    The true thinned set is $S' = Q^* \cup (S' \setminus Q)$. The true sampling oracle draws uniformly from $S'$. For any $s \in S$, the conditional probability of outputting $s$ is:
    \begin{itemize}
    \item \textbf{Case 1 ($s \in Q^*$):} Conditioned on $\mathcal{E}_{\mathrm{True}}$, $s$ is deterministically in $S'$. Since $|S'| = K$, we have $\mathbb{P}_{\mathrm{True}}(\text{output } s \mid \mathcal{E}_{\mathrm{True}}) = \frac{1}{K}$.
    \item \textbf{Case 2 ($s \in S \setminus Q$):} By \Cref{lem:binomialdecompositionappendix}, conditioned on $\mathcal{E}_{\mathrm{True}}$, the unqueried survivors $S' \setminus Q$ form a uniformly random subset of $S \setminus Q$ of size $K'$. By symmetry, the probability that a specific $s \in S \setminus Q$ is included in $S'$ is $\frac{K'}{n-r}$. If included, it is sampled with probability $\frac{1}{K}$. Thus, $\mathbb{P}_{\mathrm{True}}(\text{output } s \mid \mathcal{E}_{\mathrm{True}}) = \frac{K'}{n-r} \cdot \frac{1}{K}$.
    \item \textbf{Case 3 ($s \in Q \setminus Q^*$):} Conditioned on $\mathcal{E}_{\mathrm{True}}$, $s \notin S'$, so $\mathbb{P}_{\mathrm{True}}(\text{output } s \mid \mathcal{E}_{\mathrm{True}}) = 0$.
    \end{itemize}

    \emph{Simulated Model and repeated samples:} Let $D\subseteq S\setminus Q$ be the set of distinct elements from the unqueried part that have been revealed by previous sampling queries; initially $D=\varnothing$. We maintain the invariant that, conditioned on the entire transcript so far, the still-unrevealed unqueried survivors form a uniformly random $(K'-|D|)$-subset of $S\setminus(Q\cup D)$. This holds initially by \Cref{lem:binomialdecompositionappendix} and symmetry.

    Suppose the invariant holds and write $d=|D|$. In the true model, every $s\in Q^*\cup D$ is returned by the next uniform sample with probability $1/K$, while for every $s\in S\setminus(Q\cup D)$,
    \[
        \mathbb{P}_{\mathrm{True}}(\text{output }s\mid\text{transcript})
        =
        \frac{K'-d}{n-r-d}\cdot\frac1K.
    \]
    The simulated oracle gives exactly these probabilities: it chooses $Q_{\mathrm{yes}}\cup D$ with probability $(r_++d)/K$ and otherwise chooses $S\setminus(Q\cup D)$ with probability $(K'-d)/K$, sampling uniformly from the chosen pool. If a previously unseen element $s$ is returned, adding $s$ to $D$ preserves the invariant by symmetry; if an already revealed element is returned, the invariant is unchanged. Hence, by induction, the complete sequence of sampling-oracle answers has exactly the same joint distribution in the true and simulated models.

    

    \textbf{(iii): Complexity.}

    The membership oracle makes one Delphic membership query on $S$, taking $O(\log|\univ|)$ time, plus $O(1)$ additional work. The size oracle makes one Delphic size query on $S$ and samples one binomial random variable, so it also takes $O(\log|\univ|)$ time up to constant additional work.
    
    We now consider the entire sampling phase. Each sampling query that returns an element already contained in $Q_{\mathrm{yes}}\cup D$ takes $O(1)$ additional work. Rejection sampling from the original set $S$ is needed only when the oracle exposes a previously unseen element from $S\setminus(Q\cup D)$; whenever this happens, the newly exposed element is added to $D$. Since $\mathcal A$ asks at most $T$ sampling queries, at most $T$ such new elements can ever be exposed.
    
    Let $r=|Q|$. Suppose first that $n\ge 4T$. Throughout the sampling phase we have $r+|D|\le 2T$, and therefore a draw from $S$ is accepted with probability
    \[ \frac{n-r-|D|}{n}\ge \frac12. \]
    Thus the total number of Delphic samples needed to expose all new elements during the sampling phase is
    \[ O\!\left(T+\log\frac1\delta\right) \]
    with probability at least $1-\delta$.
    
    Now suppose that $n<4T$. The number of Delphic samples required to expose all elements that can possibly enter $D$ is dominated by the time needed to see all elements of $S\setminus Q$ by uniform sampling from $S$. For any $q\ge 0$, the probability that some element of $S\setminus Q$ has not appeared after $q$ draws is at most
    \[ n e^{-q/n}. \]
    Hence, with probability at least $1-\delta$, all such elements have been seen after
    \[ O\!\left(n\left(\log n+\log\frac1\delta\right)\right) = O\!\left(T\left(\log T+\log\frac1\delta\right)\right) \]
    Delphic samples.
    
    Combining the two cases, over the entire sampling phase the simulator uses
    \[ O\!\left(T\left(\log T+\log\frac1\delta\right)\right) \]
    Delphic sampling calls with probability at least $1-\delta$. Since every Delphic sampling call takes $O(\log|\univ|)$ time, the entire sampling phase takes
    \[ O\!\left(T\left(\log T+\log\frac1\delta\right)\log|\univ|\right) \]
    time with probability at least $1-\delta$.
    
    Finally, $|Q|\le T$ and $|D|\le T$, so the simulator stores at most
    $O(T)$ universe elements, requiring $O(T\log|\univ|)$ bits up to
    lower-order counters.
    
\end{proof}

\subsection{Proof of convergence and evaulation in \Cref{lem:Fkidentityintegral}}

We restate the result again for convenience.

\begin{lemma}\label{lem:Fkintegralidentityandconveregence}
    Let $k\in(0,1)$. Consider the integral
    
    \begin{equation*}
        I(x):=\int_0^\infty (1-e^{tx})t^{-k-1}\ \mathrm dt
    \end{equation*}
    Then the integral $I(x)$ converges and equals $x^k\frac{\Gamma(1-k)}{k}$.
\end{lemma}

\begin{proof}
    We prove convergence first. Split the integral at $t=1$. So

    \begin{equation*}
        I(x)=\int_0^1(1-e^{-tx})t^{-k-1}\ \mathrm dt+\int_1^\infty (1-e^{-tx})t^{-k-1}\ \mathrm dt
    \end{equation*}
    For $0<t\le 1$, we have $1-e^{-tx}\le tx$ and thus $(1-e^{-tx})t^{-k-1}\le xt^{-k}$. Hence

    \begin{equation*}
        \int_0^1(1-e^{-tx})t^{-k-1}\ \mathrm dt\le x\int_0^1t^{-k}\ \mathrm dt=\frac{x}{1-k}<\infty
    \end{equation*}
    For $t\ge 1$, we have $0\le 1-e^{-tx}\le 1$ and hence $(1-e^{tx})t^{-k-1}\le t^{-k-1}$ and hence

    \begin{equation*}
        \int_1^\infty(1-e^{-tx})t^{-k-1}\ \mathrm dt\le \int_1^\infty t^{-k-1}\ \mathrm dt=\frac1k<\infty
    \end{equation*}
    and thus $I(x)<\infty$ and hence converges.

    \paragraph{} Now we need to compute the value of the integral. Let us do a change of variables to $u=tf$ and use $1-e^{-u}=u\int_0^1e^{-us}\ \mathrm ds$ to get

    \begin{equation*}
        I(x)=x^k\int_0^\infty (1-e^{-u})u^{-k-1}\ \mathrm du=x^k\int_0^\infty\int_0^1e^{-us}u^{-k}\ \mathrm ds\mathrm du
    \end{equation*}
    Since $u^{-k}e^{-us}$ is non-negative on $(u,s)\in (0,\infty)\times (0,1)$, Tonelli's theorem applies and we can interchange the order of integration to get

    \begin{equation*}
        I(x)=x^k\int_0^1\int_0^\infty e^{-us}u^{-k}\ \mathrm du\mathrm ds
    \end{equation*}
    Letting $v=us$ and noting that $1-k>0$, we can write the inner integral as $\int_0^\infty e^{-us}u^{-k}\ \mathrm du=s^{k-1}\int_0^\infty e^{-v}v^{-k}\ \mathrm dv=s^{k-1}\Gamma(1-k)$, so we have

    \begin{equation*}
        I(x)=x^k\Gamma(1-k)\int_0^1s^{k-1}\ \mathrm ds=x^k\Gamma(1-k)\frac1k
    \end{equation*}
\end{proof}

\subsection{Bound used in \Cref{thm:AlgorithmEstimateKcorrectnessandspace}}

\begin{lemma}\label{lem:doublederivativebound}
    Let $f(t)=\espoly(t)t^{-k-1}$ with $k\in (0,1)$. Then

    \begin{equation*}
        |f''(t)|\le \frac1{L^{k+3}}\sum_{x\in\F{0}}\left[(k+1)(k+2)+\frac{2(k+1)}e+\frac4{e^2}\right]
    \end{equation*}
\end{lemma}

To prove this, we need another result.

\begin{lemma}\label{lem:expbounds}
    For any $a\ge 0$, the following bounds hold

    \begin{equation*}
        \sup_{x\ge 0}xe^{-ax}=\frac1{ae},\ \sup_{x\ge 0}x^2e^{-ax}=\frac4{e^2a^2}
    \end{equation*}
\end{lemma}

\begin{proof}
    Let $f(x)=xe^{-ax}$. Then the derivative of the function is $f'(x)=e^{-ax}(1-ax)$ and this vanishes at $x=1/a$. Since $f(0)=0, \lim_{x\to\infty}f(x)=0$ and $f(1/a)=1/ea$, thus $\sup_{x\ge 0}xe^{-ax}=1/ea$.

    Similarly, let $h(x)=x^2e^{-ax}$. Then the derivative of the function is $h'(x)=e^{-ax}(2x-ax^2)$ and this vanishes at $x=0$ and $x=2/a$. Now since $h(0)=0,\lim_{x\to 0}h(x)=0$ and $h(2/a)=4/e^2a^2$, thus $\sup_{x\ge 0}x^2e^{-ax}=4/e^2a^2$.
\end{proof}

\begin{proof}[Proof of \Cref{lem:doublederivativebound}]
    Let us first compute $f''(t)$. We have

    \begin{equation*}
        f''(t)=\frac{\espoly''(t)}{t^{k+1}}-2(k+1)\frac{\espoly'(t)}{t^{k+2}}+(k+1)(k+2)\frac{\espoly(t)}{t^{k+3}}
    \end{equation*}
    and hence

    \begin{align*}
        &|f''(t)|\le \left|\frac{\espoly''(t)}{t^{k+1}}\right|+2(k+1)\left|\frac{\espoly'(t)}{t^{k+2}}\right|+(k+1)(k+2)\left|\frac{\espoly(t)}{t^{k+3}}\right|\\
        &\le \frac1{t^{k+3}}\sum_{x\in\F{0}}\left[|\f{x}^2e^{-\f{x}t}t^2|+2(k+1)|\f{x}e^{-\f{x}}t|+(k+1)(k+2)|1-e^{-\f{x}}|\right]\\
        &\le\frac{1}{L^{k+3}}\sum_{x\in\F{0}}\left[\f{x}^2e^{-\f{x}t}t^2+2(k+1)\f{x}e^{-\f{x}}t+(k+1)(k+2)|\right] \qquad (\text{Since $|1-e^{-t\f{x}}|\le 1$ and $\f{x},e^{-\f{x}},t\ge 0$})\\
        &\le \frac1{L^{k+3}}\sum_{x\in\F{0}}\left[\frac4{e^2t^2}t^2+2(k+1)\frac1{et}t+(k+1)(k+2)\right] \qquad (\text{From \Cref{lem:expbounds} with $a=t$ and $x=\f{x}$})\\
        &=\frac1{L^{k+3}}\sum_{x\in\F{0}}\left[\frac4{e^2}+\frac2e(k+1)+(k+1)(k+2)\right]
    \end{align*}
\end{proof}

\subsection{Bound used in \Cref{lem:tailpartFk}}

\begin{lemma}\label{lem:Gammabound}
    For every $x\in(0,1)$ we have $\displaystyle \Gamma(x)\ge\frac{1}{2x}$.
\end{lemma}

\begin{proof}
    Let $y=x+1\in(1,2)$. It suffices to show $\Gamma(y)\ge\frac12$ for $y\in(1,2)$ because $x\Gamma(x)=\Gamma(x+1)=\Gamma(y)$, so 
    
    \begin{equation*}
        \Gamma(x)\ge\frac{1}{2x}\iff 2\Gamma(y)\ge 1
    \end{equation*}
    The Gamma function is convex on $(0,\infty)$, hence for any point $a$ we have $\Gamma(y)\ge L_a(y)$ where $L_a$ is the tangent line at $a$. Take $a=2$. Then 
    
    \begin{equation*}
        L_2(y)=\Gamma(2)+\Gamma'(2)(y-2)=1+(1-\gamma)(y-2)
    \end{equation*}
    because $\Gamma(2)=1$ and $\Gamma'(2)=\Gamma(2)\psi(2)=1\cdot(1-\gamma)$ where $\psi$ is the digamma function and $\psi(2)=H_1-\gamma=1-\gamma$. For $y\in[1,2]$ the line $L_2$ achieves its minimum at $y=1$, and 
    
    \begin{equation*}
        L_2(1)=1+(1-\gamma)(-1)=\gamma>1/2
    \end{equation*}
    Therefore for all $y\in[1,2]$, 
    
    \begin{equation*}
        \Gamma(y)\ge L_2(y)\ge L_2(1)=\gamma>1/2
    \end{equation*}
    which proves $\Gamma(y)\ge 1/2$, hence $\Gamma(x)\ge 1/(2x)$ for $x\in(0,1)$.
\end{proof}

\subsection{Computation of double derivative in \Cref{lem:ItrapIestboundSR}}

We restate the computation as a lemma.

\begin{lemma}\label{lem:SRfdoublederivbound}
    Let $f(t)=\espoly(t)re^{-rt}$ for some $r>0$. Then

    \begin{equation*}
        |f''(t)|\le |\F{0}|r(r^2+(r+\maxfreq)^2)
    \end{equation*}
\end{lemma}

\begin{proof}
    The double derivative of $f$ can be computed as

    \begin{align*}
        f''(t)&=re^{-rt}[\espoly''(t)-2r\espoly'(t)+r^2\espoly(t)]\\
        &=re^{-rt}\sum_{x\in\F{0}}[-\f{x}^2e^{-\f{x}t}-2r\f{x}e^{-\f{x}t}+r^2(1-e^{-\f{x}t})]\\
        &=re^{-rt}\sum_{x\in\F{0}}[r^2-(\f{x}+r)^2e^{-\f{x}t}]
    \end{align*}
    where we used the fact that $\espoly(t)=\sum_{x\in\F{0}}(1-e^{-\f{x}t}), \espoly'(t)=\sum_{x\in\F{0}}\f{x}e^{-\f{x}t}$ and $\espoly''(t)=\sum_{x\in\F{0}}-\f{x}^2e^{-\f{x}t}$.

    Now using the triangle inequality and the fact that $0<e^{-x}\le 1$, the absolute value of the double derivative can be upper bounded as

    \begin{equation*}
        |f''(t)|\le r\sum_{x\in\F{0}}[r^2+(\f{x}+r)^2]\le r\sum_{x\in\F{0}}[r^2+(\maxfreq+r)^2]=|\F{0}|r(r^2+(r+\maxfreq)^2)
    \end{equation*}
\end{proof}

\subsection{Proof of \Cref{lem:SLFAfderivativebound}}\label{lem:SLFAfderivativeboundappendix}

We restate the lemma here for convenience.

\begin{lemma}
    Let $f(t)=\espoly(t)\frac{e^{-t}}{t}$. Then $|f''(t)|\le |\F{0}|\tau(1+\tau)^2$ for $t\in [0,T]$ for any $T>0$.
\end{lemma}

\begin{proof}
    For a fixed $t\ge 0$ and $\f{x}\ge 1$, let

    \begin{equation*}
        h_{\f{x}}(t)=\frac{(1-e^{-\f{x}t})e^{-t}}{t}=\frac{e^{-t}-e^{-(\f{x}+1)t}}{t}
    \end{equation*}
    Let $s=1+\f{x}$. Then we have

    \begin{equation*}
        h''_{\f{x}}(t)=\frac{e^{-t}(t^2+2t+2)-e^{-st}(s^2t^2+2st+2)}{t^3}
    \end{equation*}
    Let $H(x)=(t^2x^2+2tx+2)e^{-tx}$. Then $h''_{\f{x}}(t)=\frac{H(1)-H(s)}{t^3}$. Since for any fixed $t>0$, $H(x)$ is smooth in $x$, so by the mean value theorem there exists $c\in (1,s)$ such that

    \begin{equation*}
        H(1)-H(s)=H'(x)(1-s)=-\f{x}H'(c)
    \end{equation*}
    Now $H'(x)=-t^3x^2e^{-xt}$ and thus for some $c\in (1,1+\f{x})$, $H(1)-H(s)=\f{x}t^3c^2e^{-ct}$ implying that $h_{\f{x}}''(t)=\f{x}c^2e^{-ct}$ for some $c\in (1,1+\f{x})$. Also, $\lim_{t\downarrow 0}h_{\f{x}}''(t)=\f{x}+\f{x}^2+\frac13\f{x}^3$, so we extend $h_{\f{x}}(t)$ to $0$ right-continuously with this value.

    From the above, $|h_{\f{x}}''(t)|=\f{x}c^2e^{-ct}$ for $t>0$. Since $c\in (1,1+\f{x}), e^{-ct}\le 1$ and $\f{x}\le \tau$, it follows that for any $t>0$

    \begin{equation*}
        |h_{\f{x}}''(t)|\le \f{x}(1+\f{x})^2\le \tau(1+\tau)^2
    \end{equation*}
    At $t=0$, we have 
    
    \begin{equation*}
        |h_{\f{x}}''(0)|=\f{x}+\f{x}^2+\frac13\f{x}^3\le \f{x}+2\f{x}^2+\f{x}^3=\f{x}(1+\f{x})^2\le \tau(1+\tau)^2
    \end{equation*}
    Thus $|h_{\f{x}}''(t)|\le \tau(1+\tau)^2$ for all $t\in [0,T], T>0$. Since $f(t)=\espoly(t)\frac{e^{-t}}{t}=\sum_{x\in\F{0}}h_{\f{x}}(t)$, we have

    \begin{equation*}
        |f''(t)|\le \sum_{x\in\F{0}}|h''_{\f{x}}(t)|\le |\F{0}|\tau(1+\tau)^2
    \end{equation*}
\end{proof}
\section{$\F{k}$ for $k>1$}

In this section, we show how the sampling algorithm of AMS \cite{alon1996space} can be modified to work for Delphic sets to estimate $\F{k}$ when $k>1$, and then show that our analytic approach also gives us an algorithm for the same but with a little worse space bounds when $k\in\mathbb N$.

\subsection{$\F{k}$ when $k>1$ using AMS but for Delphic}\label{sec:Fkamsanalytic}

We modify the AMS algorithm \cite{alon1996space} for $\F{k}$. Keep two variables $X$ and $C=0$. Let the stream be $\langle S_1, S_2,\dots,S_n\rangle$. Let $M_t=\sum_{j=1}^{t}|S_t|$. In the $t+1$-th update $S_{t+1}$, do the following. Flip a coin with probability $|S_{t+1}|/M_{t+1}$. If it lands head, sample an element uniformly at random from $S_{t+1}$, replace $X$ with that element and set $C=1$. If it lands tails, then check if $X\in S_{t+1}$; if yes then increase $C$ by $1$, if no then get the next update. At the end output $\hat X=m(C^k-(C-1)^k)$, where $m=\sum_{x\in\univ}\f{x}$.

\paragraph{}We show that $\mathbb E[\hat X]=\F{k}$. Let us condition on the sampled element $X=x$. Given $X=x$, the tail count $C$ is uniform in $\{1,2,\dots,\f{x}\}$. Then

\begin{equation*}
    \mathbb E[\hat X\mid X=x]=\frac1{\f{x}}\sum_{c=1}^{\f{x}}m(c^{k}-(c-1)^k)=m\frac{\f{x}^k}{\f{x}}=m\f{x}^{k-1}
\end{equation*}
Thus

\begin{equation*}
    \mathbb E[\hat X]=\sum_{x\in\univ}\mathbb E[\hat X\mid X=x]\mathbb P(X=x)=\sum_{x\in\univ}m\f{x}^{k-1}\frac{\f{x}}{m}=\sum_{x\in\univ}\f{x}^k=\F{k}
\end{equation*}
Similarly, for the variance, we have that

\begin{align*}
    \mathbb E[\hat X^2]&=m\sum_{x\in\univ}\sum_{c=1}^{\f{x}}(c^k-(c-1)^k)^2\\&\le m\sum_{x\in\univ}\sum_{c=1}^{\f{x}}kc^{k-1}(c^k-(c-1)^k)\\&\le m\sum_{x\in\univ}k\f{x}^{2k-1}\\&\le km\sum_{x\in\univ}\maxfreq^{k-1}\f{x}^k\\&=km\maxfreq^{k-1}\F{k}
\end{align*}
Now since $k\ge 1$, we have $m=\F{1}\le \F{k}$ and thus $\mathbb E[\hat X^2]\le k\maxfreq^{k-1}\F{k}^2$. Thus the variance is

\begin{equation*}
    \mathrm{Var}(\hat X)=\mathbb E[\hat X^2]-(\mathbb E[\hat X])^2\le \mathbb E[X^2]\le k\maxfreq^{k-1}\F{k}^2
\end{equation*}
Hence using Chebyschev, we get

\begin{equation*}
    \mathbb P(|\hat X-\F{k}|>\error\F{k})\le\frac{k\maxfreq^{k-1}\F{k}^2}{\error^2\F{k}^2}=\frac{k\maxfreq^{k-1}}{\error^2}
\end{equation*}
So by median of means, we can get an $(\error,\conf)$-estimate with $\tilde O_k(\tau^{k-1}\error^{-2}\log(1/\conf))$ space and update time.



\subsection{$\F{k}$ when $k\in\mathbb N$ using analytic techniques}\label{sec:FkN}

In this part, we show how to compute $\F{k}$ when $k\in \mathbb N$. The driving observations for our algorithm are the following two lemmas. To emphasize that we are doing integral moments, we write it as $\F{n}$ with $n\in\mathbb N$. But first, we need a small preliminary.

\subsubsection{Evaluating $(-1)^{n-1}\espoly^{(n)}$ at $0$}\label{subsec:appendixFkintegeranalytic}

 Here we mention the scheme used to estimate $(-1)^{n-1}\espoly^{(n)}$ at any point $t\in (0,\infty)$. The term $(-1)^{n-1}$ is there in front to make the term positive. We will use a forward difference technique to estimate this since $e^{-t}$ runs from $0$ to $\infty$.

The \textit{forward finite difference} scheme used to approximate $(-1)^{n-1}\espoly^{(n)}(0)$ is given by the fact that

\begin{equation*}
    (-1)^{n-1}\espoly^{(n)}(A)=(-1)^{n-1}\left[\frac1{h^n}\left(\sum_{i=0}^n(-1)^{n-i}\binom{n}{i}\espoly(A+ih)\right)+R^+_n\right]
\end{equation*}
with the error as

\begin{equation*}
    R^+_n=\frac{-h}{(n+1)!}\sum_{j=0}^n(-1)^{n-j}\binom{n}{j}j^{n+1}\espoly^{(n+1)}(\xi_j)
\end{equation*}
for some $\xi_j\in(A,A+jh)$ for all $j$\footnote{This follows from Taylor's theorem in \cite{atkinson2008introduction}}.
Putting in $A=0$ gives

\begin{equation*}
    (-1)^{n-1}\espoly^{(n)}(0)=(-1)^{n-1}\left[\frac1{h^n}\left(\sum_{i=0}^n(-1)^{n-i}\binom{n}{i}\espoly(ih)\right)+R^+_n\right]
\end{equation*}
with

\begin{equation*}
    R^+_n=\frac{-h}{(n+1)!}\sum_{j=0}^n(-1)^{n-j}\binom{n}{j}j^{n+1}\espoly^{(n+1)}(\xi_j)
\end{equation*}
for some $\xi_j\in(0,jh)$ for all $j$.

\subsection{$\F{k}$ estimation for $k\in\mathbb N$}

We provide a little intuition for the identity we will be using.
\paragraph{\textsc{Intuition:}} Let us go back to our jar of marbles. Let us also think of our $t$ as probe scales, trying to probe \textit{how many} marbles of a particular type there are by thinning the jar with probability $e^{-t}$. For a very small probe, the probability probability that the $j$-th marble of type $x$ survives is $1-e^{-t}\sim t$. Then what is the probability that at least one of type $x$ survives? Standard inclusion-exclusion tells us that it is

\begin{equation*}
    P\sim\binom{\f{x}}{1}t-\binom{\f{x}}{2}t^2+\binom{\f{x}}{3}t^3-\cdots
\end{equation*}
Note that each of these terms are appromately proportional to $\f{x}^n$ respectively. Taking the expectation now gives us that the terms are approximately $\F{n}$ respectively. The left hand side is $\espoly$ and the derivatives of $\espoly$ exactly peel off these terms from the \textit{power series expansion} of $\espoly$.

Another intuition that helps is to think of each derivative operation as a labelled (infinitessimal) probe; $n$ differentiations are $n$ labelled probes. For a fixed marble type $x$ there are $\f{x}$ many distinguishable occurences; each probe independently chooses an occurence to act on. So there are $\f{x}^n$ possible ordered choices for the probes, and summing over them gives us the moments. This is a repackaging of the above idea, but is more combinatorial to think in terms of.

We now make this formal.

\begin{lemma}\label{lem:espolynthderivative}
    The $n$-th derivative of $\espoly$ is given by

    \begin{equation*}
        \espoly^{(n)}(t)=(-1)^{n-1}\sum_{x\in\F{0}}\f{x}^ne^{-t\f{x}}
    \end{equation*}
\end{lemma}

\begin{proof}
    The proof is via induction. For $n=1$, we have \begin{equation}
        \espoly'(t)=\ddt{t}\sum_{x\in\F{0}}(1-e^{-t\f{x}})=\sum_{x\in\F{0}}\ddt{t}(1-e^{-t\f{x}})=\sum_{x\in\F{0}}\f{x}e^{-t\f{x}}=(-1)^{1-1}\sum_{x\in\F{0}}\f{x}^1e^{-t\f{x}}
    \end{equation}
    and hence the base case holds. Assume it is true for $n$. Then for $n+1$, we have

    \begin{align*}
        \espoly^{(n)}(t)&=\ddt{t}\espoly^{(n-1)}(t)=\ddt{t}(-1)^{n-1}\sum_{x\in\F{0}}\f{x}^ne^{-t\f{x}}=(-1)^{n-1}\sum_{x\in\F{0}}\f{x}^n\ddt{t}e^{-t\f{x}}\\
        &=(-1)^{n-1}\sum_{x\in\F{0}}\f{x}^n(-\f{x})e^{-t\f{x}}=(-1)^{n}\sum_{x\in\F{0}}\f{x}^{n+1}e^{-t\f{x}}=(-1)^{n+1-1}\sum_{x\in\F{0}}\f{x}^{n+1}e^{-t\f{x}}
    \end{align*}
    and thus by induction, the lemma holds.
\end{proof}

\begin{lemma}\label{lem:espolyderivat0}
    The following identity holds,

    \begin{equation*}
        \F{n}=(-1)^{n-1}\espoly^{(n)}(0)
    \end{equation*}
\end{lemma}

\begin{proof}
    We have, dy definition of $\espoly$,

    \begin{align*}
        \espoly^{(n)}(0)=(-1)^{n-1}\sum_{x\in\F{0}}\f{x}^ne^{-0\f{x}}=(-1)^{n-1}\sum_{x\in\F{0}}\f{x}^n=(-1)^{n-1}\F{n}
    \end{align*}
\end{proof}

We now provide an algorithm to estimate $\F{n}$ called $\estimateN$. \Cref{thm:AlgorithmEstimateNcorrectnessandspace} then proves the correctness of the algorithm along with the space and update time complexity.

    \SetAlgoNoLine%
    \begin{algorithm}[htb]
    \DontPrintSemicolon%
    \caption{$\estimateN(\stream, n, \error, \conf)$}
    \label{alg:estimateN}
    \SetKwInOut{Input}{Input}
    \SetKwInOut{Output}{Output}
    
    \Input{stream $\stream$, moment order $n\ge1$, overall error $\error$, overall failure probability $\conf$\;}
    \Output{$(\error,\conf)$-approximation of the $n$-th frequency moment $\F{n}$ of $\stream$\;}
    
    \BlankLine
    $h\gets\dfrac{\error}{\maxfreq2^{n+1}ne^{n-1}}$\;
    $\varepsilon_{\mathrm{eval}}\gets\dfrac{\error^{n}}{\maxfreq^{n-1}2^{n^2+n-2}n^ne^{(n-1)^2}}$\;
    $\delta_{\mathrm{per}} \gets \dfrac{\delta}{n+1}$\;
    \BlankLine
    \For{$i \gets 0$ \KwTo $n$}{
      $t_i \gets i\cdot h$\;
      $\alpha_i \gets \exp(-t_i)$\;
      $p_{i h} \gets\ \evaluate(\stream,\ \alpha_i,\ \varepsilon_{\mathrm{eval}},\ \delta_{\mathrm{per}})$\;
    }
    \BlankLine
    \Return
    $\displaystyle 
    (-1)^{\,n-1}\,\frac{1}{h^{n}}\sum_{i=0}^{n}(-1)^{\,n-i}\binom{n}{i}\,p_{i h}$\;
    \end{algorithm}

\begin{theorem}\label{thm:AlgorithmEstimateNcorrectnessandspace}
    The algorithm $\estimateN$ outputs an $(\varepsilon,\delta)$-approximation of $\F{n}$ with $n\in\mathbb N$. $\estimateN$ takes $\tilde O\left(\frac{(\log|\univ|+\log(1/\delta))\maxfreq^{4n-3}2^{4(n^2+n)+n+2}n^{4n+1}e^{4(n-1)^2+n-1}}{\varepsilon^{4n}}\log^4|\univ|\right)$ space and update time.
\end{theorem}

To prove this theorem, we will show that the output of $\estimateN$ is an $(\error,\conf)$-approximation of $(-1)^{n-1}\espoly^{(n)}(0)$. Then we are done by \Cref{lem:espolyderivat0}.

\begin{lemma}
        Let $Z_n$ be the output of \Cref{alg:estimateN}. Then $Z_n$ is an $(\varepsilon,\delta)$-approximation of $(-1)^{n-1}\espoly^{(n)}(0)$.
\end{lemma}

\begin{proof}
    Let

    \begin{equation*}
        \Delta_h^n\espoly(x)=\sum_{i=0}^n(-1)^{n-i}\binom{n}{k}\espoly(x+ih)
    \end{equation*}
    This implies that $(-1)^{n-1}\espoly^{(n)}(0)=\frac{\Delta_h^n\espoly(0)}{h^n}+(-1)^{n-1}R_n^+$ and hence

    \begin{align*}
        \bigg|(-1)^{n-1}\espoly^{(n)}(A)-(-1)^{n-1}\frac{\Delta_h^n\espoly(A)}{h^n}\bigg|
        &=|R_n^+|\\
        &\le\frac{h|\espoly^{(n+1)}(A)|}{(n+1)!}\sum_{j=0}^n\binom{n}{j}j^{n+1}\\
        &\le\frac{h|\espoly^{(n+1)}(A)|}{(n+1)!}2^nn^{n+1}
    \end{align*}
    Since $|\espoly^{(n+1)}(0)|\le\maxfreq|\espoly^{(n)}(0)|$, we have

    \begin{equation*}
        \bigg|(-1)^{n-1}\espoly^{(n)}(0)-(-1)^{n-1}\frac{\Delta_h^n\espoly(0)}{h^n}\bigg|\le\frac{h\maxfreq|\espoly^{(n)}(0)|}{(n+1)!}2^nn^{n+1}
    \end{equation*}
    This is a bound on the discretization relative error.
    
    Now $p_{ih}=\espoly(ih)(1+\theta_i)$ with $|\theta_i|\le\varepsilon_{\mathrm{eval}}$ for each $i$, and hence

    \begin{equation*}
        Z_n=(-1)^{n-1}\frac1{h^n}\Delta_h^n\espoly(0)+\frac1{h^n}E_{\mathrm{noise}},\qquad E_{\mathrm{noise}}=\sum_{i=0}^n(-1)^{n-i}\binom{n}{i}\espoly(ih)\theta_i
    \end{equation*}
    and thus by the triangle inequality

    \begin{equation*}
        |E_{\mathrm{noise}}|\le\varepsilon_{\mathrm{eval}}\sum_{i=0}^n\binom{n}{i}\espoly(ih)
    \end{equation*}
    This is a bound on the noise-induced relative error. Putting the two error sources together, we get

    \begin{equation*}
        \frac{|Z_n-(-1)^{n-1}\espoly^{(n)}(0)|}{|\espoly^{(n)}(0)|}\le \frac{h\maxfreq}{(n+1)!}2^nn^{n+1}+\varepsilon_{\mathrm{eval}}\frac{\sum_{i=0}^n\binom{n}{i}\espoly(ih)}{h^n|\espoly^{(n)}(0)|}
    \end{equation*}
    Now 
    
    \begin{equation*}
        \espoly(ih)=\sum_{x\in\F{0}}(1-e^{-\f{x}ih})\le \sum_{x\in\F{0}}ih\f{x}\le ihm
    \end{equation*}
    We also have from \Cref{lem:espolynthderivative}
    
    \begin{equation*}
        |\espoly^{(n)}(0)|=\sum_{x\in\F{0}}\f{x}^ne^{-\f{x}0}=\sum_{x\in\F{0}}\f{x}\f{x}^{n-1}\ge \sum_{x\in\F{0}}\f{x}=m
    \end{equation*}
    Plugging these in, we get

    \begin{equation*}
        \frac{|Z_n-(-1)^{n-1}\espoly^{(n)}(0)|}{|\espoly^{(n)}(0)|}\le \frac{h\maxfreq}{(n+1)!}2^nn^{n+1}+\varepsilon_{\mathrm{eval}}\frac{n2^{n-1}}{h^{n-1}}
    \end{equation*}
    So choosing 
    
    \begin{equation*}
        h=\frac{\error}{\maxfreq2^{n+1}ne^{n-1}},\qquad \varepsilon_{\mathrm{eval}}=\frac{\error^{n}}{\maxfreq^{n-1}2^{n^2+n-2}n^ne^{(n-1)^2}}
    \end{equation*}
    as in the algorithm, the total error comes out to at most $\varepsilon$ (here we have used the fact that $(n+1)!\ge n!\ge e(n/e)^n$). Since each of the $p_ih=\espoly(ih)(1+\theta_i)$ fails with probability at most $\conf/(n+1)$, by a union bound, the whole estimator outputs a value not in the correct range with probability at most $\conf$.





    There are $(n+1)$ such instances. The instance corresponding to $i=0$ returns $\espoly(0)=0$ exactly. For every $i\ge1$, we have
    $t_i=ih\ge h$, and hence by \Cref{cor:espolyestimatecorollary},
    \[ \frac{1}{1-e^{-t_i}}\le 1+\frac1h\le\frac{2^{n+2}n e^{n-1}\maxfreq}{\error}.\]
    The additional logarithmic dependence on this quantity is absorbed in the $\tilde O_n(\cdot)$ notation. Plugging in the values of
    $\error_{\mathrm{eval}}$ and $h$, the total space and update time is
    \[ \tilde O_n\left(\error^{-(4n+1)}\maxfreq^{4n-3}(\log|\univ|+\log(1/\conf))\log^4|\univ|\right). \]

\end{proof}
\section{Deeper reasons for \Cref{lem:Fkidentityintegral,lem:SRintegralidentity,lem:SLFAintegralidentity}}\label{sec:deeperreasonsappendix}

We start with reminding ourselves the relation between a Bernstein function and the inverse Laplace transform from \Cref{thm:BernsteininverseLaplace}; namely, if $\varphi$ is a Bernstein function with no drift and no killing term, then

\begin{equation}\label{eqn:reminderbernsteininverseLaplace}
    \varphi(t)=\int_0^\infty (1-e^{-ts})\frac1s\mathcal L^{-1}\left\{\varphi'(t)\right\}(s)\ \mathrm ds
\end{equation}
We can recover \Cref{lem:Fkidentityintegral} by computing this when $\varphi(t)=t^k$ with $k\in (0,1)$.

\begin{lemma}\label{lem:bernsteink}
    The function $\varphi(t)=t^k$ is a Bernstein function. It's L\'evy-Khintchine representation is given by

    \begin{equation*}
        t^k=\int_0^\infty (1-e^{-ts})s^{-k-1}\ \mathrm ds
    \end{equation*}
\end{lemma}

\begin{proof}
    We first prove that $\varphi$ is a Bernstein function. We have $-\varphi'(t)=-kt^{k-1}\le 0$. Now we need to show that $(-1)^n\varphi^{(n)}(t)\le 0$. The base case of $n=0$ is true. For the induction step, note that
    
    \begin{align*}
        (-1)^{n+1}\varphi^{(n+1)}(t)
        &=(-1)^{n+1}k(k-1)\cdots(k-n)t^{k-(n+1)}\\
        &=(-1)^nk(k-1)\cdots(k-(n-1))t^{k-n}(-1)(k-n)t^{-1}\\
        &=(-1)^n\varphi^{(n)}(t)(n-k)t^{-1}
    \end{align*}
    Now since $(n-k)t^{-1}>0$ and by the induction hypothesis we have $(-1)^n\varphi^{(n)}(t)\le 0$, we obtain $ (-1)^{n+1}\varphi^{(n+1)}\le 0$ showing that $\varphi$ is a Bernstein function.

    Now we show the killing and drift term of $\varphi$ is $0$. The proof is simply to note that $\lim_{t\to 0^+}\varphi(t)=0$ and thus $a=0$. Similarly, since $\lim_{t\to\infty}\varphi'(t)=\lim_{t\to\infty}kt^{k-1}=0$ because $k-1<0$, we have $b=0$.

    Finally, we need to compute the inverse Laplace transform of $\varphi'(t)$. Note that $\varphi'(t)=kt^{k-1}$. By linearity and standard rules of Laplace inversion
    
    \begin{equation*}
        \mathcal L^{-1}\{kt^{k-1}\}(s)=k\mathcal L^{-1}\{t^{k-1}\}(s)=k\mathcal L\{1/t^{1-k}\}(s)=k(s^{-k}/\Gamma(1-k))=\frac{ks^{-k}}{\Gamma(1-k)}
    \end{equation*}
    where in the first equality we used linearity and in the second equality we used the fact that $\mathcal L^{-1}\{1/t^a\}(s)=s^{a-1}/\Gamma(a)$. Plugging this into \Cref{eqn:reminderbernsteininverseLaplace}, we get

    \begin{equation*}
        t^k=\frac{k}{\Gamma(1-k)}\int_0^\infty (1-e^{-ts})s^{-k-1}\ \mathrm ds
    \end{equation*}
\end{proof}

Similarly, we can recover \Cref{lem:SRintegralidentity} by considering $\varphi(t)=\frac{t}{r+t}$.

\begin{lemma}
    The function $\varphi(t)=\frac{t}{r+t}$ is a Bernstein function. It's L\'evy-Khintchine representation is given by

    \begin{equation*}
        t^k=\int_0^\infty (1-e^{-ts})re^{-rs}\ \mathrm ds
    \end{equation*}
\end{lemma}

\begin{proof}
    We first prove that $\varphi$ is a Bernstein function. We start by noting that $\varphi(t)=\frac{t}{r+t}=1-r(r+t)^{-1}$. We need to show that $(-1)^n\varphi^{(n)}(t)\le 0$. The $n$-th derivative is easily seen to be $\varphi^{(n)}(t)=(-1)^{n-1}n!r(r+t)^{-n-1}$. Thus $(-1)^n\varphi^{(n)}(t)=-n!r(r+t)^{-n-1}\le 0$ showing that $\varphi$ is a Bernstein function.

    Now we show the killing and drift term of $\varphi$ is $0$. The proof is simply to note that $\lim_{t\to 0^+}\varphi(t)=0$ and thus $a=0$. Similarly since $\lim_{t\to\infty}\varphi'(t)=\lim_{t\to\infty}\frac{r}{(r+t)^2}=0$, we have $b=0$.

    Finally, we need to compute the inverse Laplace transform of $\varphi'(t)$. Note that $\varphi'(t)=\frac{r}{(r+t)^2}$. By linearity, frequency-shift property and standard rules of Laplace inversion
    
    \begin{equation*}
        \mathcal L^{-1}\left\{\frac{r}{(r+t)^2}\right\}(s)=r\mathcal L^{-1}\{(t+r)^{-2}\}(s)=re^{-rs}\mathcal L\{1/t^2\}(s)=re^{-rs}(s/\Gamma(2))=rse^{-rs}
    \end{equation*}
    where in the first equality we used linearity, in the second equality we used the frequency-shift property, and in the third equality we used the fact that $\mathcal L^{-1}\{1/t^a\}(s)=s^{a-1}/\Gamma(a)$. Plugging this into \Cref{eqn:reminderbernsteininverseLaplace}, we get

    \begin{equation*}
        t^k=\int_0^\infty (1-e^{-ts})re^{-rs}\ \mathrm ds
    \end{equation*}
\end{proof}

And finally, we can recover \Cref{lem:SLFAintegralidentity} by considering $\varphi(t)=\ln(1+t)$.

\begin{lemma}
    The function $\varphi(t)=\ln(1+t)$ is a Bernstein function. It's L\'evy-Khintchine representation is given by

    \begin{equation*}
        t^k=\int_0^\infty (1-e^{-ts})\frac{e^{-s}}{s}\ \mathrm ds
    \end{equation*}
\end{lemma}

\begin{proof}
    We first prove that $\varphi$ is a Bernstein function. We need to show that $(-1)^n\varphi^{(n)}(t)\le 0$. The $n$-th derivative is easily seen to be $\varphi^{(n)}(t)=(-1)^{n-1}(n-1)!(1+t)^{-n}$. Thus $(-1)^n\varphi^{(n)}(t)=-(n-1)!(1+t)^{-n}\le 0$ showing that $\varphi$ is a Bernstein function.

    Now we show the killing and drift term of $\varphi$ is $0$. The proof is simply to note that $\lim_{t\to 0^+}\varphi(t)=0$ and thus $a=0$. Similarly since $\lim_{t\to\infty}\varphi'(t)=\lim_{t\to\infty}\frac{1}{1+t}=0$, we have $b=0$.

    Finally, we need to compute the inverse Laplace transform of $\varphi'(t)=\frac1{1+t}$. Note that $\varphi'(t)=\frac{r}{(r+t)^2}$. This is a standard Laplace inverse formula $\mathcal L\left\{\frac1{t+a}\right\}(s)=e^{-as}$. This gives us
    
    \begin{equation*}
        \mathcal L^{-1}\left\{\frac1{1+t}\right\}(s)=e^{-s}
    \end{equation*}
    Plugging this into \Cref{eqn:reminderbernsteininverseLaplace}, we get

    \begin{equation*}
        t^k=\int_0^\infty (1-e^{-ts})re^{-rs}\ \mathrm ds
    \end{equation*}
\end{proof}

\end{document}